\documentclass[11pt,letterpaper]{article}

\usepackage[margin=1in]{geometry}        
\usepackage{setspace}                     
\usepackage[T1]{fontenc}                  
\usepackage[utf8]{inputenc}               
\usepackage{lmodern} 
\usepackage{graphicx}
\usepackage{natbib}
\usepackage{booktabs}
\usepackage{authblk}
\usepackage{enumerate}
\usepackage{enumitem}

\usepackage{amsmath,amssymb,amsthm,mathtools,bbm}

\numberwithin{equation}{section}

\theoremstyle{plain}
\newtheorem{theorem}{Theorem}[section]
\newtheorem{lemma}[theorem]{Lemma}
\newtheorem{proposition}[theorem]{Proposition}
\newtheorem{corollary}[theorem]{Corollary}

\theoremstyle{definition}
\newtheorem{definition}{Definition}
\newtheorem{assumption}{Assumption}
\newtheorem{remark}{Remark}[section]

\usepackage[ruled, linesnumbered]{algorithm2e}
\usepackage{xcolor}

\usepackage[
  colorlinks,
  linkcolor=blue,
  citecolor=blue,
  urlcolor=blue,
  breaklinks
]{hyperref}

\def\E{\mathbb{E}}
\def\P{\mathbb{P}}
\def\R{\mathbb{R}}
\def\N{\mathbb{N}}
\newcommand{\bx}{\mathbf{x}}
\newcommand{\bX}{\mathbf{X}}

\def\Hcal{\mathcal{H}}
\def\Ccal{\mathcal{C}}
\def\Pcal{\mathcal{P}}

\newcommand{\One}[1]{{\mathbbm{1}}\left\{{#1}\right\}}
\newcommand{\one}[1]{{\mathbbm{1}}_{{#1}}}
\def\conch{{\textsc{CONCH}}}
\DeclareMathOperator*{\argmax}{argmax}

  \newcommand\independent{\protect\mathpalette{\protect\independenT}{\perp}}
\def\independenT#1#2{\mathrel{\rlap{$#1#2$}\mkern2mu{#1#2}}}

\usepackage[nameinlink,noabbrev]{cleveref}
\crefname{assumption}{assumption}{assumptions}

\newcommand{\papertitle}{Conformal changepoint localization}
\newcommand{\paperauthorA}{Rohan Hore}
\newcommand{\paperauthorB}{Aaditya Ramdas}
\newcommand{\affilOne}{Department of Statistics and Data Science, Carnegie Mellon University}

\newcommand{\corrEmail}{rhore@andrew.cmu.edu}
\newcommand{\keywordslist}{Offline changepoint localization, conformal inference, conformal p-values, Confidence set, Distribution-free inference.}

\title{\papertitle}

\author[1]{\paperauthorA\thanks{Corresponding author: \corrEmail}}
\author[1]{\paperauthorB}
\affil[1]{\affilOne}

\date{\today}

\newcommand{\paperabstract}[1]{%
  \begin{abstract}
    #1
  \end{abstract}
}
\newcommand{\paperkeywords}[1]{%
  \vspace{0.5em}
  \noindent\textbf{Keywords: } #1
}

\usepackage[nameinlink,noabbrev]{cleveref}

\begin{document}

\maketitle

\paperabstract{%
We study the problem of offline changepoint localization in a distribution-free setting. One observes a vector of data with a single changepoint, assuming that the data before and after the changepoint are i.i.d.\ (or more generally exchangeable) from arbitrary and unknown distributions. The goal is to produce a finite-sample confidence set for the index at which the change occurs without making any other assumptions. Existing methods often rely on parametric assumptions, tail conditions, or asymptotic approximations, or only produce point estimates. In contrast, our distribution-free algorithm, CONformal CHangepoint localization (\textsc{CONCH}), only leverages exchangeability arguments to construct confidence sets with finite sample coverage.  
By proving a \emph{conformal Neyman--Pearson lemma}, we derive principled score functions that yield informative (small) sets. Moreover, with such score functions, the normalized length of the confidence set shrinks to zero under weak assumptions. We also establish a universality result showing that any distribution-free changepoint localization method must be an instance of \conch{}. Experiments suggest that \textsc{CONCH} delivers precise confidence sets even in challenging settings involving images or text.
}

\paperkeywords{\keywordslist}

\section{Introduction} \label{sec:introduction}
We study the problem of offline changepoint localization. Given an ordered sequence of $\mathcal{X}$-valued observations $X_1,\ldots,X_n$ with an unknown changepoint location $\xi\in\{1,\ldots,n-1\}$, the goal is not merely to provide a point estimate $\hat{\xi}$, but also to quantify the uncertainty associated with that estimate and construct a confidence set for the true changepoint.

Specifically, suppose that $X_1,\ldots,X_\xi$ are drawn independently from a distribution $P_0$, while $X_{\xi+1},\ldots,X_n$ are drawn independently from a different distribution $P_1\neq P_0$. Our objective is to perform changepoint localization in a distribution-free manner; that is, given a level $\alpha\in(0,1)$, we seek a confidence set $C$ for $\xi$ satisfying
\begin{equation}\label{eq:DF-guarantee}
\inf_{P_0,\xi,P_1\neq P_0} \P(\xi\in C)\geq 1-\alpha.
\end{equation}
In words, the coverage guarantee should hold uniformly over all changepoint locations $\xi$ and all pairs of distinct pre-change and post-change distributions $P_0$ and $P_1$. The term ``distribution-free'' refers precisely to this unrestricted infimum, in contrast to approaches that assume specific distributional families such as Gaussian, bounded, or other parametric models.

Importantly, no structural assumptions are imposed on the observation space $\mathcal X$. Consequently, the same framework applies not only to Euclidean data, but also to more complex data modalities such as images, text, and other structured objects.

In domains such as operations engineering, econometrics, and biostatistics, the ability to retrospectively (and confidently) assess the time of distributional change can be critical. Consider, for instance, a manufacturing context: quality measurements of a component may remain stable until a machine begins to malfunction, after which the measurements exhibit a systematic shift. Once the production batch has concluded, it becomes essential to confidently determine when this shift first arose in order to diagnose the source of the malfunction and implement corrective measures. 

\subsection{Related work}\label{sec:review}

Given its wide practical relevance, offline changepoint analysis has been extensively studied; see \citet{truong2020selective,duggins2010parametric} for surveys. Some of this literature focuses on changepoint detection, where the goal is to test the null hypothesis that there is no changepoint against the alternative that at least one changepoint is present, while some papers additionally try to provide a point estimate for the change location that is asymptotically consistent. In contrast, our problem of changepoint localization additionally seeks to quantify uncertainty around an estimated changepoint location and construct a confidence set for the true changepoint. As we discuss later, this amounts to testing the null hypothesis that a candidate location $t$ is the true changepoint, separately for each $t$. While related, these are different statistical objectives, calling for different techniques.
Existing approaches to changepoint analysis can be broadly categorized as follows.

\medskip\textbf{(a) Changepoint detection methods.}
The changepoint detection problem is classical and well studied. Beginning with CUSUM \citep{page1955test}, a large literature has developed around modern detection procedures, including kernel changepoint procedures (KCP) \citep{arlot2019kernel}, ChangeForest \citep{londschien2023random}, and graph-based methods \citep{chu2019asymptotic}. Conformal martingale methods \citep{vovk2003testing,volkhonskiy2017inductive,vovk2021testing,vovk2021retrain,nouretdinov2021conformal,shin2022detectors} provide powerful tools for online changepoint detection, but do not yield confidence sets for localization.

\medskip\textbf{(b) Parametric localization methods.}
Several localization procedures rely on strong model assumptions. Likelihood-based approaches assume specific parametric families (e.g., Gaussian mean-shift or linear regression models) and derive point estimators and confidence intervals through asymptotic likelihood theory \citep{kim1989likelihood,quandt1958estimation,gurevich2006guaranteed}. More recent post-detection localization procedures \citep{saha2025post} likewise require restrictive assumptions, such as known and non-overlapping pre-change and post-change model classes.

\medskip\textbf{(c) Nonparametric localization methods.}
Several nonparametric approaches provide localization guarantees only asymptotically. These include SMUCE \citep{frick2014multiscale}, regression-based methods \citep{xu2024change}, and Gaussian mean-shift confidence intervals \citep{fotopoulos2010exact}, among others \citep{bhattacharyya1968nonparametric,zou2007empirical}. The procedure of \citet{verzelen2023optimal} achieves strong theoretical optimality guarantees, but depends on non-computable constants that limit its practical applicability.

\medskip\textbf{(d) Resampling-based localization methods.}
Bootstrap-based approaches \citep{cho2022bootstrap} provide uncertainty quantification for changepoint estimates but generally lack finite-sample validity guarantees and can be computationally intensive. Rank-based nonparametric procedures \citep{pettitt1979non,ross2012two} are distribution-free for detection, but do not provide confidence sets for localization. Multi-changepoint methods \citep{anastasiou2022detecting,truong2020selective} typically follow isolate-detect strategies and primarily return point estimates.

Overall, existing localization methods either rely on restrictive model assumptions, provide only asymptotic guarantees, or focus primarily on point estimation rather than finite-sample valid uncertainty quantification.

\medskip\noindent\textbf{Conformal approaches to localization.}
Recently, \citet{dandapanthula2025conformal} introduced MCP localization, the first fully distribution-free approach to changepoint localization (in the sense of~\eqref{eq:DF-guarantee}). Their method constructs confidence sets using a matrix of conformal $p$-values and enjoys finite-sample validity guarantees. However, as we demonstrate later in Appendix~\ref{app:gaussian_all_pvalue_comaprison}, MCP often produces confidence sets that are wider than necessary, motivating the search for sharper distribution-free localization procedures.

More broadly, our work is inspired by the same conformal inference literature. Originally introduced by \citet{vovk1999machine,shafer2008tutorial} for distribution-free predictive inference, conformal methods have since been extended to a wide range of statistical tasks, including outlier detection \citep{bates2023testing}, post-prediction screening \citep{jin2023selection}, and conditional two-sample testing \citep{wu2024conditional}, among others.

In this work, we build on these ideas to develop a simple and principled framework for changepoint localization that is distribution-free, finite-sample valid, and capable of producing informative confidence sets. We also prove that every distribution-free confidence set must be an instance of CONCH. We formalize the objective of distribution-free changepoint localization in Section~\ref{sec:dist_free_chpt_loc}.

\subsection{Summary of our contributions} 
\begin{itemize}
    \item We introduce \conch{} (CONformal CHangepoint localization) in \Cref{alg:conch}, a novel framework that, given any changepoint plausibility score $S$ and a confidence level $1-\alpha$, produces a finite-sample, distribution-free confidence set for the true changepoint without making any restrictive assumptions on the pre- and post-change distributions.
    
    \item While our framework is valid for any choice of score function, offering great flexibility to the user, its statistical performance can be substantially improved by employing scores tailored to the problem at hand. We derive an expression of optimal score function based on a novel ``Conformal Neyman–Pearson'' lemma (Lemma~\ref{lem:conch_NP_lemma}), which may be of independent interest.

    \item While the optimal score requires oracle knowledge, we propose practical ``near-optimal'' alternatives learnable from data that yield narrow confidence sets. We further show that, under mild regularity conditions, the normalized length of the confidence
    set converges to zero as the sample size grows (Theorem~\ref{thm:conch-consistency}).

   \item Next, we show that \conch{} has a universality property (Theorem~\ref{thm:universality_changepoint}): any distribution-free confidence set for the changepoint is an instance of our framework. This further enables us to give a simple calibration technique that turns any heuristic confidence set or any point estimate for the changepoint location into a distribution-free valid confidence set.

   \item We illustrate the practical effectiveness of \conch{} across a range of synthetic and real-world datasets, covering both image and text domains. In all cases, the resulting \conch{} confidence sets concentrate tightly around the true changepoint. Notably, our framework can wrap around any black-box classifier trained to distinguish pre- and post-change samples, yielding informative confidence sets even under subtle distributional shifts.
\end{itemize}

\section{Distribution-free changepoint localization}\label{sec:dist_free_chpt_loc}
We formally describe the problem of distribution-free offline changepoint localization that was alluded to in~\eqref{eq:DF-guarantee}. Throughout, $\N$ denotes the natural numbers and, for $K\in\N$, we write $[K]:=\{1,\dots,K\}$. For any set $S$, let $\mathcal{M}(S)$ denote the collection of probability measures on $S$, and let $2^{S}$ denote its power set. We use $\overset{d}{=}$ to denote equality in distribution.
 
With this notation in place, consider an ordered sequence of $\mathcal{X}$-valued random variables $\bX = (X_1,\ldots,X_n)$ for some $n \in \N$. Likewise, we use $\bx = (x_1,\ldots,x_n)$ to denote a generic element of $\mathcal{X}^n$. We assume that there exists an unknown changepoint $\xi \in [n-1]$ such that  
\[
(X_1,\ldots,X_{\xi}) \sim \Pcal_{0,\xi}, 
\qquad 
(X_{\xi+1},\ldots,X_n)\sim \Pcal_{1,\xi},
\]  
where $\Pcal_{0,\xi} \in \mathcal{M}(\mathcal{X}^\xi)$ and $\Pcal_{1,\xi} \in \mathcal{M}(\mathcal{X}^{n-\xi})$ denote the pre-change and post-change distributions, respectively. Throughout, we assume that $\xi$ is a genuine distributional change, in the sense that the one-dimensional marginals induced by $\Pcal_{0,\xi}$ and $\Pcal_{1,\xi}$ are distinct.
We write the corresponding joint distribution as $\Pcal = \Pcal_{0,\xi} \times \Pcal_{1,\xi}$. In line with the distribution-free perspective, we impose no structural assumptions on $\Pcal_{0,\xi}$ or $\Pcal_{1,\xi}$ beyond the following.  

\begin{assumption}\label{assn:exchangeability}
    $\Pcal_{0,\xi}$ and $\Pcal_{1,\xi}$ are exchangeable. Specifically, for any permutations $\pi_L:[\xi]\to[\xi]$ and $\pi_R:[n]\setminus[\xi]\to[n]\setminus[\xi]$, it holds that the pre-change and post-change segments are independent, i.e., $\Pcal_{0,\xi}\independent \Pcal_{1,\xi}$ and that
   \[
   (X_1,\dots,X_\xi) \overset{d}{=} (X_{\pi_L(1)},\dots,X_{\pi_L(\xi)}), 
   \quad 
   (X_{\xi+1},\dots,X_n) \overset{d}{=} (X_{\pi_R(\xi+1)},\dots,X_{\pi_R(n)}).
   \]
\end{assumption}
In words, \Cref{assn:exchangeability} requires that the distribution of $\bX$ is invariant under arbitrary permutations of the entries to the left of $\xi$ and, independently, under permutations of those to its right. A canonical example, mentioned in the introduction, is the i.i.d. changepoint model: the \emph{pre-change} observations $(X_1,\dots,X_\xi)$ are i.i.d.\ from some $P_0$, and independently, the \emph{post-change} observations $(X_{\xi+1},\dots,X_n)$ are i.i.d.\ from some $P_1$.  

For any $t \in [n-1]$, let $\Hcal_{0,t}$ denote the hypothesis that $t$ is the true changepoint and that the distributions $\Pcal_{0,t}$ and $\Pcal_{1,t}$ satisfy \Cref{assn:exchangeability}. 
We write $\P_t$ and $\E_t$ to denote probability and expectation, respectively, under any distribution that satisfies the null hypothesis $\Hcal_{0,t}$. 
We can now formally define what it means to construct a distribution-free confidence set for the changepoint.

\begin{definition}\label{eq:DF-conf_defn}
Fix $\alpha\in (0,1)$. A mapping $\mathcal{C}_{1-\alpha}:\mathcal{X}^n \to 2^{[n-1]}$ is called a \emph{distribution-free confidence set for changepoint} at level $1-\alpha$ if $\P_{\xi}\!\left(\,\xi \in \mathcal{C}_{1-\alpha}(\bX)\,\right) \;\geq\; 1-\alpha$ whenever \Cref{assn:exchangeability} holds.
\end{definition}

\Cref{assn:exchangeability} is considerably weaker than the working assumptions underlying most existing changepoint localization methods reviewed in Section~\ref{sec:review}. Prior approaches typically rely on strong parametric models or asymptotic approximations, in contrast to the minimal nature of our assumption. While some recent methods \citep{dandapanthula2025conformal} offer distribution-free guarantees under similarly mild conditions, they generally yield more diffuse confidence sets. Our approach instead enables sharper localization while retaining finite-sample validity, making it a significant contribution in this direction. The next section formally introduces our method.

\section{Methodology}
\label{sec:conch}

\subsection{Conformal changepoint localization} \label{sec:conch_description}
This section develops a conformal framework for localizing the true changepoint.  We adapt conformal $p$-values \citep{vovk1999machine,vovk2005algorithmic} to the changepoint localization problem in an efficient manner, yielding a general framework to construct confidence sets that satisfy \eqref{eq:DF-conf_defn}. 

\IncMargin{1.2em} 
\begin{algorithm}[t]
    \caption{\conch{}: conformal changepoint localization algorithm}
    \label{alg:conch}
    \KwIn{$(X_t)_{t=1}^n$ (data), $1-\alpha$ (target coverage), $S:\mathcal{X}^n\to \R^{\,n-1}$ (CPP score)}
    \KwOut{$\mathcal{C}^{\,\mathrm{CONCH}}_{1-\alpha}$ (\conch{} confidence set at level $1-\alpha$)}

    \For{$t \in [n-1]$}{
        $\Pi_t \gets \{\pi \in \mathcal{S}_n : \textnormal{for all } i\le t,\ \pi(i)\le t\ \textnormal{ and for all } i>t,\ \pi(i)>t\}$\;
        \ForEach{$\pi \in \Pi_t$}{Evaluate $S_t(\pi(\bX))$\;}
        $p_t \gets \frac{1}{|\Pi_t|}\sum_{\pi\in \Pi_t} \One{S_t(\pi(\bX))\leq S_t(\bX)}$\;
    }
    $\mathcal{C}^{\,\mathrm{CONCH}}_{1-\alpha} \gets \{\, t \in [n-1] : p_t > \alpha \,\}$\;
    \Return{$\mathcal{C}^{\,\mathrm{CONCH}}_{1-\alpha}$}
\end{algorithm}
\DecMargin{1.2em} 
We call our proposed algorithm the CONformal CHangepoint localization (CONCH) algorithm, formally given in \Cref{alg:conch}. Our framework relies on two key components:
\begin{itemize}
    \item \textbf{ChangePoint Plausibility (CPP) score:} We call any mapping $S:\mathcal{X}^n \to \R^{n-1}$ a changepoint plausibility score. Intuitively, for each candidate index $t\in [n-1]$, $S_t=(S(\cdots))_{t}$ assigns a score to quantify the chance that $t$ is indeed a changepoint; a larger $S_t$ indicates a stronger plausibility of $t$ being a changepoint.
    \item \textbf{Split-permutation group:} For any $t \in [n-1]$, we define the reduced set of permutations
\begin{equation}\label{eqn:permutation_set}
         \Pi_t := \Bigl\{ \pi \in \mathcal{S}_n : \pi(i) \leq t \;\; \text{for all } i \leq t, \;\; \pi(i) > t \;\; \text{for all } i > t \Bigr\}.
    \end{equation}
    Any $\pi\in \Pi_t$ permutes indices to the left and right of $t$, without mixing indices across $t$. 
\end{itemize}
Note that, if $t$ is indeed the true changepoint, elements of $\Pi_t$ preserve the pre- and post-change exchangeability. The validity of our framework crucially relies on this observation. More precisely, starting from any user-specified CPP score $S$, we define a conformal $p$-value $p_t$ for each index $t \in [n-1]$ by looking at the normalized rank of the true score $S_t(\bX)$ within the set of all permuted scores, $\{S_t(\pi(\bX)):\,\pi \in \Pi_t\}$, i.e.,
\begin{equation}\label{eq:pvalue_conch}
     p_t:=\frac{1}{|\Pi_t|}\sum_{\pi\in \Pi_t} \One{S_t(\pi(\bX))\leq S_t(\bX)}.
\end{equation} 
Intuitively, under $\Hcal_{0,t}$, every permutation $\pi \in \Pi_t$ is equally likely and therefore, $p_t$ is super-uniform under the null $\Hcal_{0,t}$, a result we formally establish in Theorem~\ref{thm:coverage-conch}. For brevity, the proof is deferred to Appendix~\ref{app:proof_of_coverage}.
Finally, the changepoint confidence set is then constructed by thresholding these $p$-values at a pre-specified level $\alpha$:
    \[
    \mathcal{C}^{\,\mathrm{CONCH}}_{1-\alpha}:=\{t\in [n-1]: p_{t}>\alpha\}.
    \]

\begin{theorem}\label{thm:coverage-conch}
    For each $t \in [n-1]$, $p_t$ in \eqref{eq:pvalue_conch} is a valid $p$-value under $\Hcal_{0,t}$. In particular, for any $\alpha \in (0,1)$,
    \(\P_{\xi}\left(p_{\xi}\leq \alpha\right)\leq \alpha\).
    Consequently, $\mathcal{C}^{\mathrm{CONCH}}_{1-\alpha}$ is a distribution-free confidence set for changepoint.
\end{theorem}
We highlight that the \conch{} algorithm and the above coverage guarantee do not impose any restriction on the choice of CPP score, thereby providing significant flexibility for users to design their own plausibility measure. In particular, the score function may depend non-trivially on the multiset $\{X_1,\ldots,X_n\}$. 
Before proceeding further, we make a few remarks on the \conch{} framework.

\begin{remark}[Exact validity]
    Theorem~\ref{thm:coverage-conch} ensures that the \conch{} $p$-value in~\eqref{eq:pvalue_conch} is super-uniform, which implies that the resulting confidence set $\mathcal{C}_{1-\alpha}^{\conch}$ may be conservative in finite samples. Nonetheless, in our experiments (Section~\ref{sec:experiments}), we find that these $p$-values still lead to sharp localization across all empirical settings considered. Moreover, employing a randomized construction of $p_t$, as in~\eqref{eq:pvalue_conch_randomized}, yields an exactly uniform $p$-value, thereby providing exact finite-sample validity (Theorem~\ref{thm:coverage_exactly_1-alpha}) for the resulting confidence set.
\end{remark}

\begin{remark}[Time-reversal symmetry]\label{remark:symmetry}
The \conch{} confidence set is invariant under the reversal of the timeline. Let $\mathbf{Y}=(Y_1,\ldots,Y_n)$ be the reversal of $\bX=(X_1,\ldots,X_n)$, i.e., $Y_i:=X_{n-i+1}$. Localizing $\xi\in[n-1]$ given $\bX$ is then equivalent to localizing $n-\xi$ given $\mathbf{Y}$. Indeed, the permutation group $\Pi_t$ acting on $\bX$ corresponds to $\Pi_{n-t}$ acting on $\mathbf{Y}$, and with the score functions defined accordingly, the \conch{} confidence set computed from $\mathbf{Y}$ is exactly the image of the \conch{} confidence set from $\bX$ under the map $t\mapsto n-t$.
\end{remark}

\subsection{\conch-MC: randomized variant for scalability}\label{sec:conch_mc}

To compute the \conch{} $p$-value $p_t$ in \eqref{eq:pvalue_conch}, one must enumerate all permutations in $\Pi_t$ and compute the corresponding score $S_t(\pi(\bX))$ for each $\pi$. For large $n$, this may be computationally expensive. To reduce computational burden, we proceed as follows: independent of the observed data $\bX$, sample $\pi^{(1)},\ldots,\pi^{(M)} \stackrel{\text{i.i.d.}}{\sim} \text{Unif}(\Pi_t)$,
and then calculate:
\begin{equation}\label{eq:pvalue_conch_MC}
    \hat{p}_t := \frac{1 + \sum_{k=1}^M \One{ S_t(\pi^{(k)}(\bX)) \leq S_t(\bX)}}{1+M}.
\end{equation}  
This yields a \emph{randomized} confidence set $\{\, t \in [n-1] : \hat{p}_t > \alpha \,\}$.
We refer to this procedure as \conch-MC, presented formally in \Cref{alg:conch_MC} in Appendix~\ref{app:conch_variant}. Similar to \conch{}, any randomly sampled $\pi\in \Pi_t$ preserves pre-change and post-change exchangeability under $\Hcal_{0,t}$, thereby providing us with a valid $p$-value $\hat{p}_t$, as we state in Theorem~\ref{thm:coverage-conch_MC} and prove in Appendix~\ref{app:proof_of_coverage}. 

\begin{theorem}\label{thm:coverage-conch_MC}
    For any $t \in [n-1]$, $\hat{p}_t$ defined in \eqref{eq:pvalue_conch_MC} is a valid $p$-value under $\Hcal_{0,t}$. In particular, for any $\alpha \in (0,1)$,
    \(\P_{\xi}(\hat{p}_{\xi}\leq \alpha)\leq \alpha\).
    Consequently, $\mathcal{C}^{\mathrm{CONCH\text{-}MC}}_{1-\alpha}$ is a distribution-free confidence set for changepoint, i.e.,
    \[
    1-\alpha\le \P(\xi\in \mathcal{C}^{\mathrm{CONCH\text{-}MC}}_{1-\alpha}).
    \]
\end{theorem}

\subsection{\conch{}-SPLIT: sample-split variant}
\label{sec:conch_split}

Recall that, in general, the CPP score in Algorithm~\ref{alg:conch} may depend non-trivially on the full dataset. For readers familiar with conformal predictive inference, this corresponds to an adaptation of the full-conformal approach to our problem, which often can be computationally costly. We now introduce a split-conformal analogue, called \conch{}-SPLIT, in which one subset of the observations is used to learn the CPP score and the remaining observations are used for computing \conch{} $p$-values.

Formally, let $\mathcal I_1$ and $\mathcal I_2$ be a fixed, data-independent partition of $[n]$ and assume that they are non-trivial, i.e. $|\mathcal{I}_1|,|\mathcal{I}_2|\ge 2$. We allow the two subsets to have unequal sizes, but typically use an interlaced split (such as an odd and even subsequences), so that both subsequences span the full timeline. Write $\mathcal D_1=(X_i)_{i\in\mathcal I_1}$ and $\mathcal D_2=(X_i)_{i\in\mathcal I_2}$. Enumerating $\mathcal I_2=\{i_1<\cdots<i_m\}$, we further write $\boldsymbol Y=(X_{i_1},\ldots,X_{i_m})$.

Using only $\mathcal D_1$, we construct a CPP score $\widehat S=\mathcal A(\mathcal D_1):\mathcal X^m\to\mathbb R^{m-1}$. The learned score $\widehat S$ is then held fixed while running \conch{} on $\boldsymbol Y$. Specifically, for each $t\in[m-1]$, define the split-permutation group $\Pi_t^{(2)}
:=\left\{\pi\in\mathcal S_m:\pi(r)\le t\ \text{for }r\le t,\quad \pi(r)>t\ \text{for }r>t\right\}$,
and compute the \conch{} $p$-value
\begin{equation}\label{eqn:conch_split_pvalue}
    p_t^{\dagger}:=\frac{1}{|\Pi_t^{(2)}|}\sum_{\pi\in\Pi_t^{(2)}}\One{\widehat S_t(\pi(\boldsymbol Y))\le \widehat S_t(\boldsymbol Y)}.
\end{equation}
Naturally these $p$-values define a confidence set for the changepoint in $\mathcal{D}_{n,2}$.
To map the resulting confidence set back to the original timeline, write $i_0=1$ and $i_{m+1}=n$ and define $\mathcal B_t:=\{i_t,\ldots,i_{t+1}-1\}$ for $t\in[m]\cup \{0\}$. The \conch{}-SPLIT confidence set is then given by
\begin{equation}\label{eqn:conch_split_CI}
    \mathcal C^{\conch\text{-SPLIT}}_{1-\alpha}:=\mathcal B_0\cup\mathcal B_m\cup\bigcup_{\substack{t\in[m-1]\\p_t^{\dagger}>\alpha}}\mathcal B_t.
\end{equation}
The algorithm is formally given in \Cref{alg:conch-split} in Appendix~\ref{app:conch_variant}.

\begin{corollary}[Validity of \conch{}-SPLIT]
\label{thm:coverage-conch-split}
Consider the setting of Theorem~\ref{thm:coverage-conch}, and fix $\alpha\in(0,1)$. Suppose that the partition $(\mathcal I_1,\mathcal I_2)$ is chosen independently of the data $\bX$. Then, we have that
$\mathbb P_\xi\bigl(\xi\in\mathcal C^{\conch\text{-SPLIT}}_{1-\alpha}\bigr)\ge 1-\alpha$.
\end{corollary}

The result follows immediately by first conditioning on $\mathcal{D}_{n,1}$, and applying Theorem~\ref{thm:coverage-conch}. Compared with Algorithm~\ref{alg:conch}, \conch{}-SPLIT avoids repeatedly fitting the score across permutations. This can substantially reduce the computational cost for complex learning algorithms $\mathcal A$, at the expense of smaller sample sizes for fitting and inference. Note that for practical consideration, the \conch{} $p$-values in~\eqref{eqn:conch_split_pvalue} can be replaced by \conch-MC $p$-values.

\section{Choosing the CPP score}\label{sec:score}

\subsection{General properties of CPP score}
The \conch{} confidence sets introduced earlier remain valid for any choice of CPP score, offering substantial flexibility. However, a well-chosen score yields a narrower and more informative set. In this section, we first establish general properties of CPP scores and then derive an optimal score. Although this optimal score requires oracle knowledge, we propose practical alternatives that closely approximate it. Proofs of the results in this section are deferred to Appendix~\ref{app:proof_of_optimality}.

\begin{proposition}\label{prop:score-properties}
Fix $n\in\mathbb{N}$ and $\alpha\in(0,1)$.
\begin{enumerate}[label=(\roman*), ref=(\roman*)]
    \item \textnormal{\textbf{(Symmetry yields trivial $p$-values).}}
    Fix $t\in[n-1]$. If $S$ is $t$-symmetric, i.e., satisfies $S_t(\cdot)=S_t(\pi(\cdot))$ for all $\pi\in \Pi_t$,
    then $p$-value $p_t$ in \eqref{eq:pvalue_conch} equals $1$, and $\mathbb{P}(t\in\Ccal^{\mathrm{CONCH}}_{1-\alpha})=1$.
    \item \textbf{(Conformal data-processing inequality).}
 Let $C_1$ be the \conch{} set based on score $S$, with $p$-values ${p_{t,1}}$. For any non-decreasing $f:\mathbb{R}\to\mathbb{R}$, let $C_2$ be the corresponding set based on $f(S)$, with $p$-values ${p_{t,2}}$. Then $p_{t,1}\le p_{t,2}$ for all $t\in [n-1]$, and therefore $C_1\subseteq C_2$.

\end{enumerate}
\end{proposition}
Part~(i) shows that $t$-symmetric CPP scores yield trivial conch{} $p$-values, regardless of whether \(\mathcal H_{0,t}\) holds, and therefore lead to overly conservative confidence sets. Such scores should be avoided in practice. Part~(ii) shows a monotonicity property of \conch{}: applying any non-decreasing transformation to the CPP score can only enlarge the resulting set, and any \emph{strictly} increasing transformation leaves the set unchanged. These properties guide practical choices of CPP scores that yield informative confidence sets.

\subsection{Optimal CPP score function}\label{sec:optimal_score}

Now, we focus on the canonical setting, namely the i.i.d. changepoint model. Specifically, let $\Pcal_{\mathrm{IID}}$ denote the class of distributions for which there exists $\xi\in [n-1]$ such that \[\Pcal_{0,\xi}=\otimes_{t=1}^{\xi} P_0,\quad \Pcal_{1,\xi}=\otimes_{t=\xi+1}^{n} P_1,\]
where $P_0$ and $P_1$ admit densities $f_0$ and $f_1$ with respect to a common dominating measure $\nu$ on $\mathcal{X}$.

Below we establish an expression for an optimal CPP score function, assuming the knowledge of both densities $f_0$, $f_1$, and the true changepoint $\xi$. By framing the task of identifying an optimal score as a testing problem with a point null and a point alternative, we can directly apply the classical Neyman–Pearson (NP) lemma. This yields a similar optimality result tailored to the setting of distribution-free changepoint localization, which we call the \emph{second}\footnote{The first instance of such a Conformal NP Lemma appears in \citet{dandapanthula2025conformal}, which establishes an analogous NP optimality result for a conformal $p$-value-based changepoint test.} \emph{Conformal NP Lemma}. Before stating the lemma formally, we first set up some notation.

Fix any $t\in[n-1]$, and let $\mathcal{X}_{L,t}:=\{X_1,\ldots,X_t\}$ and $\mathcal{X}_{R,t}:=\{X_{t+1},\ldots,X_n\}$ denote the (unordered) left and right multisets. Now for any $r\in [n-1]$, let $\mathcal{P}_{\bX}^{[r]}=\mathcal{P}_{0,r}\times \mathcal{P}_{1,r}$
be the law of $\bX=(X_1,\ldots,X_n)$ corresponding to a changepoint at $r$ under the i.i.d.\ model class $\Pcal_{\mathrm{IID}}$. We also define $\Pcal^{[r]}_{\bX\mid \mathcal{X}_{L,t},\mathcal{X}_{R,t}}$ for the associated conditional distribution of $\bX$ given $(\mathcal{X}_{L,t},\mathcal{X}_{R,t})$, 
and write $\Hcal'_r$ to hypothesize that $\bX\mid(\mathcal{X}_{L,t},\mathcal{X}_{R,t})\sim\Pcal^{[r]}_{\bX\mid \mathcal{X}_{L,t},\mathcal{X}_{R,t}}$
for any $r\in [n-1]$. Note that $\Pcal^{[t]}_{\bX\mid \mathcal{X}_{L,t},\mathcal{X}_{R,t}}$ is a discrete uniform distribution. Now, suppose we want to test $\Hcal'_t$ using conformal $p$-values. Given a score $s:\mathcal{X}^n\to\R$ and the permutation set $\Pi_t$, we define the randomized conformal $p$-value
\begin{equation}\label{eq:conformal_pvalue}
p_t(s) =\frac{1}{|\Pi_t|}\sum_{\pi\in\Pi_t}\One{s(\pi(\bX))<s(\bX)}\ +\ U\cdot \frac{1}{|\Pi_t|}\sum_{\pi\in\Pi_t}\One{s(\pi(\bX))=s(\bX)},
\end{equation}
where $U\sim\mathrm{Unif}(0,1)$ is independent of $\bX$ and the permutations.
By Theorem~\ref{thm:coverage_exactly_1-alpha}, $p_t(s)$ is ${\rm Unif}[0,1]$ under $\Hcal_t^\prime$. Consequently, the test
$\phi_t(\bX;s)=\One{p_t(s)\leq \alpha)}$
controls type~I error at level $\alpha$ for the null $\Hcal^\prime_{t}$ with any score function $s$. We note that we use the randomized $p$-value here to utilize the entire type~I budget and maximize power with an appropriate score function.

We seek an optimal score $s^\star$ such that the corresponding test $\phi_t(\bX;s^\star)$ achieves maximum power against an alternative $\Hcal^\prime_{r}$ (with $r\neq t$). 
The second Conformal Neyman–Pearson lemma, stated below, formally establishes that the likelihood ratio $s^\star(\cdot)=\Pcal_X^{[t]}(\cdot)/\Pcal_X^{[r]}(\cdot)$ defines the optimal test.

\begin{lemma}[Second Conformal NP lemma]\label{lem:conch_NP_lemma}
    Fix $t,r\in [n-1]$ with $t\neq r$. The power, $\E_{\Hcal^\prime_{r}}[\phi_t(\bX;s)]$, is maximized by the score function
    \[
    s^\star(x_1,\ldots,x_n) :=\frac{\prod_{i\leq t}f_0(x_i)\prod_{i>t} f_1(x_i)}{\prod_{i\leq r}f_0(x_i)\prod_{i>r} f_1(x_i)}.
    \]
\end{lemma}

Finally, the Conformal NP lemma can be leveraged within the \conch{} framework to derive the CPP score that would yield the narrowest confidence set. We observe that the conformal $p$-value in \eqref{eq:pvalue_conch} must be valid under $\Hcal_{0,t}$, while being sufficiently small to sharply detect the true changepoint $\xi\neq t$ under $\Hcal_{0,\xi}$. Since only the $t$-th component of CPP score, $S_t$, determines $p_t$, the task of optimizing $S_t$ boils down to finding the optimal test for $\Hcal^\prime_{t}~\text{vs.}~\Hcal^\prime_{\xi}$. 

We make this connection precise in the theorem below. For notational convenience, we let $\bar{C}_{1-\alpha}^{\conch}(S)$ denote the randomized \conch{} set \eqref{eqn:conch_set_randomized} constructed with the CPP score $S$.
\begin{theorem}\label{thm:optimal_score}
    Any strictly increasing transformation of the CPP score \(S^\textnormal{OPT}\) defined by
    \begin{equation}\label{eq:optimal_CPP_score}
       S_t^\textnormal{OPT}(x_1,\ldots,x_n)
       = \frac{\prod_{i\le t} f_0(x_i)\prod_{i> t} f_1(x_i)}{\prod_{i\le \xi} f_0(x_i)\prod_{i> \xi} f_1(x_i)}
    \end{equation}
    achieves the minimum expected length of the \conch{} confidence set. In particular, for any $\xi\in[n-1]$ and score \(S:\mathcal{X}^n\to\R^{n-1}\),
    \( \E_{\Hcal_{0,\xi}\,\cap\,\Pcal_{\textnormal{IID}}}\!\big[\,|\bar{C}_{1-\alpha}^{\conch}(S)|\,\big]
    \;\ge\;
    \E_{\Hcal_{0,\xi}\,\cap\,\Pcal_{\textnormal{IID}}}\!\big[\,|\bar{C}_{1-\alpha}^{\conch}(S^{\textnormal{OPT}})|\,\big].
    \)
\end{theorem}

The optimal CPP score function~\eqref{eq:optimal_CPP_score} depends on the unknown pre-change and post-change densities $f_0$ and $f_1$ as well as the true changepoint $\xi$, and is therefore not directly implementable in practice. In the next subsection, we propose score functions that closely mimic the optimal score, thus providing `near-optimal' performance in practice.

\subsection{Practical choices for CPP score}\label{sec:practical_scores}

\paragraph{Weighted mean difference.} If $P_0$ and $P_1$ differ only by a location shift, we may take
    \begin{equation}\label{score:weighted_mean_difference}
        S_t(x_1,\cdots,x_n) = \left\|\frac{\sum_{i=1}^t w_{t,i} x_i}{\sum_{i=1}^t w_{t,i}}-\frac{\sum_{i>t}^n w_{t,i} x_i}{\sum_{i>t}^n w_{t,i}}\right\|_1.
    \end{equation}
    The weights $\{w_{t,i}\}$ are introduced to break the $t$-wise symmetry property, and therefore to avoid trivial confidence sets. Since, under the i.i.d.\ changepoint model, observations at time points near a candidate index $t \neq \xi$ are drawn from the same distribution on both sides of $t$, assigning larger weights to observations closer to $t$ tends to produce similar weighted averages on the two sides. Consequently, such a localized weighting scheme yields a smaller CPP score $S_t$, which in turn leads to smaller $p_t$ values for $t\neq \xi$ and hence more efficient confidence sets.
    
    Reasonable choices for weights include: $w_{t,i} = 1-(|i-t|/n)$ or $w_{t,i} = \exp\bigl(-|i-t|/n\bigr)$.
   If $t\in [n-1]$ is believed to be a changepoint, the weighted means on the left and right sides should differ substantially, producing a high CPP score at $t$ as required.
   
    \paragraph{Oracle log likelihood-ratio (LLR).} Suppose $f_0$ and $f_1$ are known. Then, the denominator of the
    optimal CPP score function in~\eqref{eq:optimal_CPP_score} can be approximated by evaluating the complete likelihood at the maximum likelihood estimate (MLE) instead of the true changepoint $\xi$. This yields the CPP score given by 
    \begin{equation}\label{score:oracle_llr}
    S_t(x_1,\cdots,x_n)=\log\left(\frac{\prod_{i\leq t}f_0(x_i)\prod_{i>t} f_1(x_i)}{\prod_{i\leq \hat{\xi}(\bx)}f_0(x_i)\prod_{i>\hat{\xi}(\bx)} f_1(x_i)}\right),    
    \end{equation}
    where  $\hat{\xi}(\bx)\in\argmax_{s\in [n-1]}\log \bigl(\prod_{i\leq s}f_0(x_i)\prod_{i>s} f_1(x_i)\bigr)$
    is the MLE\footnote{Note that the MLE estimator $\hat{\xi}_{\rm OR}(\bx)$ is a function of the observed data $(x_1,\ldots,x_n)$. Thus, for each permutation $\pi$, computing the permuted score $S_t(\pi(\mathbf{X}))$ requires first computing $\hat{\xi}_{\rm OR}(\pi(\mathbf{X}))$, MLE estimate on the permuted data.} of the changepoint.
    If $t\in [n-1]$ is indeed the changepoint, then $\hat{\xi}\approx t$ and $S_t$ will be large, indicating strong plausibility for a change. Since this score closely approximates \eqref{eq:optimal_CPP_score}, it is expected to sharply localize the changepoint, as verified in experiments.

     \paragraph{Learned LLR.}
    When $f_0$ and $f_1$ are unknown, for each $t\in [n-1]$, one can plug in estimates (parametric or non-parametric) $\hat{f}_{t,0}$ and $\hat{f}_{t,1}$, 
    and instead consider the CPP score given by
    \begin{equation}\label{score:learned_llr}
    S_t(x_1,\cdots,x_n)=\log\left(\frac{\prod_{i\leq t}\hat{f}_{t,0}(x_i)\prod_{i>t} \hat{f}_{t,1}(x_i)}{\prod_{i\leq \hat{\xi}(\bx)}\hat{f}_{\hat{\xi}(\bx),0}(x_i)\prod_{i>\hat{\xi}(\bx)} \hat{f}_{\hat{\xi}(\bx),1}(x_i)}\right)
    \end{equation}
    with $\hat{\xi}(\bx)\in\argmax_{s\in [n-1]}\log \bigl(\prod_{i\leq s}\hat{f}_{s,0}(x_i)\prod_{i>s} \hat{f}_{s,1}(x_i)\bigr)$ being the corresponding MLE.

    Here, the density estimates $\hat{f}_{t,0}$ and $\hat{f}_{t,1}$ may be learned separately from the observations to the left and right of $t$. While this provides the most general formulation, more computationally efficient alternatives are often available. For instance, when an independent labeled dataset is available, the corresponding densities may be learned once on that dataset and reused across all candidate changepoints. When such labeled data are unavailable, one may instead cluster the observed data, viewed as an unordered multiset, into two groups and estimate the corresponding densities separately. Since the resulting groups are unlabeled, the two density estimates can then be assigned to $\hat{f}_{t,0}$ and $\hat{f}_{t,1}$ so as to best match the observations to the left and right of $t$.
    
    \paragraph{Classifier based LLR.} Instead of estimating the densities $f_0$ and $f_1$ directly, one can train a binary classifier $\hat{g}$ to distinguish post-change from pre-change samples (labeled $Y=1$ and $Y=0$, respectively). By Bayes' rule, we have
    $\log(f_1(x)/f_0(x))=\log \bigl(\P(Y=1|X=x)/\P(Y=0\mid X=x)\bigr)-\log (\pi_1/\pi_0)$,
    where $\pi_1$ and $\pi_0$ are class priors. If $\hat{g}$ is trained on balanced data and we write $\hat{g}(x) \in (0,1)$ to denote the predicted probability of post-change membership, then 
    \[\log \frac{f_1(x)}{f_0(x)}\approx \mathrm{logit}\, \hat{g}(x):=\log \frac{\hat{g}(x)}{1-\hat{g}(x)}.\]
    The log odds components in \eqref{score:oracle_llr} can then be approximated by the classifier logits to define a practically implementable CPP score. While the choice of classifiers does not affect the validity of our method, a well-trained classifier improves power.

    The LLR-based scores above are presented in their most general form, corresponding to the full-conformal implementation of \conch{} in Algorithm~\ref{alg:conch}. Since the estimated changepoint $\hat{\xi}$, and possibly the estimated densities or log-likelihood ratio, depend on the observed data sequence, they must generally be recomputed for each permutation, which can be computationally expensive. Alternatively, one may use \conch{}-SPLIT from Section~\ref{sec:conch_split}: the first split is used to estimate $\hat{\xi}$ and, for the learned LLR score, the underlying densities or log-likelihood ratio, after which these quantities are held fixed while \conch{} is applied to the second split.

\section{Asymptotic sharpness of \conch{}-SPLIT confidence sets}
\label{sec:sharpness_learned_llr}

We now show that, with the LLR-based scores from the previous section, \conch{} enables asymptotically sharp localization. In particular, the \conch{}-SPLIT confidence set has a vanishing relative cardinality. We consider the i.i.d.\ changepoint model: for each $n$, we observe $\mathcal D_n:=(X_{1,n},\ldots,X_{n,n})\in\mathcal X^n$ with a single changepoint at $\xi_n\in[n-1]$, such that $X_{1,n},\ldots,X_{\xi_n,n}\overset{iid}{\sim}P_0$, while $X_{\xi_n+1,n},\ldots,X_{n,n}\overset{iid}{\sim}P_1$. We assume that $P_0$ and $P_1$ admit densities $f_0$ and $f_1$, respectively, with respect to a common dominating measure $\nu$, and write $\ell(x):=\log(f_0(x)/f_1(x))$ for the corresponding log-likelihood ratio. Since we do not have access to $\ell$, we suppose that an estimator $\hat{\ell}_n$ of $\ell$ is obtained from an auxiliary dataset $\mathcal D_n^\prime$ that is independent of $\mathcal D_n$.

For notational simplicity and clarity of exposition, throughout this section we assume that $n$ is even. To run \conch{}-SPLIT, we use the fixed interlaced partition $\mathcal I_1=\{1,3,\ldots,n-1\}$ and $\mathcal I_2=\{2,4,\ldots,n\}$. Thus, $\mathcal D_{n,1}:=(X_{1,n},X_{3,n},\ldots,X_{n-1,n})$ is the training split, on which we learn the MLE, while $\mathcal D_{n,2}:=(X_{2,n},X_{4,n},\ldots,X_{n,n})$ is the calibration split used to compute the conformal $p$-values. Both splits have size $m:=n/2$.

Using the training split $\mathcal D_{n,1}$, we first compute the reference changepoint estimator
\begin{equation}
\label{eq:split_learnt_MLE}
\hat{\xi}_n
\in \argmax_{s\in[m-1]}
\sum_{i=1}^{s}\hat{\ell}_n(X_{2i-1,n}).
\end{equation}
Since the training split is half the size of the original sequence, we expect $2\hat{\xi}_n$ to be a good estimate of the true changepoint location $\xi_n$. As in Section~\ref{sec:conch_split}, enumerate the observations in the calibration split as $Y_{i,n}:=X_{2i,n}$ for $i\in[m]$. Conditional on $\mathcal D_n^\prime$ and $\mathcal D_{n,1}$, both $\hat{\ell}_n$ and $\hat{\xi}_n$ are held fixed while \conch{} is applied to $\mathcal D_{n,2}$. To implement \conch{}, we use the learned LLR CPP score
\begin{equation}
\label{score:CPP_split_learnt}
S_{t,n}(\boldsymbol Y)
=\begin{cases}
-\sum_{i=t+1}^{\hat{\xi}_n}\hat{\ell}_n(Y_{i,n}),
& t\leq \hat{\xi}_n,\\
\sum_{i=\hat{\xi}_n+1}^{t}\hat{\ell}_n(Y_{i,n}),
& t>\hat{\xi}_n.
\end{cases}
\end{equation}
This is the sample-split analogue of the learned LLR score described in~\eqref{score:learned_llr}. This formulation accommodates both estimating $f_0$ and $f_1$ separately and directly learning $\ell$, for example through a probabilistic classifier trained on an independent dataset.

Let $\mathcal C^{\conch\text{-SPLIT}}_{n,1-\alpha}$ denote the confidence set obtained by applying \conch{} to $\mathcal D_{n,2}$ with the score in~\eqref{score:CPP_split_learnt} and mapping the resulting set back to the original timeline as in~\eqref{eqn:conch_split_CI}. Under a mild consistency condition on $\hat{\ell}_n$, this confidence set is asymptotically sharp.

\begin{theorem}[Sharpness with learned LLR score]
\label{thm:conch-consistency}
Suppose that, in the above setting, there exists $\tau\in(0,1)$ such that $\xi_n/n\to\tau$ as $n\to\infty$, and that $\mathrm{Var}_{X\sim P_0}(\ell(X)),\mathrm{Var}_{X\sim P_1}(\ell(X))\in(0,\infty)$. Further, suppose that, for each $P\in\{P_0,P_1\}$,
\begin{equation}
\label{eqn:consistency_of_hat_ell_n}
\E_{\substack{\mathcal D_n^\prime,\,X\sim P,\\ X\independent\mathcal D_n^\prime}}
\left[\bigl|\hat{\ell}_n(X)-\ell(X)\bigr|^2\right]
\longrightarrow 0
\qquad\text{as }n\to\infty.
\end{equation}
Then, $\bigl|\mathcal C^{\conch\text{-SPLIT}}_{n,1-\alpha}\bigr|/(n-1)\xrightarrow{P}0$ as $n\to\infty$.
\end{theorem}

The proof is provided in Appendix~\ref{app:proof_of_consistency}. Note that sample splitting can also be replaced by cross-fitting. That is, one may learn the score on $\mathcal D_{n,1}$ and run \conch{} on $\mathcal D_{n,2}$ at confidence level $1-\alpha/2$, then reverse the roles of the two splits and combine the two resulting confidence sets by taking an intersection. This allows every observation to contribute to both score learning and conformal $p$-value computation.

Here, the condition $\xi_n/n\to\tau\in(0,1)$ ensures that the
changepoint remains in the interior of the sequence, making the
localization problem nontrivial. The finite-variance condition is
imposed primarily for simplicity of exposition, and can be weakened to the
existence of a finite $(1+\delta)$-th moment for some $\delta>0$ with additional modifications to the proof.

We note that the assumption in~\eqref{eqn:consistency_of_hat_ell_n} can be satisfied using an additional sample-splitting scheme. Given $4n$ observations, define four equal splits of the data as $\mathcal D_{(1)}=\{X_{4k-3}:k\in[n]\}$, $\mathcal D_{(2)}=\{X_{4k-2}:k\in[n]\}$, $\mathcal D_{(3)}=\{X_{4k-1}:k\in[n]\}$, and $\mathcal D_{(4)}=\{X_{4k}:k\in[n]\}$. A natural approach is to first estimate the component densities $f_0$ and $f_1$ from $\mathcal D_{(1)}\cup\mathcal D_{(2)}$, yielding estimators $\hat f_{0,n}$ and $\hat f_{1,n}$. We may then define $\hat{\ell}_n=\log(\hat f_{0,n}/\hat f_{1,n})$, estimate $\hat{\xi}_n$ using $\mathcal D_{(3)}$, and finally run \conch{} on $\mathcal D_{(4)}$. The estimator $\hat{\ell}_n$ will satisfy~\eqref{eqn:consistency_of_hat_ell_n} if for instance, the densities $f_0$ and $f_1$ and their estimates are bounded above and bounded away from $0$, and the density estimators achieve $L_1$-consistency, namely,
\begin{equation}
\label{eq:density_l1_consistency}
\|\hat f_{0,n}-f_0\|_1,\ 
\|\hat f_{1,n}-f_1\|_1
=o_P(1).
\end{equation}

To obtain estimators $\hat f_{0,n}$ and $\hat f_{1,n}$, one may first construct a preliminary changepoint estimator $\xi_n^\dagger$ using $\mathcal D_{(1)}$ such that $|\xi_n^\dagger-\xi_n|=o_P(n)$. This is a relatively mild requirement and can be achieved using the kernel changepoint detection method \citep{garreau2018consistent} or MCP localization \citep{bhattacharyya2025theoretical}. Given such a preliminary point estimator, the observations in $\mathcal D_{(2)}$ lying to the left and right of $\xi_n^\dagger$ contain i.i.d.\ samples from $P_0$ and $P_1$, respectively, with at most an $o_P(1)$ fraction of contamination from the opposite distribution. Estimating $f_0$ and $f_1$ from these segments is therefore a problem of robust density estimation under a vanishing contamination model, a setting that is well studied \citep{chen2016general,uppal2020robust} and for which there exist estimators satisfying~\eqref{eq:density_l1_consistency}.

Note that condition~\eqref{eqn:consistency_of_hat_ell_n} only requires consistency of the estimated log-likelihood ratio. Although high dimensionality may slow the convergence of $\hat{\ell}_n$, and consequently the rate at which the confidence set shrinks, asymptotic sharpness continues to hold provided that $\hat{\ell}_n$ consistently estimates $\ell$.

When we have access to the oracle log-likelihood ration $\ell$, the result on asymptotic sharpness can be further strengthened. In fact, in Theorem~\ref{thm:conch_consistency_oracle}, we show that when $\ell$ is known, the corresponding \conch{} confidence set is stochastically bounded. In other words, the \emph{relative} size of the confidence set is $\mathrm{O}_P(1/n)$ since the original formulation of bounded width is in relation to seeing $n$ observations.

\section{Universality of the \conch{} algorithm}\label{sec:universality}
In earlier sections, we have established \conch{} as a flexible framework for constructing distribution-free confidence sets for the changepoint. One may naturally ask: is \conch{} one of many such distribution-free changepoint localization approaches that one may come up with?
In this section, we give a conclusive answer to this question, which is that \conch{} truly captures the full class of distribution-free changepoint localization methods.

\begin{theorem}\label{thm:universality_changepoint}
    Fix $\alpha\in (0,1)$. Let $C:\mathcal{X}^n \to 2^{[n-1]}$ be any procedure that maps $\bX \in \mathcal{X}^n$ to a set $C(\bX) \subseteq [n-1]$ such that
    \(\P_{\xi}\bigl(\xi\in C(\bX)\bigr)\geq 1-\alpha.
    \)
    Then, there exists a CPP score function $S:\mathcal{X}^n\to\R^{n-1}$ such that $C$ coincides exactly with the set $\mathcal{C}^{\conch}_{1-\alpha}$ constructed with the score $S$.
\end{theorem}

Theorem~\ref{thm:universality_changepoint} may be viewed as a changepoint localization analogue of the universality result for full conformal prediction in predictive inference \citep{vovk2005algorithmic}, but requires some novel arguments.
The proof is provided in Appendix~\ref{app:proof of universality theorem}. In simple terms, this result states that a particular choice of CPP score leads to a specific instance within the universal class of valid procedures for changepoint localization. 

Moreover, this universality naturally yields an algorithm for refining existing confidence sets: \conch{} can wrap around any heuristic or model-based confidence set or existing point estimates of $\xi$, calibrating it to achieve exact distribution-free coverage while preserving the original method’s structural or modeling advantages.

\section{\conch{} as a calibration framework}
\subsection{Calibrating heuristic confidence sets}\label{sec:calibration}
\IncMargin{1.2em}
\begin{algorithm}[t]
    \caption{\conch-CAL: \conch{} calibration algorithm}
    \label{alg:conch-calib}
    \KwIn{$(X_t)_{t=1}^n$ (dataset), $t_0:\mathcal{X}^n\to [n-1]$ (point estimate), and $\mathrm{pval}:\mathcal{X}^n\to[0,1]^{\,n-1}$ ($p$-value function)}
    \KwOut{$\mathcal{C}^{\,\mathrm{CONCH\text{-}CAL}}_{1-\alpha}$ (\conch-CAL confidence set at level $1-\alpha$)}
    Define $\hat{S}:\mathcal{X}^n\to\R^{\,n-1}$ as in \eqref{score:p-value_score} \;
    \For{$t \in [n-1]$}{ 
        Compute \conch{} $p$-value $p_t$ as in \eqref{eq:pvalue_conch} with $S_t$ replaced by $\hat{S}_t$\;
    }
    $\mathcal{C}^{\,\mathrm{CONCH\text{-}CAL}}_{1-\alpha} \gets \{\, t \in [n-1] : p_t > \alpha \,\}$\;
    \Return{$\mathcal{C}^{\,\mathrm{CONCH\text{-}CAL}}_{1-\alpha}$}
\end{algorithm}
\DecMargin{1.2em} 

We now demonstrate how any confidence set for the true changepoint $\xi$ can be calibrated to attain distribution-free coverage guarantee.
Suppose we are given a non-empty confidence set $C:\mathcal{X}^n \to 2^{[n-1]}$ that may or may not be valid, even asymptotically; for instance, one produced by a Bayesian or Bootstrap method. We can construct a CPP score function based on the given set and thereby obtain a distribution-free, finite-sample valid confidence set by running \conch{}.
Two such natural constructions of CPP score are:
\begin{itemize}
    \item \textbf{Set membership score.} Define $S_t(x_1,\ldots,x_n) = \One{t \in C(x_1,\ldots,x_n)}$, which records only whether $t$ is included in the given confidence set.
    \item \textbf{Set distance score.} Define $S_t(x_1,\ldots,x_n) = (1+\min_{\ell \in C(x_1,\ldots,x_n)} |t-\ell|)^{-1}$, which refines the membership score by measuring the distance of $t$ to the nearest index included in the set.
\end{itemize}

Running the \conch{} algorithm with either score yields a valid confidence set. However, by Proposition~\ref{prop:score-properties}~$(ii)$, the set-distance score always produces a narrower set than the set-membership score. Nonetheless, both scores are fairly coarse and often lead to wide confidence sets. In particular, for indices near $0$ or $n$, these scores can induce $t$-symmetry, producing artificially inflated $p$-values in those regions (Proposition~\ref{prop:score-properties}~$(i)$). Since this behavior is undesirable in practice, we next introduce a more informative CPP score that yields sharper confidence sets.

Most existing model-based or resampling-based approaches also produce a point estimate $t_0:\mathcal{X}^n\to [n-1]$. Moreover, in many cases, they first construct a $p$-value function $\mathrm{pval}:\mathcal{X}^n\to[0,1]^{n-1}$, which is then appropriately thresholded to form the confidence set $C$. Both components ($t_0(\bx), \mathrm{pval}$) can be combined to define a more informative CPP score,
\begin{equation}\label{score:p-value_score}
\hat{S}_t(x_1,\ldots,x_n)= \frac{\mathrm{pval}(x_1,\ldots,x_n;t)}{\mathrm{pval}(x_1,\ldots,x_n;t_0(\bx))}.
\end{equation}
Applying \conch{} with this score yields what we refer to as the \conch-CAL algorithm, formally presented in \Cref{alg:conch-calib}.  We note that here the point estimate $t_0$ depends on the ordered sequence $(x_1,\ldots,x_n)$, and thus the denominator $\mathrm{pval}(\cdot,\ldots,\cdot;t_0(\bx))$ is not invariant under permutations.

By construction, this produces a valid distribution-free confidence set while retaining the original method’s assessment of the changepoint. In practice, this allows analysts to exploit the strengths of bootstrap or Bayesian methods, such as their interpretability, while simultaneously ensuring finite-sample coverage. Although one could in principle
use $\mathrm{pval}(x_1,\ldots,x_n;t)$ directly as the CPP score in \conch{}, this approach typically inherits the same conservativeness observed with set-membership and set-distance scores. See Appendix~\ref{sec:calibration_expt} for an empirical evaluation of these methods.

\subsection{Calibrating point estimates into confidence sets}
\label{sec:point_estimate_calibration}
The application of \conch{} as a wrapper is not just limited to existing heuristic confidence sets. More generally, any point estimate of the true changepoint $\xi$ can be incorporated into the \conch{} framework through an appropriate choice of CPP score from Section~\ref{sec:practical_scores}.

To illustrate this idea, we consider modifications of the weighted mean-difference score and the learned LLR score. First, for each $t\in[n-1]$, let
\[
M_{L,t}:=\frac{\sum_{i=1}^t w_{t,i}X_i}{\sum_{i=1}^t w_{t,i}}, \qquad
M_{R,t}:=\frac{\sum_{i=t+1}^n w_{t,i}X_i}{\sum_{i=t+1}^n w_{t,i}},
\]
denote the weighted means of the observations to the left and right of $t$, respectively. Given a point estimate $\hat{\xi}_0:\mathcal{X}^n\to \R^{n-1}$ of the changepoint $\xi$, we define the CPP score
\[
S_t=-\bigl(|M_{L,t}-M_{L,\hat{\xi}_0}|+|M_{R,t}-M_{R,\hat{\xi}_0}|\bigr).
\]
This ensures that candidate changepoints close to the reference estimate $\hat{\xi}_0$ yield larger CPP scores.

Next, recall the learned log-likelihood ratio CPP score from~\eqref{score:learned_llr}. Observe that this score compares the log-likelihood of the data under the hypothesis that the changepoint is at $t$ against the log-likelihood under the hypothesis that the changepoint is at $\hat{\xi}$, the maximum likelihood estimate. A natural way to calibrate an existing point estimate using \conch{} is therefore to replace this reference point $\hat{\xi}(\bX)$ by the given estimate $\hat{\xi}_0(\bX)$.

Specifically, writing \[\widehat{\mathcal L}(\bX,t):=\sum_{i\le t}\log\widehat f_{0,t}(X_i)+\sum_{i>t}\log\widehat f_{1,t}(X_i),\]
we define the CPP score $S_t(\bX)=\widehat{\mathcal L}(\bX,t)
-\widehat{\mathcal L}(\bX,\hat{\xi}_0)$,
where $\widehat f_{0,t}, \widehat f_{1,t}$ denote densities estimated separately from  observations to the left and right of $t$.

Running \conch{} with either of the above CPP scores yields a confidence set that calibrates the uncertainty associated with the given point estimate. More broadly, this provides a simple mechanism for wrapping an arbitrary changepoint estimator within the \conch{} framework and converting it into a finite-sample valid confidence set. In Appendix~\ref{app:conch_around_changepoint}, we empirically demonstrate this approach using a changepoint estimate from the KCP algorithm  \citep{arlot2019kernel}.
\section{Experiments}\label{sec:experiments}

\subsection{Detecting Gaussian mean-shift}\label{sec:gaussian_mean_shift}

We begin our empirical evaluation of \conch{} with the classical Gaussian mean-shift model. In all experiments\footnote{All code required to reproduce our experiments is available at \url{https://github.com/rohanhore/CONCH}.} reported in this paper, confidence sets are computed using the Monte--Carlo variant \conch-MC (\Cref{alg:conch_MC}) with $M=300$, which retains the same finite-sample validity guarantee as the full-permutation procedure while remaining computationally feasible in practice.

Specifically, we generate a sequence of $n=1000$ i.i.d.\ observations with a changepoint at $\xi=400$: the pre-change distribution is $\Pcal_{0,\xi} = \bigotimes_{t=1}^{\xi}\mathcal{N}(-1,1)$, while the post-change distribution is $\Pcal_{1,\xi} = \bigotimes_{t=\xi+1}^{n}\mathcal{N}(1,1)$.  
We evaluate \conch{} using four choices of CPP scores, introduced earlier in Section~\ref{sec:practical_scores}:
(a) weighted mean difference, with a specified weight function, (b) oracle log-likelihood ratio (LLR), (c) parametrically learned LLR, assuming knowledge of the Gaussian family and (d) nonparametrically learned LLR, via kernel density estimates. Here and throughout this paper, while applying \conch, we recompute the MLE for all the LLR scores for each permutation.
\begin{figure}[!h] 
\centering 
\includegraphics[width=1\linewidth]{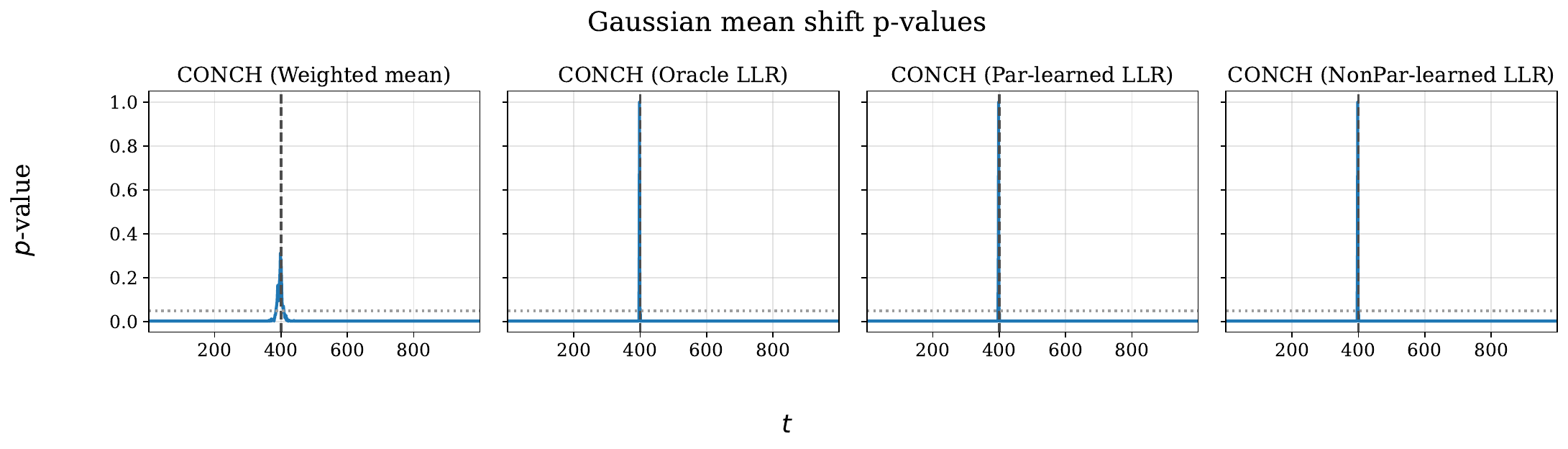} 
\caption{\conch{} $p$-value profiles for the Gaussian mean-shift experiment using four CPP scores. All scores yield valid confidence sets, with the LLR-based scores providing substantially sharper localization.}
\label{fig:gaussian_experiment} 
\end{figure}

Figure~\ref{fig:gaussian_experiment} displays the distribution of the resulting $p$-values produced by each method. \conch{} produces sharply localized confidence sets across all score choices. The weighted-mean score results in the widest interval, $[385,408]$, whereas all three LLR-based scores (oracle, parametrically learned and non-parametrically learned) yield a much narrower set $\{397,398,400\}$. 

Overall, these results highlight two key features: (i) the validity of \conch{} is preserved regardless of the choice of score, and (ii) more informative scores lead to sharper localization. 
Appendix~\ref{app:gaussian_all_pvalue_comaprison} presents an additional comparison with \citet{dandapanthula2025conformal}.

\subsection{Real data experiments}

\subsubsection{DomainNet: detecting domain shift}\label{sec:domainnet_expt}

In this experiment, we tackle the problem of detecting a domain shift using the DomainNet dataset \citep{peng2019moment}, consisting of six diverse domains (e.g., real, sketch, painting). Among these, we use the \emph{real} and \emph{sketch} domains to construct a changepoint localization setting. Further, we convert all images to grayscale to remove color cues and increase the similarity between domains. Specifically, before the changepoint ($\xi=350$), we observe samples from the real domain, and after $\xi$, we observe samples from the sketch domain, totaling $n=800$ observations. Example images are shown in the left panel of Figure~\ref{fig:domain_shift_combined}. Since DomainNet was collected via online search, class labels may not perfectly align with visual semantics, making the domain-shift localization problem more challenging.

\begin{figure}[!t]
    \centering
    \includegraphics[width=\linewidth]{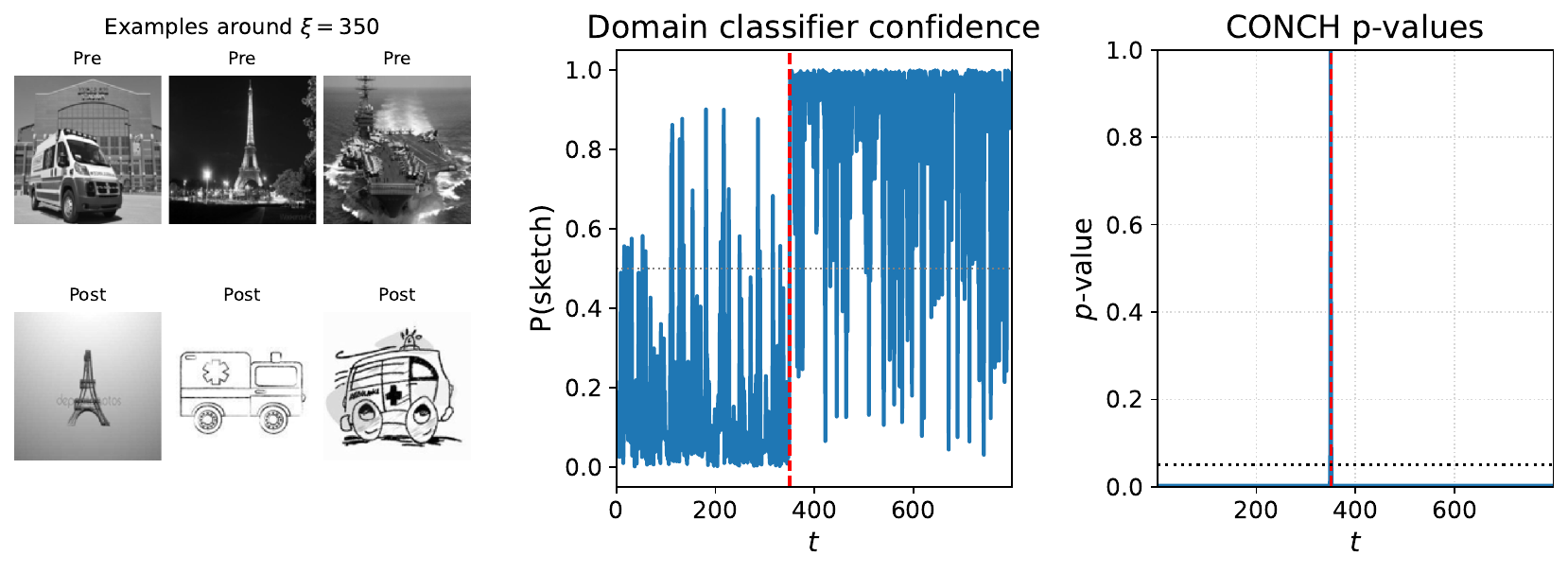}
    \caption{DomainNet domain-shift experiment. Left: example images around the changepoint $\xi=350$, with real images before $\xi$ and sketches after $\xi$. Middle: predictions from a CNN trained on an independent labeled dataset to distinguish the two domains. Right: \conch{} $p$-values obtained from the classifier logits, yielding the $95\%$ confidence set $[350,351]$.}

    \label{fig:domain_shift_combined}
\end{figure}

We first train a CNN classifier on an independent labeled dataset of $2000$ samples to distinguish real images from hand-drawn sketches. Although the classifier provides substantial discriminative information, it does not directly translate into distribution-free guarantees for changepoint localization. The \conch{} framework bridges this gap by converting classifier outputs into a principled distribution-free procedure via the use of classifier based LLR score introduced in Section~\ref{sec:practical_scores}. As shown in Figure~\ref{fig:domain_shift_combined}, the resulting \conch{} $p$-values yield a narrow $95\%$ confidence set $[350,351]$, accurately localizing the true changepoint.

\subsubsection{SST-2: detecting sentiment change using language models}\label{sec:sentiment_expt}
We next demonstrate our method on text data, showing that it can localize changepoints in language settings. Using the Stanford Sentiment Treebank (SST-2) dataset of movie reviews with binary sentiment labels \citep{socher2013recursive}, we simulate a shift from predominantly positive to predominantly negative sentiment, mirroring real-world tasks such as detecting changes in customer feedback or public opinion. We observe $n=1000$ reviews with a changepoint at $\xi=400$: before $\xi$, reviews are i.i.d. positive ($P_0$); after $\xi$, reviews are i.i.d. negative ($P_1$). For example see Figure~\ref{fig:sst2_conch_pvalues}.

First, we consider the pre-trained DistilBERT model fine-tuned for sentiment classification \citep{sanh2019distilbert}, and then the corresponding model logits are used to build a CPP score for our \conch{} method, which yields a $95\%$ confidence set $[400,401]$, effectively pinpointing the changepoint. Even under a subtler scenario, where sentiment shifts only from $60\%$ positive to $40\%$ positive, we obtain a nontrivial $95\%$ confidence set $[326,463]$, demonstrating sharp localization of the changepoint in complex settings. See Figure~\ref{fig:sst2_conch_pvalues} for a visualization of the \conch{} $p$-values in both settings.

\begin{figure}[!h]
\centering

\begin{minipage}[t]{0.33\textwidth}
\vspace{0pt}
\small

\fbox{%
\begin{minipage}{0.9\linewidth}
\textbf{Example reviews}

\vspace{0.5em}

\begin{tabular}{@{}l p{0.68\linewidth}@{}}
$t=399:$ & ``juicy writer''\\
$t=400:$ & ``intricately structured and well-realized drama''\\
$t=401:$ & ``painfully''\\
$t=402:$ & ``than most of jaglom's self-conscious and gratingly irritating films''
\end{tabular}
\end{minipage}
}
\end{minipage}
\hfill
\begin{minipage}[t]{0.65\textwidth}
\vspace{0pt}
\centering
\includegraphics[width=\linewidth]{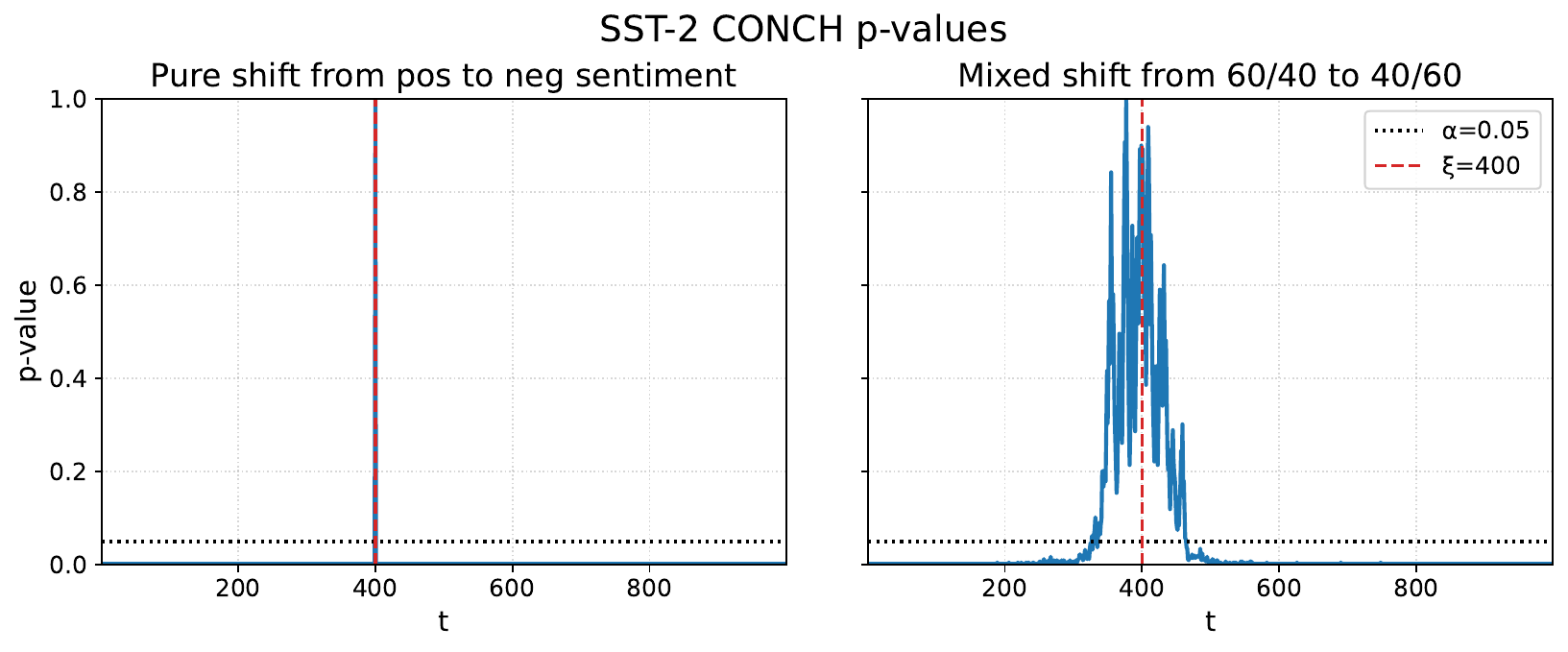}
\end{minipage}

\caption{The leftmost panel shows example reviews around the changepoint $\xi=400$: reviews before $\xi$ are positive, while reviews after $\xi$ are negative. The two panels on the right display \conch{} $p$-values for a shift from positive to negative reviews (left), yielding a $95\%$ confidence set $[400,401]$, and for a shift from $60\%$ positive to $40\%$ positive reviews (right), yielding a $95\%$ confidence set $[326,463]$.}
\label{fig:sst2_conch_pvalues}
\end{figure}

textbf{Additional experiments:}
Appendix~\ref{app:additional_expt} presents a collection of supplementary experiments. We first compare \conch{} with several well-known changepoint detection methods in terms of localization accuracy, and compare localization performance of \conch{} with \cite{dandapanthula2025conformal}. We then empirically verify the finite-sample coverage guarantees of the proposed framework and demonstrate how \conch{} can be used to calibrate existing changepoint confidence sets and point estimators. Next, we investigate several extensions of \conch{}, including multiple changepoint localization and localization under temporal dependence. Finally, we consider a two-urn example (urn-shift detection), a digit-class shift on MNIST, and a class-shift experiment on CIFAR-100, further illustrating the robustness and flexibility of the proposed framework.

\section{Conclusion}\label{sec:conclusion}
In this work, we introduced \textsc{CONCH}, a framework for distribution-free offline changepoint localization. Our approach leverages conformal $p$-values to construct confidence sets with finite-sample, distribution-free guarantees. We provided design guidelines, including principled choices of score functions and a Monte Carlo approximation to the full-permutation $p$-value, that enhance both the power and practicality of the framework. With an appropriate score function, the \conch{} confidence sets shrink with sample size, as one would expect from any useful localization procedure.

We established a universality result positioning \conch{} as a canonical method for distribution-free offline changepoint localization. This, in turn, paves the way for simple calibration procedures that can wrap around any localization algorithm to yield valid confidence sets and accurately quantify the uncertainty associated with a point estimate of the true changepoint.

In many practical applications, data often exhibit temporal dependence and may contain multiple changepoints. A promising direction for future work is therefore to develop computationally efficient extensions of \conch{} for multiple-changepoint localization, possibly by integrating ideas from wild binary segmentation \citep{fryzlewicz2014wild}. Another important direction is to extend the framework beyond independent observations and establish validity guarantees under temporal dependence, thereby substantially broadening the scope and applicability of conformal changepoint localization.

While a complete treatment of these challenges is beyond the scope of the present work, the appendices provide several initial steps in these directions. In particular, Appendix~\ref{app:extensions} studies extensions of \conch{} to multiple changepoint localization and temporally dependent data, alongside a variant that attains exact finite-sample coverage guarantees. Appendix~\ref{app:proofs} contains complete proofs of all theoretical results, together with several auxiliary results, and Appendix~\ref{app:additional_expt} provides all the supplementary experiments.
\section*{Acknowledgments}
The authors acknowledge the funding from the Sloan Fellowship.

\bibliographystyle{plainnat}
\bibliography{bibliography}

\newpage
\appendix

\tableofcontents

\section{Extensions of \conch{} algorithm}\label{app:extensions}
\IncMargin{1.2em} 
\begin{algorithm}[t]
    \caption{\conch-MC: \conch{} with random permutations}
    \label{alg:conch_MC}
    \KwIn{$(X_t)_{t=1}^n$ (dataset), $1-\alpha$ (target coverage), $M$ (number of permutations) and $ S:\mathcal{X}^n\to \R^{n-1} $ (CPP score function)}
    \KwOut{$\mathcal{C}^{\,\mathrm{CONCH\text{-}MC}}_{1-\alpha}$ (\conch-MC confidence set at level $1-\alpha$)}
    \For{$t \in [n-1]$}{ 
        $\Pi_t \gets \{\pi \in \mathcal S_n : \textnormal{for all~}i\leq t, \ \pi(i)\le t\,\, \textnormal{and for all}\, i>t,\ \pi(i)>t\}$\;
        \For{$k \in [M]$}{
            Sample $\pi^{(k)}\sim\Pi_t$\; 
            Evaluate $S_t(\pi^{(k)}(\bX))$\;}
        $\hat{p}_t \gets \frac{1}{M+1}\bigl(1+\sum_{k=1}^M \One{S_t(\pi^{(k)}(\bX))\leq S_t(\bX)}\bigr)$\;\label{line:pvalue_conch_MC}
    }
    $\mathcal{C}^{\mathrm{CONCH\text{-}MC}}_{1-\alpha} \gets \{ t \in [n-1] : \tilde{p}_t > \alpha \}$ 
    
    \Return{$\mathcal{C}^{\mathrm{CONCH\text{-}MC}}_{1-\alpha}$}
\end{algorithm}
\DecMargin{1.2em} 

\IncMargin{1.2em} 
\begin{algorithm}[t]
\caption{\conch{}-SPLIT: sample-split conformal changepoint localization}
\label{alg:conch-split}
\KwIn{$(X_i)_{i=1}^n$ (data), $1-\alpha$ (target coverage),
$\mathcal I_1,\mathcal I_2$ (data-independent partition),
and $\mathcal A$ (CPP-score learning rule)}
\KwOut{$\mathcal C^{\conch\text{-SPLIT}}_{1-\alpha}$}

Construct $\mathcal D_1$ and
$\boldsymbol Y=(X_{i_1},\ldots,X_{i_m})$ from the two splits\;
Learn $\widehat S\leftarrow\mathcal A(\mathcal D_1)$\;

\For{$t\in[m-1]$}{
    Construct $\Pi_t^{(2)}$\;
    Evaluate $\widehat S_t(\pi(\boldsymbol Y))$ for each
    $\pi\in\Pi_t^{(2)}$, keeping $\widehat S$ fixed\;
    Compute $p_t^{\dagger}$ as in \eqref{eqn:conch_split_pvalue}\;
}
Construct the sets $\mathcal B_0,\ldots,\mathcal B_m$\;
$\mathcal C^{\conch\text{-SPLIT}}_{1-\alpha}
\leftarrow
\mathcal B_0\cup\mathcal B_m
\cup\displaystyle\bigcup_{t:\,p_t^{\dagger}>\alpha}\mathcal B_t$\;

\Return{$\mathcal C^{\conch\text{-SPLIT}}_{1-\alpha}$}\;
\end{algorithm}
\DecMargin{1.2em} 

\subsection{Achieving exact validity for \conch{} confidence sets}\label{app:conch_variant}
While both $p$-values $p_t$ and $\tilde{p}_t$ in \eqref{eq:pvalue_conch} and \eqref{eq:pvalue_conch_MC} control the Type~I error under $\Hcal_{0,t}$ at level~$\alpha$, it is sometimes desirable to attain \emph{exact} level-$\alpha$ validity. Achieving exact validity can yield more powerful or sharper procedures. To this end, we introduce a randomized refinement of the $p$-values that guarantees exact validity.
Specifically, define
\begin{equation}\label{eq:pvalue_conch_randomized}
    \bar{p}_t := \frac{1}{|\Pi_t|}\sum_{\pi\in \Pi_t} \One{S_t(\pi(\bX)) < S_t(\bX)}
    + U \cdot \frac{1}{|\Pi_t|}\sum_{\pi\in \Pi_t} \One{S_t(\pi(\bX)) = S_t(\bX)},
\end{equation}
where $U \sim \textnormal{Unif}[0,1]$. Plugging these randomized $p$-values into the general \conch{} framework yields the confidence set
\begin{equation}\label{eqn:conch_set_randomized}
   \bar{\Ccal}^{\mathrm{CONCH}}_{1-\alpha} = \{\, t \in [n-1] : \bar{p}_t > \alpha \,\}, 
\end{equation}
which attains confidence exactly $1-\alpha$, as formalized below.
\begin{theorem}\label{thm:coverage_exactly_1-alpha}
    For each $t \in [n-1]$, $\bar{p}_t$ defined in \eqref{eq:pvalue_conch_randomized} is a valid $p$-value under $\Hcal_{0,t}$ conditional on the multisets $\{X_1,\ldots,X_\xi\}$ and $\{X_{\xi+1},\ldots, X_n\}$. Therefore, for any $\alpha \in (0,1)$,
    \[
        \P_{\xi}\!\left(\bar{p}_{\xi} \le \alpha\right) = \alpha,
    \]
     and consequently, $\P_{\xi}\!\left(\xi \in \bar{\Ccal}^{\mathrm{CONCH}}_{1-\alpha}\right) = 1 - \alpha$ where $\bar{\Ccal}^{\mathrm{CONCH}}_{1-\alpha}$ is as defined in \eqref{eqn:conch_set_randomized}.
\end{theorem}
The proof of the result is deferred to Appendix~\ref{app:proof_of_coverage}.

\subsection{Extension: localization of multiple changepoints}
\label{sec:mult_chpt_extension}
In this section, we extend \conch{} to the setting of multiple changepoint localization. To that end, note that the permutation framework underlying \conch{} can, in principle, be extended directly to the multiple changepoint setting while retaining finite-sample validity, as in Theorems~\ref{thm:coverage-conch} and~\ref{thm:coverage-conch_MC}.

Suppose, for simplicity, that the number of changepoints $K>1$ is known. One may treat the entire tuple of changepoint locations, $\boldsymbol{\xi}:=(\xi_1,\ldots,\xi_K)$, as the parameter of interest. Let
\[
\mathcal{R}_K=\Bigl\{(t_1,\ldots,t_K): 1\le t_1<t_2<\ldots<t_K\le n-1~\text{for all}~j\in [K]\Bigr\}
\]
denote the collection of all admissible changepoint configurations. 

We extend the \conch{} framework by first generalizing the CPP score function to be a map $S:\mathcal{X}^n\to\mathbb{R}^{\mathcal{R}_K}$, where $S_{\boldsymbol{t}}$ quantifies the plausibility of $\boldsymbol{t}$ being the true changepoint configuration. Next, we generalize the permutation framework by permuting observations independently within each of the segments induced by $\boldsymbol{t}$. Note that any candidate configuration $\boldsymbol{t}=(t_1,\ldots,t_K)\in\mathcal{R}_K$ induces the segmentation $[t_0+1,t_1],[t_1+1,t_2],\ldots, [t_K+1,t_{K+1}]$ with letting $t_0=0$ and $t_{K+1}=n$.
Specifically, for each $\boldsymbol{t}\in\mathcal{R}_K$, define the split-permutation group
\[
\Pi_{\boldsymbol{t}}:=\Bigl\{\pi\in\mathcal{S}_n:\pi([t_j+1,t_{j+1}])=[t_j+1,t_{j+1}]\text{ for all } j=0,\ldots,K\Bigr\}.
\]
We may then define the conformal $p$-value for the candidate configuration $\boldsymbol{t}$ as
\[
p_{\boldsymbol{t}}:=\frac{1}{|\Pi_{\boldsymbol{t}}|}
\sum_{\pi\in\Pi_{\boldsymbol{t}}}
\One{S_{\boldsymbol{t}}(\pi(\bX))\leq S_{\boldsymbol{t}}(\bX)}.
\]
Finally, the changepoint confidence set is constructed by thresholding these $p$-values at level $\alpha$:
\[
\mathcal{C}_{1-\alpha}:=
\{\boldsymbol{t}\in\mathcal{R}_K: p_{\boldsymbol{t}}>\alpha\}.
\]
By the same exchangeability argument underlying Theorem~\ref{thm:coverage-conch}, this extension satisfies the finite-sample validity guarantee
\[
\P\bigl(\boldsymbol{\xi}\in\mathcal{C}_{1-\alpha}
\bigr)\geq 1-\alpha.
\]

The main challenge with this approach is computational. The cardinality of $\mathcal{R}_K$ grows rapidly with $n$, and each candidate tuple induces a different permutation structure. Consequently, a direct implementation quickly becomes infeasible even for moderate sample sizes.

Motivated by this observation, we instead propose a practical segmentation-based extension of \conch{}. The key idea is that, given a sufficiently consistent segmentation algorithm, the sequence can be partitioned into disjoint segments such that each segment contains, with high probability, at most one changepoint. We may then run \conch{} independently within each segment and aggregate the resulting confidence sets to obtain an overall confidence set for the multiple changepoint problem.

Formally, suppose there exist $K \in [n]$ changepoints
$0 = \xi_0 < \xi_1 < \cdots < \xi_K < n = \xi_{K+1}$,
such that for each $\ell \in \{0,1,\ldots,K\}$,
\begin{equation}\label{model:multiple_changepoint}
    (X_{\xi_{\ell}+1}, \ldots, X_{\xi_{\ell+1}}) \overset{iid}{\sim} \mathcal{P}^{(\ell)},
\end{equation}
To be consistent with \Cref{assn:exchangeability}, we assume each $\mathcal{P}^{(\ell)}$ is different and that the $(K+1)$ segments are mutually independent. Further, suppose a segmentation algorithm returns  
(a)~an estimate $\hat{K}$ of the number of changepoints, and  
(b)~an ordered sequence of estimated changepoints
$0 = \hat{\xi}_0 < \hat{\xi}_1 < \cdots < \hat{\xi}_{\hat{K}} < n = \hat{\xi}_{\hat{K}+1}$,
such that $\hat{K} \approx K$ and $\hat{\xi}_\ell \approx \xi_\ell$ for all $\ell \in [K]$.  

\IncMargin{1.2em}
\begin{algorithm}[t]
    \caption{\conch-SEG: Segmentwise \conch{} for Multiple Changepoints}
    \label{alg:conch-seg}
    \KwIn{$(X_t)_{t=1}^n$ (data); $\hat{K}$ and $0=\hat{\xi}_0<\hat{\xi}_1<\cdots<\hat{\xi}_{\hat{K}}<n=\hat{\xi}_{\hat{K}+1}$ (estimated changepoints); $S:\mathcal{X}^n\to\R^{n-1}$ (CPP score)}
    \KwOut{$\mathcal{C}^{\conch\text{-SEG}}_{1-\alpha}$ (overall confidence set at level $1-\alpha$)}
    Compute $(\tilde{X}_0,\ldots,\tilde{X}_{\hat{K}})$ as in~\eqref{eq:segmentation}\;
    Initialize $\mathcal{C} \gets \varnothing$\;
    \For{$\ell \in [\hat{K}]$}{
        $(L_\ell, R_\ell)\gets (\tilde{X}_{\ell-1},\,\tilde{X}_\ell)$\;
        Let $X^{(\ell)} \gets (X_{L_\ell}, \ldots, X_{R_\ell})$\;
        Define score $S^{(\ell)}:\mathcal{X}^{R_\ell - L_\ell + 1} \to \R^{R_\ell - L_\ell}$\;
        Compute \conch{} $p$-values $\{p_t : t \in [L_\ell,\,R_\ell{-}1]\}$ as in~\eqref{eq:pvalue_conch}, using $S^{(\ell)}$ on $X^{(\ell)}$\;
        Set $\mathcal{C}_\ell \gets \{\, t \in [L_\ell,\,R_\ell{-}1] : p_t > \alpha \,\}$\;
        Update $\mathcal{C} \gets \mathcal{C} \cup \mathcal{C}_\ell$\;
    }
    \Return{$\mathcal{C}^{\conch\text{-SEG}}_{1-\alpha} \gets \mathcal{C}$}
\end{algorithm}
\DecMargin{1.2em}

Based on these estimates, we discretize the timeline $\{1,\ldots,n\}$ into $\hat{K}$ data-dependent segments centered at the $\hat{\xi}_\ell$’s. Specifically, for $\ell\in\{0,\ldots,\hat{K}\}$, define
\begin{equation}\label{eq:segmentation}
    \tilde{X}_\ell :=
    \begin{cases}
        0, & \text{if } \ell = 0,\\[4pt]
        \big\lfloor\tfrac{1}{2}(\hat{\xi}_\ell + \hat{\xi}_{\ell+1})\big\rfloor, & \text{if } \ell \in [\hat{K}-1],\\[4pt]
        n, & \text{if } \ell = \hat{K}.
    \end{cases}
\end{equation}
We then take the $\ell$-th segment to be $[\tilde{X}_{\ell-1}+1,\,\tilde{X}_\ell]$ for $\ell\in[\hat{K}]$. Running \conch{} independently on each segment and aggregating the segmentwise sets yields our overall confidence set. The resulting procedure, denoted \conch-SEG, is summarized in \Cref{alg:conch-seg}.

Kernel-based changepoint detection (KCP) methods \citep{harchaoui2007retrospective,arlot2019kernel} provide consistent estimators of changepoints under mild conditions \citep[e.g.,][]{garreau2018consistent}. Consequently, the \conch{} framework can be seamlessly wrapped around a KCP routine to construct confidence sets in the multiple-changepoint setting.

In Appendix~\ref{app:multiple_changepoint_expt}, we consider a Gaussian mean-shift model with multiple changepoints and empirically show that \conch-SEG, when wrapped around a KCP algorithm, sharply localizes the changepoints, demonstrating that this extension is both practical and statistically powerful.

On the theory side, a direct analysis of \conch-SEG is challenging because segmentation and \conch{} are applied to the same data, potentially violating \Cref{assn:exchangeability}. That said, Appendix~\ref{app:multiple_changepoint_theory} establishes asymptotic validity of a cross-fitted variant of \conch-SEG. Specifically, cross-fitting uses disjoint folds of the data to estimate changepoints and subsequently run \conch{} within the estimated segments. This restores the required independence and exchangeability structure, yielding an asymptotic coverage guarantee.

\subsection{Changepoint localization beyond independent observations}\label{app:conch-dep}

Throughout this paper, we often restrict to the i.i.d.\ changepoint model, which assumes that the observed data sequence consists of independent observations. In many practical applications, however, such as financial time series or network traffic data, the observed data sequence exhibits temporal dependence.

That said, changepoint localization under dependence presents additional challenges. First, under arbitrary dependence, it is not immediately clear how a changepoint should be formally defined. One must first impose some structural assumptions on the data-generating process and define the notion of change relative to that model. For instance, as a natural generalization of \Cref{assn:exchangeability}, one may assume stationarity on either side of the changepoint together with some form of weak dependence.

Even under such assumptions, distribution-free finite-sample validity is generally difficult to achieve without explicitly accounting for the underlying dependence structure. Nevertheless, in many practical applications the dependence is relatively weak. Therefore, if the permutation principle underlying \conch{} is restricted to a subclass of permutations that largely preserves the local dependence structure, one may regard the resulting permuted samples as approximately exchangeable copies.

To illustrate the idea formally, in the spirit of \Cref{assn:exchangeability}, suppose that the pre-change and post-change sequence is independent, the data sequence is stationary on either side of the changepoint and satisfies an $m$-dependence assumption, that is, observations separated by more than $m$ time points are independent. With the knowledge of $m$, for each candidate changepoint $t\in[n-1]$, define the subsequence
\[
\mathcal D_t^{(m)}
=
\Bigl(
\ldots,
X_{t-2(m+1)},
X_{t-(m+1)},
X_t,
X_{t+1},
X_{t+1+(m+1)},
X_{t+1+2(m+1)},
\ldots
\Bigr),
\]
obtained by retaining only observations whose indices differ by multiples of $m+1$ from either $t$ to the left or from $t+1$ to the right.

If $t=\xi$, then under this setting, the observations in $\mathcal D_\xi^{(m)}$ are mutually independent. Further, by stationarity on either side of the changepoint, observations to the left of $\xi$ in the subsequence are i.i.d.\ from the pre-change distribution, while observations to the right are i.i.d.\ from the post-change distribution. Consequently, the subsequence $\mathcal D_\xi^{(m)}$ satisfies \Cref{assn:exchangeability}.

We therefore define the $p$-value $p_t^{(m)}$ at candidate changepoint $t$ by applying the standard \conch{} procedure to the subsequence $\mathcal D_t^{(m)}$, treating the location corresponding to $t$ in the subsequence as the candidate changepoint. We call the resulting procedure \conch-DEP, and the confidence set is then given by
\[
\mathcal C^{\conch\text{-DEP}}_{1-\alpha}
=
\{t\in[n-1]: p_t^{(m)}>\alpha\}.
\]

\begin{theorem}[Finite-sample validity of \conch-DEP]
\label{thm:conch_dep}
Suppose the observed data sequence is $m$-dependent and stationary on either side of the true changepoint $\xi$, and the pre-change sequence is independent of post-change sequence. Then, for any $\alpha\in(0,1)$,
\[
\P\!\left(\xi\in \mathcal C^{\conch\text{-DEP}}_{1-\alpha}\right)\ge 1-\alpha.
\]
\end{theorem}

The proof is immediate from the above discussion and the finite-sample validity of \conch{} established in Theorem~\ref{thm:coverage-conch}.

While under the $m$-dependent setting finite-sample validity can be attained exactly, beyond $m$-dependence the same is generally no longer attainable in a finite sample manner, and one must instead aim for approximate or asymptotic validity guarantees. While the dependence structure is rarely known in practice, observations that are sufficiently far apart often behave approximately independently under weak dependence. Consequently, one may expect the \conch-DEP procedure, with a suitably large spacing parameter $m$ (chosen based on an independent split of data), to provide a useful approximation to the idealized independent setting.

In Appendix~\ref{app:dependent_expt}, we empirically investigate the performance of \conch-DEP under dependent data and observe behavior consistent with this intuition.

That said, a systematic study of changepoint localization under general temporal dependence, together with corresponding validity guarantees, remains an interesting direction for future research.
\section{Proofs}\label{app:proofs}
In this section, we present proofs of the main results stated in the paper, along with several auxiliary results that support them.

\subsection{Proving coverage guarantees for \conch{}}\label{app:proof_of_coverage}
\subsubsection{Proof of Theorem~\ref{thm:coverage-conch}}
    First, observe that under the null $\Hcal_{0t}$, $\pi(\bX)\overset{d}{=}\bX$ for any $\pi\in \Pi_t$. We define a function $p_t:\mathcal{X}^n\to [0,1]$ by
    \[
    p_t(\bx):=\frac{1}{|\Pi_t|}\sum_{\pi\in \Pi_t} \One{S_t(\pi(\bx))\leq S_t(\bx)}.
    \]
    Further, note that $p_t\equiv p_t(\bX)$.
    Therefore,
    \begin{align*}
        \P_{t}\left(p_t(\bX)\leq \alpha\right)&=\frac{1}{|\Pi_t|}\sum_{\pi\in \Pi_t} \P_{t}\left(p_t(\pi(\bX))\leq \alpha\right)\\
        &=\E_{t}\left[\frac{1}{|\Pi_t|}\sum_{\pi\in \Pi_t}\One{p_t(\pi(\bX))\leq \alpha}\right]\\
        &=\E_{t}\left[\frac{1}{|\Pi_t|}\sum_{\pi\in \Pi_t}\One{\frac{1}{|\Pi_t|}\sum_{\pi^\prime\in \Pi_t} \One{S_t(\pi^\prime(\bX))\leq S_t(\pi(\bX))}\leq \alpha}\right]\leq \alpha,
    \end{align*}
    where the penultimate step follows by noting that $\pi\circ \Pi_t=\Pi_t$, and the last inequality follows by \citet[Lemma~3]{harrison2012conservative}. This completes the proof. $\hfill\mathsf{\square}$

    \subsubsection{Proof of Theorem~\ref{thm:coverage-conch_MC}}
    Given permutations $\pi_{1,t},\ldots,\pi_{M,t}\in \Pi_t$, we define the function
    \[
    \hat{p}_t(\bx;\pi_{1,t},\ldots,\pi_{M,t}):=\frac{1+\sum_{k=1}^M\One{s_t(\pi_{k,t}(\bx))\leq s_t(\bx)}}{1+M},
    \]
    Consider an additional uniform draw $\pi_{0,t}$ from $\Pi_t$. 
    
    Hence, note that with $\pi_{1,t},\ldots,\pi_{M,t}\overset{iid}{\sim}\textnormal{Unif}(\Pi_t)$, we have that
    \[
    (\pi_{1,t},\ldots,\pi_{M,t})\overset{d}{=}(\pi_{0,t}\circ\pi_{1,t},\ldots,\pi_{0,t}\circ\pi_{M,t}).
    \]
    Moreover, conditional on $\pi_{0,t},\pi_{1,t},\ldots,\pi_{M,t}$, $\bX\overset{d}{=}\pi_{0,t}(\bX)$ under the null $\mathcal{H}_{0,t}$. Consequently,
    \begin{multline*}
        \hat{p}_t(\bX;\pi_{1,t},\ldots,\pi_{M,t})\overset{d}{=}\hat{p}_t(\bX;\pi_{0,t}\circ\pi_{1,t},\ldots,\pi_{0,t}\circ\pi_{M,t})\overset{d}{=}\hat{p}_t(\pi_{0,t}(\bX);\pi_{0,t}\circ\pi_{1,t},\ldots,\pi_{0,t}\circ\pi_{M,t}).
    \end{multline*}
    Finally, note that for $\hat{p}_t$, defined in \eqref{eq:pvalue_conch_MC},  $\hat{p}_t\equiv \hat{p}_t(\bX;\pi_{1,t},\ldots,\pi_{M,t})$, and therefore,
    \begin{align*}
            \hat{p}_t(\bX;\pi_{1,t},\ldots,\pi_{M,t})&\overset{d}{=}\hat{p}_t(\pi_{0,t}(\bX);\pi_{0,t}\circ\pi_{1,t},\ldots,\pi_{0,t}\circ\pi_{M,t})\\
            &=\frac{1+\sum_{k=1}^M\One{s_t(\pi_{k,t}(\bX))\leq s_t(\pi_{0,t}(\bX))}}{M+1}\\
            &=\frac{\sum_{k=0}^M\One{s_t(\pi_{k,t}(\bX))\leq s_t(\pi_{0,t}(\bX))}}{M+1},
        \end{align*}
    i.e., the rank of $s_t(\pi_{0,t}(\bX))$ in the exchangeable collection $\{s_t(\pi_{0,t}(\bX)),s_t(\pi_{1,t}(\bX)),\ldots, s_t(\pi_{M,t}(\bX))\}$. Consequently,
    \[
    \P_t\left(\hat{p}_t=\hat{p}_t(\bX;\pi_{1,t},\ldots,\pi_{M,t})\leq \alpha\right)\leq \alpha.
    \]
    This proves the result. 
    $\hfill\mathsf{\square}$
    
\subsubsection{Proof of Theorem~\ref{thm:coverage_exactly_1-alpha}}
 We begin by letting $F$ denote the distribution of $S_t(\pi(\bX))$ conditional on the multisets $M_{\textnormal{left}}:=\{X_1,\ldots,X_t\}$ and $M_{\textnormal{right}}:=\{X_{t+1},\ldots,X_n\}$, where $\pi \sim \textnormal{Unif}(\Pi_t)$. 
Then
\[
\bar{p}_t = \lim_{y \uparrow S_t(\bX)} F(y) + U \bigl(F(S_t(\bX)) - \lim_{y \uparrow S_t(\bX)} F(y)\bigr).
\]
Under $\Hcal_{0,t}$, we have $S_t(\bX) \overset{d}{=} S_t(\pi(\bX))$ conditional on $M_{\textnormal{left}}$ and $M_{\textnormal{right}}$. 
Hence, by \citet[Lemma~E.1]{dandapanthula2025conformal}, the $p$-value $\bar{p}_t$, conditional on $M_{\textnormal{left}}$ and $M_{\textnormal{right}}$, follows $\textnormal{Unif}[0,1]$ (see also \citealp{brockwell2007universal}). 
Therefore,
\[
\P_t(\bar{p}_t \le \alpha) 
= \E_t\!\left[\P_t\!\left(\bar{p}_t \le \alpha \mid M_{\textnormal{left}}, M_{\textnormal{right}}\right)\right]
= \E_t[\alpha] = \alpha.
\]
This completes the proof.

\subsection{Proving properties of the CPP score and its optimal form}\label{app:proof_of_optimality}
\subsubsection{Proof of Proposition~\ref{prop:score-properties}}
The first part of the result follows immediately by noting that when $S_t$ satisfies the $t$-symmetry, then by \eqref{eq:pvalue_conch} $p_t$ is identically equal to $1$, as required.

For the second part, fix $t\in [n-1]$. By definition (see \eqref{eq:pvalue_conch}),
\[
p_{t,1}=\frac{1}{|\Pi_t|}\sum_{\pi\in \Pi_t} \One{S_t(\pi(\bX))\leq S_t(\bX)},\quad 
p_{t,2}=\frac{1}{|\Pi_t|}\sum_{\pi\in \Pi_t} \One{f(S_t(\pi(\bX)))\leq f(S_t(\bX))}.
\]
Since $f$ is non-decreasing,
\[
S_t(\pi(\bX))\leq S_t(\bX)\implies f(S_t(\pi(\bX)))\leq f(S_t(\bX)),
\]
and therefore $p_{t,1}\leq p_{t,2}$. As this holds for all $t\in [n-1]$, it further follows that $C_1\subseteq C_2$.

\subsubsection{Proof of Lemma~\ref{lem:conch_NP_lemma}: second conformal NP lemma}\label{app:proof_of_conformal_NP}
In the setup of Section~\ref{sec:optimal_score}, we consider the following hypothesis testing problem:
\[
\Hcal^\prime_{0}:\bX\mid \mathcal{X}_{L,t}, \mathcal{X}_{R,t}\sim \mathcal{P}^{[t]}_{\bX\mid \mathcal{X}_{L,t}, \mathcal{X}_{R,t}}
\quad \text{v.s.} \quad  
\Hcal^\prime_{1}:\bX\mid \mathcal{X}_{L,t}, \mathcal{X}_{R,t}\sim \mathcal{P}^{[r]}_{\bX\mid \mathcal{X}_{L,t}, \mathcal{X}_{R,t}}.
\]
Given samples $\bX\in \mathcal{X}^n$, observe that
\[
\frac{\mathsf{d}\bigl(\mathcal{P}^{[r]}_{\bX\mid \mathcal{X}_{L,t}, \mathcal{X}_{R,t}}\bigr)}{\mathsf{d}\bigl(\mathcal{P}^{[t]}_{\bX\mid \mathcal{X}_{L,t}, \mathcal{X}_{R,t}}\bigr)}(\bX)\propto \frac{\mathsf{d}(\mathcal{P}^{[r]}_{\bX})}{\mathsf{d}(\mathcal{P}^{[t]}_{\bX})}
= \frac{\prod_{i\leq r}f_0(X_i)\prod_{i>r} f_1(X_i)}
       {\prod_{i\leq t}f_0(X_i)\prod_{i>t} f_1(X_i)}
= {s^\star(\bX)}^{-1}.
\]
By the Neyman–Pearson lemma \citep[Theorem~3.2.1~(ii)]{lehmann2005testing}, any test $\phi(\bX)$ that attains exact validity at level~$\alpha$ under $\Hcal_{0}^\prime$ and satisfies
\begin{equation}\label{eqn:NP_optimal_form}
    \phi(\bX)=
\begin{cases}
    1 & \text{if } {s^\star(\bX)}^{-1} > \tau_{\alpha},\\[4pt]
    0 & \text{if } {s^\star(\bX)}^{-1} < \tau_{\alpha},
\end{cases}
\end{equation}

for an appropriate threshold $\tau_{\alpha}\in \R$, is most powerful for testing $\Hcal_{0}^\prime$ against $\Hcal_{1}^\prime$.

As discussed in Section~\ref{sec:optimal_score}, the test $\phi_t(\cdot;s)=\One{p_t(s)\leq \alpha}$ controls the Type~I error exactly at level~$\alpha$ under $\Hcal_{0}^\prime$ for any score function~$s$. Therefore, to establish the optimality of $s^\star$, it suffices to show that $\phi_t(\cdot;s^\star)$ admits the form given in \eqref{eqn:NP_optimal_form}.

Define $\bX_{\pi}=\pi(\bX)$ for $\pi\sim \text{Unif}(\Pi_t)$, and let $F_{s^\star(\bX_\pi)}$ denote the conditional cumulative distribution function of $(s^\star(\bX_{\pi}))^{-1}$ given $\bX$. 
Set
\[
\tau_{\alpha}:=\inf\{y\in \R: F_{s^\star(\bX_\pi)}(y)\geq 1-\alpha\}.
\]
By the definition of $p_t$ in \eqref{eq:conformal_pvalue}, we have
\begin{align*}
    {s^\star(\bX)}^{-1} > \tau_{\alpha} &\implies p_t(s^\star)\leq \alpha,\\
    {s^\star(\bX)}^{-1} < \tau_{\alpha} &\implies p_t(s^\star)> \alpha,
\end{align*}
as required.
This completes the proof.
$\hfill\mathsf{\square}$

\subsubsection{Proof of Theorem~\ref{thm:optimal_score}}
We observe that only the $t$-th coordinate of CPP score $S_t$ determines the CONCH $p$-value $\bar{p}_t$ defined in \eqref{eq:pvalue_conch_randomized}. Therefore, with the notation laid out in Section~\ref{sec:optimal_score}, we can write
\[
    (n-1)-\E_{\Hcal_{0,\xi}\,\cap\,\Pcal_{\textnormal{IID}}}[\bar{C}_{1-\alpha}^{\conch}(S)]=\sum_{t=1}^{n-1} \E_{\Hcal_{0,\xi}\,\cap\,\Pcal_{\textnormal{IID}}}[\One{p_t(S_t)\leq \alpha}].
\]
Thus, obtaining an expression for the optimal $S_t$ boils down to the analysis of $\E_{\Hcal_{0,\xi}\,\cap\,\Pcal_{\textnormal{IID}}}[\One{p_t(S_t)\leq \alpha}]$.
Finally, note that $\Hcal_{0,\xi}\,\cap\,\Pcal_{\textnormal{IID}}=\mathcal{P}^{[\xi]}$. Hence, by tower law, we have that 
\[
\E_{\mathcal{P}^{[\xi]}}[\One{p_t(S_t)\leq \alpha}]=\E\Bigl[\,\E\bigl[\,\One{p_t(S_t)\leq \alpha}\bigm| \{X_1,\ldots,X_t\},\{X_{t+1},\ldots,X_n\}\,\bigr]\,\Bigr]
\]
Moreover, by Theorem~\ref{thm:coverage_exactly_1-alpha}, the $p$-value $\bar{p}_t=p_t(S_t)$ is valid under $\Hcal^\prime_{t}$. Hence, applying Lemma~\ref{lem:conch_NP_lemma}, the optimal form of $S^\textnormal{OPT}_t$ follows readily.

\subsection{Proving asymptotic sharpness of the 
$\conch$-SPLIT confidence set}\label{app:consistency}
In this section, we give the proof for the asymptotic sharpness of $\conch$-SPLIT confidence sets for oracle LLR score (Theorem~\ref{thm:conch_consistency_oracle}) and for learned LLR score (Theorem~\ref{thm:conch-consistency}). 

\paragraph{Notation.} Recall that $\ell(x)=\log(f_0(x)/f_1(x))$, and hence observe that
$\mathrm{KL}(P_0\|P_1)=\E_{X\sim P_0}[\ell(X)]$ and  $\mathrm{KL}(P_1\|P_0)=\E_{X\sim P_1}[-\ell(X)],$
where $\mathrm{KL}(P\|Q)$ denotes the Kullback-Leibler divergence between distributions $P$ and $Q$. Moreover, we define 
\[
\mathrm{J}(P_0,P_1)=\mathrm{KL}(P_0\|P_1)+\mathrm{KL}(P_1\|P_0),\qquad I(P_0,P_1)= \mathrm{KL}(P_0\|P_1)\wedge \mathrm{KL}(P_1\|P_0).
\]
The corresponding var-entropy measures are given by $\sigma_0^2:=\mathrm{Var}_{X\sim P_0}(\ell(X))$ and $
\sigma_1^2:=\mathrm{Var}_{X\sim P_1}(\ell(X))$.
We let $\sigma_\star$ denote $\max\{\sigma_0,\sigma_1\}$.

Recall that the actual computation of the \conch{} $p$-values are done on $\mathcal{D}_{n,2}:=(Y_1,\ldots,Y_m)$ and $m=n/2$. We write $(p_{1,m},\ldots,p_{m-1,m})$ for the corresponding \conch{} $p$-values. Let us also write $\xi'_n:=\lfloor \xi_n/2\rfloor$ to denote the changepoint of the data sequence $\mathcal{D}_{n,2}$.

Next, given an estimate $\hat{\ell}_n$, for any $k\in [n-1]$, we define the empirical averages
\[
\hat{\mu}_{k,n,L}:=\frac{1}{k}\sum_{i=1}^{k}\hat{\ell}_n(Y_i),\qquad
\hat{\mu}_{k,n,R}:=\frac{1}{m-k}\sum_{i=k+1}^{m}-\hat{\ell}_n(Y_i),\qquad
\hat{v}_{n}:=\frac{1}{m}\sum_{i=1}^{m}{\hat{\ell}_n}^{\,2}(Y_i).
\]
These quantities in turn enable an empirical upper bound on the \conch{} $p$-values. Finally, we define the quantities
\[
\Gamma_0 := \E_{\substack{\mathcal{D}^{\prime}_n, X\sim P_0,\\ X\independent \mathcal{D}_n^{\prime}}}[|\hat{\ell}_n(X)-\ell(X)|^2],\quad \Gamma_1 := \E_{\substack{\mathcal{D}^{\prime}_n X\sim P_1,\\ X\independent \mathcal{D}^{\prime}_n,}}[|\hat{\ell}_n(X)-\ell(X)|^2].
\]

The rest of the section is organized as follows:
\begin{itemize}
    \item Intuitively, the \conch{} confidence score~\eqref{score:CPP_split_learnt} builds a confidence set around the MLE estimate $\hat{\xi}_n$. Therefore, to obtain sharp \conch{} confidence sets, the initial MLE estimates must themselves be weakly consistent. In Appendix~\ref{app:mle_consistency}, we show that indeed $\hat{\xi}_n$ computed on $\mathcal{D}_{n,1}$ is at most $\mathrm{o}_P(\sqrt{n})$ away from the true changepoint $\xi'_n$ on $\mathcal{D}_{n,2}$.
    
    \item In Appendix~\ref{app:oracle_sharpness}, we then prove that the \conch{} confidence set computed with the oracle LLR score has $\mathrm{O}_P(1)$ length. This establishes the optimality and asymptotic sharpness of the oracle LLR score.
    
    \item Next, in Appendix~\ref{app:proof_of_consistency}, we extend the asymptotic sharpness result to the setting where $\ell$ is replaced by an estimate $\hat{\ell}_n$, showing that the normalized length of the corresponding \conch{} confidence set converges to $0$.

    \item Then, in Appendix~\ref{app:local_oracle_sharpness}, we extend the asymptotic sharpness results to the setting where the distributions $P_0$ and $P_1$ vary with the sample size $n$ and progressively become less distinguishable.

    \item The asymptotic sharpness results are initially established for the \conch{}-SPLIT procedure, where the full-permutation \conch{} algorithm is applied to the even subsequence. In Appendix~\ref{app:mc_sharpness_theory}, we show that the same results continue to hold when \conch{}-MC is used instead.
    
    \item Finally, in Appendix~\ref{app:lemma_consistency}, we state and prove all auxiliary lemmas needed for the preceding results.
\end{itemize}

\subsubsection{Consistency of MLE estimates}\label{app:mle_consistency}
\begin{theorem}\label{thm:mle_oracle_consistency}
    Suppose, the oracle log-likelihood ratio $\ell$ is given, and at sample size $n$, we compute the MLE $\hat{\xi}_{n}$ as defined in~\eqref{eq:split_learnt_MLE} with $\hat{\ell}_n\equiv \ell$, and that $0<\sigma_\star<\infty$. Then, it holds that
    \[|\hat{\xi}_{n}-\xi_n^\prime|=\mathrm{O}_P(1).\]
    This further implies that 
    \[
    |2\hat{\xi}_n-\xi_n|=\mathrm{O}_P(1).
    \]
\end{theorem}
\begin{proof}
Let $\xi_n''$ denote the changepoint location on the odd subsequence $(X_{1,n},X_{3,n},\ldots)$, and note that
\[
|\xi_n'-\xi_n''|\leq1.
\]
Therefore, it suffices to prove that
\[|\hat{\xi}_{n}-\xi''_n|=\mathrm{O}_P(1).\]
We start with recalling that
\[
\hat{\xi}_{n}\in \argmax_{s\in [m-1]} L(s),\qquad \text{where}~L(s)=\sum_{i=1}^s \ell(X_{2i-1,n}).
\]
For any $t>\xi''_{n}$, we can write $L(t)=L(\xi''_n)+\sum_{s=\xi''_n+1}^t \ell(X_{2s-1,n})$.
By a union bound, for any $M\in \N$,
\begin{align*}
    \P(\,\hat{\xi}_{n}\ge \xi''_n+M)
    &= \P\bigl(\cup_{t\,\ge\,\xi''_n+M}\{L(t)\ge L(\xi''_n)\}\bigr)\\
    &\le \sum_{t\,\ge\,\xi''_n+M}^n \!\! \P\,(L(t)\ge L(\xi''_n))= \sum_{t\,\ge\,\xi''_n+M}^n \!\! \P\Bigl(\!\sum_{s=\xi''_n+1}^t \ell(X_{2s-1,n})\ge 0\Bigr).
\end{align*}
Observe that $\{X_{2s-1,n}\}_{s=\xi''_n+1}^t$ are i.i.d samples from $P_1$. Therefore, by Lemma~\ref{lem:negative_drift_partial_sums}, there exists $\gamma>0$ such that
\begin{align*}
    \P(\,\hat{\xi}_{n}\ge \xi''_n+M)\le \sum_{t\,\ge\,\xi''_n+M}^n e^{-\gamma\, (t-\xi''_n)}\le C_0\,C_1^M,
\end{align*}
for some $C_0>0$ and $C_1<1$. Similarly, we can show that there exists $C'_0>0$ and $C'_1<1$ such that $\P(\,\hat{\xi}_{n}\le \xi''_n-M)\le C'_0\, {C'_1}^M$. Together, this yields $|\hat{\xi}_{n}-\xi''_n|=\mathrm{O}_P(1)$, and proves the first part. Finally, recalling that $\xi'_n:=\lfloor \xi_n/2\rfloor$, the second claim is immediate.
\end{proof}

\begin{theorem}\label{thm:mle_general_consistency}
Suppose, the estimated log-likelihood ratio $\hat{\ell}_n$ satisfies \eqref{eqn:consistency_of_hat_ell_n}. Then,
 $\hat{\xi}_n$, defined in~\eqref{eq:split_learnt_MLE} satisfies
    \[|\hat{\xi}_{n}-\xi'_n|=\mathrm{o}_P(\sqrt{n}).\]
    Consequently, we have $|2\hat{\xi}_n-\xi_n|=\mathrm{o}_P(\sqrt{n})$.
\end{theorem}
\begin{proof}
As in the earlier proof, let $\xi_n''$ denote the changepoint location on the odd subsequence $(X_{1,n},X_{3,n},\ldots)$, and note that
\[
|\xi_n'-\xi_n''|\leq1.
\]
Therefore, it suffices to prove that
\[|\hat{\xi}_{n}-\xi''_n|=\mathrm{o}_P(\sqrt{n}).\]
We start by defining 
\begin{align*}
    S_k^{(1)}:=\sum_{s=\xi''_n+1}^{\xi''_n+k} \hat{\ell}_n(X_{2s-1,n})-\ell(X_{2s-1,n}),\qquad
    S_k^{(0)}:=\sum_{s=1}^{k} \hat{\ell}_n(X_{2s-1,n})-\ell(X_{2s-1,n}).
\end{align*}
Next, fix $\delta>0$ and let $\mathcal{A}$ be the event that
\[
    \max_{k=1,\ldots,m-\xi''_n} |S_k^{(1)}|\le \bigl(k+\sqrt{m}\bigr)\cdot \Gamma_1^{1/2}/\delta,\qquad
    \text{and}\quad \max_{k=1,\ldots,\xi''_n} |S_k^{(0)}|\le \bigl(k+\sqrt{m}\bigr)\cdot \Gamma_0^{1/2}/\delta.
\]
By Lemma~\ref{lem:L2_error_ell_n}, note that $P(\mathcal{A}^c)\le 4\delta$.
Now, recall that
\[
\hat{\xi}_{n}\in \argmax_{s\in [m-1]} \hat{L}(s),\qquad \hat{L}(s):=\sum_{i=1}^s \hat{\ell}_n(X_{2i-1,n}).
\]
For any $t>\xi''_{n}$, write $\hat{L}(t)=\hat{L}(\xi''_n)+\sum_{s=\xi''_n+1}^t \hat{\ell}_n(X_{2s-1,n})$.
Take any $\varepsilon>0$. By a union bound, 
\begin{align*}
    &\P\,(\,\hat{\xi}_{n}\ge \xi''_n+\varepsilon\sqrt{m})
    \le \P\,\bigl(\cup_{t\,\ge\,\xi''_n+\varepsilon\sqrt{m}}\{\hat{L}(t)\ge \hat{L}(\xi''_n)\},\,\mathcal{A}\bigr)+P(\mathcal{A}^c)\\
    &\hspace{3cm}\le \!\!\sum_{t\,\ge\,\xi''_n+\varepsilon\sqrt{m}}^m \P\!(\hat{L}(t)\ge \hat{L}(\xi''_n),\,\mathcal{A})+4\delta= \sum_{t\,\ge\,\xi''_n+\varepsilon\sqrt{n}}^n \!\!\P\,\Bigl(\!\!\sum_{s=\xi''_n+1}^t \hat{\ell}_n(X_{2s-1,n})\ge 0,\,\mathcal{A}\Bigr)+4\delta.
\end{align*}
Recall that $\{X_{2s-1,n}\}_{s=\xi''_n+1}^t$ are i.i.d samples from $P_1$. Therefore, for any $t\ge \xi''_n+\varepsilon\sqrt{m}$, 
\begin{align*}
   \P\Bigl(\,\sum_{s=\xi''_n+1}^t \hat{\ell}_n(X_{2s-1,n})\ge 0,\,\mathcal{A}\Bigr)&\le \P\Bigl(\,\sum_{s=\xi''_n+1}^t \ell(X_{2s-1,n})\ge -|S^{(1)}_{t-\xi''_n}|,\,\mathcal{A}\Bigr)\\
   &\le \P\Bigl(\,\frac{1}{t-\xi''_n}\sum_{s=\xi''_n+1}^t \ell(X_{2s-1,n})\ge -\bigl(\frac{1}{\delta}+\frac{1}{\varepsilon\delta}\bigr)\cdot \Gamma_1^{1/2}\Bigr).
\end{align*}
By~\eqref{eqn:consistency_of_hat_ell_n}, $\Gamma_1\to 0$ as $n\to\infty$. Hence, for sufficiently large $n$, by Lemma~\ref{lem:negative_drift_partial_sums}, there exists $\gamma>0$ such that for any $t\ge \xi''_n+\varepsilon\sqrt{m}$,
\[
\P\Bigl(\,\sum_{s=\xi''_n+1}^t \hat{\ell}_n(X_{2s-1,n})\ge 0,\,\mathcal{A}\Bigr)\le e^{-\gamma\, (t-\xi''_n)}\le e^{-\gamma\, \varepsilon\sqrt{m}}
\]
Hence, it follows that for sufficiently large $n$,
\begin{align*}
    \P(\,\hat{\xi}_{n}\ge \xi''_n+\varepsilon\sqrt{m})\le n\,e^{-\gamma\, \varepsilon\sqrt{m}}+4\delta.
\end{align*}
 Similarly, we can show that 
 \[\P(\,\hat{\xi}_{n}\le \xi''_n-\varepsilon\sqrt{m})\le n\,e^{-\gamma\, \varepsilon\sqrt{m}}+4\delta.\]
Together, since $\varepsilon$ and $\delta$ are arbitrary, this yields $|\hat{\xi}_{n}-\xi''_n|=\mathrm{o}_P(\sqrt{n})$, as required. Lastly, recalling that $\xi'_n:=\lfloor \xi_n/2\rfloor$, the second result follows.
\end{proof}

\subsubsection{Sharpness of \conch-SPLIT with oracle LLR score}\label{app:oracle_sharpness}

We now show that, as the sample size grows, the \conch-SPLIT confidence set, evaluated with the oracle LLR score, concentrates within a small neighborhood of the true changepoint $\xi_n$, enabling sharp localization. We consider the setting in Section~\ref{sec:sharpness_learned_llr}, and assume that we have access to the oracle likelihood ratio $\ell$. We write $(p_{1,n},\ldots,p_{m,n})$ for the \conch{} $p$-values defined on $\mathcal{D}_{n,2}$ based on the oracle LLR score, that is \eqref{score:CPP_split_learnt} with $\hat{\ell}_n\equiv \ell$.

\begin{theorem}[Sharpness with oracle LLR score] \label{thm:conch_consistency_oracle}
Suppose that, in the above setting, there exists $\tau\in(0,1)$ such that $\xi_n/n\,\to\,\tau$ as $n\to\infty$, and $\mathrm{Var}_{X\sim P_0}(\ell(X)),\mathrm{Var}_{X\sim P_1}(\ell(X))\in (0,\infty)$. Then, for \conch{} $p$-values $(p_{1,n},\ldots,p_{m-1,n})$ defined above, we have for every $\delta>0$,
\[
\lim_{\kappa\to\infty}
\limsup_{n\to\infty}
\P\left(
    \max_{|t-\xi'_n|\ge\kappa}
    p_{t,n}>\delta\right)=0.
\]
Consequently, it follows that $|\mathcal{C}^{\conch\text{-SPLIT}}_{n,1-\alpha}|=\mathrm{O}_P(1)$, meaning that it is asymptotically sharp i.e., $\bigl|\,\mathcal{C}^{\conch\text{-SPLIT}}_{n,1-\alpha}\,\bigr|/(n-1)\ \xrightarrow{P}\ 0$ as $n\to\infty$.
\end{theorem}

\begin{proof}
We start by observing that, by symmetry, it suffices to show that for some $\kappa>0$, 
\[
\lim_{\kappa\to\infty}
\limsup_{n\to\infty} \P\left(
    \max_{t\ge \xi'_n +\kappa}
    p_{t,n}>\delta\right)= 0.
\]
The proof is now split into three key steps. First, in Step~1, we restrict to a high-probability, non-growing set around the changepoint $\xi'_n$. Next, in Step~2, we use a data-dependent bound on the $p$-values to reduce the problem to deriving the concentration of a few empirical terms. Finally, in Step~3, we show that one of these key empirical terms has a negative drift uniformly over all indices $t\ge\xi'_n+\kappa$ with high probability and complete the proof. 

\paragraph{Step~1:  Restrict to a high-probability set for $\hat{\xi}_n$.}

Fix any $\eta,\delta\in (0,1)$, and note that it is enough to show that there exists a $\kappa>0$ such that
\[
\lim_{n\to\infty}\,\P\left(\max_{\,t\,\ge\,\xi'_n+\kappa}\ p_{t,n}> \delta\right)\leq 5\eta.
\]
For the oracle LLR score~\ref{score:oracle_llr}, we have $\hat{\ell}_n\equiv \ell$ for all $n\in\N$. Consequently, $\Gamma_0=0$, $\Gamma_1=0$ and $\hat{\xi}_{n}$ is as in~\eqref{eq:split_learnt_MLE}. By Theorem~\ref{thm:mle_oracle_consistency}, $\hat{\xi}_{n}$ satisfies $|\hat{\xi}_{n}-\xi'_n|=\mathrm{O}_P(1)$. Since $\xi_n/n\to\tau$ and $\xi'_n=\lfloor \xi_n/2\rfloor$, there exist $C_0>0$ and $N_0\in\N$ such that for all $n>N_0$,
\begin{equation*}\label{eqn:control_on_xi_hat_n}
    \P\left(|\hat{\xi}_n-\xi'_n|>C_0\right)\le \eta,
    \qquad 
    \frac{\xi'_n}{n}\ge \frac{\tau}{4},\qquad \frac{\xi'_n-C_0}{n}\ge \frac{\tau}{8}.
\end{equation*}
Now suppose $n>N_0$ and $\kappa>C_0$. By the law of total probability,
\begin{align}\label{eq:total_prob_bound}
    \P\left(\max_{\,t\,\ge\,\xi'_n+\kappa}\ p_{t,n}> \delta\right)\ &\le \ \P\left(\max_{\,t\,\ge\,\xi'_n+\kappa}\ p_{t,n}> \delta,\ |\hat{\xi}_{n}-\xi'_n|\le C_0\right)+ \eta\notag\\
 &\le \ \sum_{k_n=\xi'_n-C_0}^{\xi'_n+C_0}\P\left(\max_{\,t\,\ge\,\xi'_n+\kappa}\ p_{t,n}> \delta,\ \hat{\xi}_{n}=k_n\right)+ \eta.
\end{align}
Hence, it suffices to upper bound each of the summands in~\eqref{eq:total_prob_bound}. 

\paragraph{Step~2: Data-dependent upper bound on $p$-values, and its implication.}
Fix any $k_n$ such that $\xi'_n-C_0\le k_n\le \xi'_n+C_0$. By Lemma~\ref{lem:deterministic_bound_onn_tildep}, for any $t\in[m-1]$,
\[
p_{t,n}\le \frac{m}{t}\,\frac{\hat{v}_n}{\hat{m}_n\, \Delta^2_{t,\hat{\xi}_n}}\one{\Delta_{t,\hat{\xi}_n}<0} + \one{\Delta_{t,\hat{\xi}_n}\ge 0},
\]
where
\[
\Delta_{t,j}:=\frac{1}{\hat{m}_n}\sum_{i=j+1}^{t}\hat{\ell}_n(Y_i)-\hat{\mu}_{t,n,L},\qquad \hat{m}_n:=|t-\hat{\xi}_n|.
\]
Since $\kappa>C_0$, on the event $\{\hat{\xi}_{n}=k_n\}$ we have
\[
\hat{m}_n=t-k_n\ge(\kappa-C_0),\qquad \frac{m}{t}\le \frac{m}{\xi'_n}\le \frac{2}{\tau}.
\]
Therefore, on $\{\hat{\xi}_{n}=k_n\}$, the empirical bound on $p_{t,n}$ (uniform over $t\ge\xi'_n+\kappa$) becomes
\[
    \max_{\,t\,\ge\,\xi'_n+\kappa} p_{t,n}
    \le \max_{\,t\,\ge\,\xi'_n+\kappa}\left\{
    \frac{2}{\tau}\,\frac{\hat{v}_n}{(\kappa-C_0)\, \Delta^2_{t,k_n}}\one{\Delta_{t,k_n}<0}\;+\;\one{\Delta_{t,k_n}\ge 0}\right\}.
\]
Furthermore, since $\epsilon_{2,n}=0$ for all $n\in\N$, Lemma~\ref{lem:upper_bound_on_var} implies that with probability at least $1-\eta$,
\[
    \hat{v}_n\le \frac{4\bigl(\sigma^2_\star+\mathrm{J}(P_0,P_1)^2\bigr)}{\eta}.
\]
Therefore, using~\eqref{eq:total_prob_bound} and the fact that $\delta\in(0,1)$, a union bound yields
\begin{align}\label{eqn:negative_drift_oracle_req}
    &\P\left(\max_{\,t\,\ge\,\xi'_n+\kappa}\ p_{t,n}>\delta\right)\notag\\ 
    &\hspace{0.15cm}\le \ \sum_{k_n=\xi'_n-C_0}^{\xi'_n+C_0}
 \P\left(\max_{\,t\,\ge\,\xi'_n+\kappa}\left\{\frac{1}{\tau^2}\,
 \frac{4(\sigma^2_\star+\mathrm{J}(P_0,P_1)^2)}{\eta(\kappa-C_0)\,\Delta^2_{t,k_n}}\one{\Delta_{t,k_n}<0}\;+\;\one{\Delta_{t,k_n}\ge 0}\right\}> \delta,\ 
 \hat{\xi}_{n}=k_n\right)+ 2\eta\notag\\
 &\hspace{0.15cm}\le \ \sum_{k_n=\xi'_n-C_0}^{\xi'_n+C_0}
 \P \left(\max_{\,t\,\ge\,\xi'_n+\kappa}
 \Delta_{t,k_n}> -\left(
 \frac{4(\sigma^2_\star+\mathrm{J}(P_0,P_1)^2)}{\tau^2\delta\eta(\kappa-C_0)}\right)^{1/2}\right)+ 2\eta.
\end{align}
Hence, we only need to argue that with high probability, $\Delta_{t,k_n}$ has a negative drift uniformly over $t\ge\xi'_n+\kappa$ for every $\xi'_n-C_0\le k_n\le \xi'_n+C_0$. 

\paragraph{Step~3: Proving uniform negative drift of $\Delta_{t,k_n}$, and completing the proof.}
Fix any $k_n$ such that $\xi'_n-C_0\le k_n\le \xi'_n+C_0$. We can write
\[
\Delta_{t,k_n}=\frac{1}{t-k_n}\sum_{i=k_n+1}^{t}\ell(Y_i)
-\hat{\mu}_{t,n,L}=\sum_{i=1}^{t}a_i\,\ell(Y_i),
\]
where 
\[
a_i:=-\frac{1}{t}\one{i\le k_n}+\left(\frac{1}{t-k_n}-\frac{1}{t}\right)\one{i>k_n}.
\]
For such $k_n$, we may decompose $\Delta_{t,k_n}$ as
\[
\Delta_{t,k_n}=-\frac{1}{t}\sum_{i=1}^{\xi'_n-C_0-1} \ell(Y_i)\;+\;\sum_{i=\xi'_n-C_0}^{\xi'_n+C_0} a_i\ell(Y_i)\;+\;\frac{k_n}{(t-k_n)t}\sum_{i=\xi'_n+C_0+1}^{t} \ell(Y_i).
\]
Now, we define 
\begin{multline*}
    S_1:=-\frac{1}{t}\sum_{i=1}^{\xi'_n-C_0-1} \ell(Y_i),\qquad S_2:=\sum_{i=\xi'_n-C_0}^{\xi'_n+C_0} \frac{1}{\kappa-C_0}|\ell(Y_i)|,\\
    \text{and}\quad S_3:=\max_{\,t\,\ge\,\xi'_n+\kappa}
    \frac{k_n}{(t-k_n)t}\sum_{i=\xi'_n+C_0+1}^{t} \ell(Y_i).\hspace{5em}
\end{multline*}
Noting that $|a_i|\le \frac{1}{t-k_n}\le \frac{1}{\kappa-C_0}$, we obtain
\begin{equation}\label{eq:decomposition_Delta}
    \max_{\,t\,\ge\,\xi'_n+\kappa}\Delta_{t,k_n}
    \le \;S_1\;+\;S_2\;+\;S_3.
\end{equation}
Let $N_1\in N$ such that $N_1>N_0$ and $N_1> 2C_0/\tau^2$.
Since $Y_1,\ldots,Y_{\xi'_n-C_0}\stackrel{iid}{\sim} P_0$ and $Y_{\xi'_n+C_0},\ldots,Y_{t}\stackrel{iid}{\sim} P_1$, Lemma~\ref{lem:negative_drift_partial_sums} yields constants $c,\gamma>0$ such that
\begin{align*}
    &\P\left(\frac{1}{\xi'_n-C_0}\sum_{i=1}^{\xi'_n-C_0-1} \ell(Y_i)\ge c\,I(P_0,P_1) \right)\ge 1-e^{-\gamma(\xi'_n-C_0)},\\[2pt]
    &\P\left(\frac{1}{t-\xi'_n-C_0}\sum_{i=\xi'_n+C_0+1}^{t} \ell(Y_i)\le -c\,I(P_0,P_1)\right)\ge 1-e^{-\gamma(t-\xi'_n-C_0)}.
\end{align*}
These imply high-probability bounds on $S_1$ and $S_3$.  
For $S_1$, since $t\le n$, with probability $1-e^{-\gamma(\xi'_n-C_0)}$, for $n>N_1$,
\[
S_1\le -\frac{\xi'_n-C_0}{n}c\,I(P_0,P_1) \le -\frac{\tau}{8} c\,I(P_0,P_1).
\]
For $S_3$, with probability at least $1-e^{-\gamma(t-\xi'_n-C_0)}$,
\[
\frac{k_n}{(t-k_n)t}\sum_{i=\xi'_n+C_0+1}^{t} \ell(Y_i)\le -\frac{k_n(t-\xi'_n-C_0)}{(t-k_n)t}\, cI(P_0,P_1)\le 0.
\]
Hence,
\[
\P(S_3\le 0)\ge 1-\sum_{t>\xi'_n+\kappa}e^{-\gamma(t-\xi'
_n-C_0)}\ge 1-\frac{e^{-\gamma(\kappa-C_0)}}{1-e^{-\gamma}}.
\]
Finally, by Markov's inequality and Cauchy--Schwarz, with probability at least $1-\eta/C_0$,
\begin{align*}
    S_{2}\le \frac{C_0}{\eta(\kappa-C_0)}\cdot \sum_{i=\xi'_n-C_0}^{\xi'_n+C_0} \E[|\ell(Y_i)|]&\le \frac{C_0}{\eta(\kappa-C_0)}\cdot \sum_{i=\xi'_n-C_0}^{\xi'_n+C_0} (\E[|\ell(Y_i)|^2])^{1/2}\\
    &\le\frac{C_0}{\eta(\kappa-C_0)}\sum_{i=\xi'_n-C_0}^{\xi'_n+C_0} (\sigma_\star^2+\mathrm{J}^2(P_0,P_1))^{1/2}\\
    &\le \frac{3C_0^2(\sigma_\star^2+\mathrm{J}^2(P_0,P_1))^{1/2}}{(\kappa-C_0)\eta},
\end{align*}
Combining the bounds on $S_1$, $S_2$, and $S_3$, if
\begin{align}\label{eqn:condn_1_on_kappa}
    -\frac{\tau}{8}\,cI(P_0,P_1)+
    \frac{3C_0^2(\sigma_\star^2+\mathrm{J}^2(P_0,P_1))^{1/2}}{(\kappa-C_0)\eta}\le -\left(\frac{8(\sigma^2_\star+\mathrm{J}(P_0,P_1)^2)}{\tau\delta\eta(\kappa-C_0)}\right)^{1/2},
\end{align}
then it follows that
\[
\max_{\,t\,\ge\,\xi'_n+\kappa}\Delta_{t,k_n}
\le S_1+S_2+S_3\le -\left(\frac{8(\sigma^2_\star+\mathrm{J}(P_0,P_1)^2)}{\tau\delta\eta(\kappa-C_0)}\right)^{1/2}.
\]
Hence, using~\eqref{eqn:negative_drift_oracle_req} and a union bound,
\begin{align*}
    \P\left(\max_{\,t\,\ge\,\xi'_n+\kappa}\ p_{t,n}> \delta\right)
    &\le \sum_{k_n=\xi'_n-C_0}^{\xi'_n+C_0}
    \P\left(\max_{\,t\,\ge\,\xi'_n+\kappa}\Delta_{t,k_n}
    >-\left(
    \frac{8(\sigma^2_\star+\mathrm{J}(P_0,P_1)^2)}{\tau\delta\eta(\kappa-C_0)}\right)^{1/2}\right)+2\eta\\
    &\le \sum_{k_n=\xi'_n-C_0}^{\xi'_n+C_0}\left(e^{-\gamma(\xi'_n-C_0)}
    + \frac{e^{-\gamma(\kappa-C_0)}}{1-e^{-\gamma}}+\frac{\eta}{C_0}\right)
    +2\eta \\
    &\le (2C_0+1)\, e^{-\gamma(n\tau^2-C_0)}+ 
    \frac{2C_0+1}{1-e^{-\gamma}}\,e^{-\gamma(\kappa-C_0)}+4\eta,
\end{align*}
where the last step uses $\xi'_n\ge \tau n/4$.  
Now choose $\kappa$ large enough (independent on $n$) so that
\[
\eqref{eqn:condn_1_on_kappa}\,\, \text{holds},
\qquad\text{and}\qquad
\frac{(2C_0+1)}{1-e^{-\gamma}}\,e^{-\gamma(\kappa-C_0)}\le \eta.
\]
For such a choice of $\kappa$, 
\[
\lim_{n\to\infty}\P\!\left(\max_{\,t\,\ge\,\xi'_n+\kappa}\ p_{t,n}> \delta\right)\le 5\eta,
\]
as required. Next, write 
\[
\mathcal C^{(2),\conch}_{m,1-\alpha}
:=\{t\in[m-1]:p_{t,n}>\alpha\}
\]
for the \conch{} confidence sets defined on the even subsequence $\mathcal{D}_{n,2}$. By the first part of the result, we obtain 
\[
|\mathcal C^{(2),\conch}_{m,1-\alpha}|=\mathrm{O}_P(1).
\]
Therefore, the second conclusion follows by~\eqref{eqn:conch_split_CI}.
\end{proof}
The theorem establishes that the \conch-SPLIT procedure with the sample-split oracle LLR score is rate-optimal in
cardinality, since a valid confidence set containing the true changepoint cannot have cardinality converging to zero in probability while maintaining nontrivial asymptotic coverage.

\subsubsection{Asymptotic sharpness of \conch{} with learned LLR score}\label{app:proof_of_consistency}
We continue with the notation introduced in Appendix~\ref{app:oracle_sharpness} and now prove asymptotic sharpness for the learned LLR scores.
\begin{theorem}\label{thm:conch_consistency_non_uniform}
In the setting of Theorem~\ref{thm:conch-consistency}, the $p$-values satisfy
\[
\E\left[\frac{1}{m-1}\sum_{i=1}^{m-1} p_{i,n}\right]\longrightarrow 0\qquad \text{as } n\to\infty,
\]
where recall that $m=n/2$.

\end{theorem}

\begin{proof}
We start by noting that it suffices to prove that
\[
\limsup_{n\to\infty}
\E\left[\frac{1}{m-1}\sum_{i=1}^{m-1}p_{i,n}\right]
\le 9\eta
\]
for every $\eta\in(0,1)$. Fix any $\eta\in(0,1)$. The proof is
split into four steps. We first separate the indices that are far
from the changepoint $\xi_n'$ from those lying in a $\sqrt n$
neighborhood. The latter contribute negligibly to the normalized
sum. In Step~2, we obtain a data-dependent upper bound on the
$p$-values corresponding to the far indices. In Step~3, we show
that the key empirical term in this bound has a negative drift.
Finally, in Step~4, we combine these bounds to complete the proof.

\paragraph{Step~1: Split the indices into a $\sqrt n$-neighborhood
of $\xi_n'$ and its complement.}

Let $\mathcal F_n:=\sigma(\mathcal D_{n,1},\mathcal D_n')$ be the sigma-algebra generated by $\mathcal D_{n,1}$ and $\mathcal D_n'$. Conditional on $\mathcal F_n$, both $\hat\ell_n$ and $\hat\xi_n$ are fixed, and the calibration observations $Y_1,\ldots,Y_m$ are independent of $\mathcal F_n$.

Let us define the event
\[
\mathcal G_n:=\bigcap_{j\in\{0,1\}}\left\{
\E_{X\sim P_j}\left[|\hat\ell_n(X)-\ell(X)|^2\,\middle|\,\mathcal D_n'\right]
\le \frac{\Gamma_j}{\eta}\right\},
\]
and note that by Markov's inequality and a union bound, 
\begin{equation}\label{eq:good_learning_event}
\P(\mathcal G_n^c)\le 2\eta.
\end{equation}
Moreover, by Theorem~\ref{thm:mle_general_consistency}, $|\hat\xi_n-\xi_n'|=\mathrm{o}_P(n^{1/2})$.
Further, by~\eqref{eqn:consistency_of_hat_ell_n}, and recalling $\xi_n'/m\to\tau$, there exists $N_0\in\N$ such that for all $n\ge N_0$,
\begin{equation}\label{eq:large_n_conditions_learned}
\frac{\xi_n'}{m}\ge\frac{\tau}{2},
\qquad \P\bigl(|\hat\xi_n-\xi_n'|> n^{1/2}\bigr)\le\eta,
\qquad \frac{\Gamma_0} {\eta\,\E_{X\sim P_0}[\ell^2(X)]},\frac{\Gamma_1}{\eta\,\E_{X\sim P_1}[\ell^2(X)]}\le\frac14,
\end{equation}
and that
\begin{equation}\label{eq:conditional_learning_error_small}
\frac{16}{\tau I(P_0,P_1)}
\left(\frac{ \max_{j\in\{0,1\}} \Gamma_j}{\eta}\right)^{1/2}\le\eta.
\end{equation}
We fix $\kappa$ sufficiently large so that $\frac{1}{\kappa-1}\le\frac{\tau}{4}$.

We partition $[m-1]$ into the far indices
\[
\mathcal I:=\left\{
i\in[m-1]:|i-\xi_n'|\ge\kappa n^{1/2}\right\}
\]
and the near indices $[m-1]\setminus\mathcal I$. Since
$p_{i,n}\le1$, we have
\begin{align}
\E\left[\frac{1}{m-1}\sum_{i=1}^{m-1}p_{i,n}\right]
&\le\E\left[\frac{1}{m-1}\sum_{i\in\mathcal I}p_{i,n}\right]+\frac{2\kappa n^{1/2}+1}{m-1}\notag\\
&\le\E\left[\frac{1}{m-1}\sum_{i\in\mathcal I}
p_{i,n}\One{|\hat\xi_n-\xi_n'|\le n^{1/2}}\right]
+\eta+\frac{2\kappa n^{1/2}+1}{m-1}.
\label{eqn:prob_decomposition}
\end{align}
It therefore remains to control the first term on the
right-hand side.

\paragraph{Step~2: Data-dependent upper bound for the far-index
$p$-values.}

By Lemma~\ref{lem:deterministic_bound_onn_tildep}, we obtain
\begin{equation}\label{eqn:deterministic_upper_bound_learned}
p_{t,n}\le\frac{m}{t}\frac{\hat v_n}{\hat m_n\Delta_{t,\hat\xi_n}^2}\one{\Delta_{t,\hat\xi_n}<0}
+\one{\Delta_{t,\hat\xi_n}\ge0},
\end{equation}
where we define
\[
\Delta_{t,\hat\xi_n}=\frac{1}{\hat m_n}\sum_{i=\hat\xi_n+1}^t\hat\ell_n(Y_i)-\hat\mu_{t,n,L},\qquad\hat m_n=|t-\hat\xi_n|.
\]

Fix any $t\in\mathcal I$. Without loss of generality, suppose
$t>\xi_n'$; the case $t<\xi_n'$ is symmetric. On the event $\mathcal H_n:=\{|\hat\xi_n-\xi_n'|\le n^{1/2}\}$,
we have $\hat m_n=t-\hat\xi_n\ge(\kappa-1)n^{1/2}$.
Moreover, for all sufficiently large $n$,
\[
\frac{m}{t}\le\frac{m}{\xi_n'}\le\frac{2}{\tau}.
\]
Thus, it follows that
\[
\frac{m}{t\hat m_n}
\le
\frac{2}
{\tau(\kappa-1)n^{1/2}}
\le\frac{1}{2n^{1/2}}.
\]
Consequently, on $\mathcal H_n$,
\begin{equation}\label{eq:data_dependent_bound_learned}
p_{t,n}
\le
\frac{1}{2n^{1/2}}
\frac{\hat v_n}{\Delta_{t,\hat\xi_n}^2}
\one{\Delta_{t,\hat\xi_n}<0}+\one{\Delta_{t,\hat\xi_n}\ge0}.
\end{equation}

\paragraph{Step~3: Negative drift of
$\Delta_{t,\hat\xi_n}$.}

Throughout this step, we work conditionally on $\mathcal F_n$
and on the event $\mathcal H_n$. Under this conditioning,
$\hat\ell_n$, $\hat\xi_n$, and the coefficients $(a_i)$, defined below, are fixed.
Write
\[
\Delta_{t,\hat\xi_n}=\sum_{i=1}^t a_i\hat\ell_n(Y_i)=
Z_{t,n}+R_{t,n},
\]
where
\[
Z_{t,n}:=\sum_{i=1}^t a_i\ell(Y_i),\qquad R_{t,n}:=\sum_{i=1}^t a_i\{\hat\ell_n(Y_i)-\ell(Y_i)\},
\]
and 
\[
a_i=-\frac1t\one{i\le\hat\xi_n}+\left(\frac{1}{t-\hat\xi_n}-\frac1t \right) \one{i>\hat\xi_n}.
\]
We first control the oracle component $Z_{t,n}$. Since $\E_{P_0}[\ell(X)]=\mathrm{KL}(P_0\|P_1)$ and $\E_{P_1}[\ell(X)]=-\mathrm{KL}(P_1\|P_0)$, 
a direct calculation gives
\[
\E[Z_{t,n}\mid\mathcal F_n]=\begin{cases}-\dfrac{\xi_n'}{t}\,\mathrm J(P_0,P_1),
&\hat\xi_n>\xi_n',
\\[7pt]-\dfrac{\xi_n'}{t}\,\mathrm J(P_0,P_1)
+\dfrac{(\hat m_n-m_n)_+}{\hat m_n}\,\mathrm J(P_0,P_1),
&\hat\xi_n\le\xi_n',\end{cases}
\]
where $m_n=t-\xi_n'$ and $\hat m_n=t-\hat\xi_n$.
On $\mathcal H_n$, further,
\[
\frac{(\hat m_n-m_n)_+}{\hat m_n}=\frac{(\xi_n'-\hat\xi_n)_+}{t-\hat\xi_n}\le\frac{1}{\kappa-1}.
\]
Also, since $t\le m$, for all sufficiently large $n$,
\[
\frac{\xi_n'}{t}
\ge\frac{\xi_n'}{m}
\ge\frac{\tau}{2}.
\]
Together, therefore, we have
\begin{align}
\E[Z_{t,n}\mid\mathcal F_n]&\le-\left(\frac{\tau}{2}-\frac{1}{\kappa-1}\right)\mathrm J(P_0,P_1)\le
-\frac{\tau}{4}\mathrm J(P_0,P_1)\le-\frac{\tau}{2}I(P_0,P_1).
\label{eq:oracle_mean_drift_learned}
\end{align}
Moreover, note that $\mathrm{Var}(Z_{t,n}\mid\mathcal F_n)\le\sigma_\star^2\sum_{i=1}^t a_i^2$. Since
\[
\sum_{i=1}^t a_i^2=\frac{\hat\xi_n}{t^2}
+\frac{\hat\xi_n^2}{\hat m_n t^2}\le\frac1t+\frac1{\hat m_n}\le\frac{2}{(\kappa-1)n^{1/2}},
\]
we obtain
\[
\mathrm{Var}(Z_{t,n}\mid\mathcal F_n)
\le\frac{2\sigma_\star^2}{(\kappa-1)n^{1/2}}.
\]
Therefore, by Chebyshev's inequality, after suitably increasing $N_0$, on $\mathcal H_n$,
\begin{equation}\label{eq:oracle_deviation_learned}
\P\left(Z_{t,n}>-\frac{\tau}{4}I(P_0,P_1)
\,\middle|\,\mathcal F_n\right)\le\eta.
\end{equation}
We next control the error component $R_{t,n}$. Observe that $\sum_{i=1}^t|a_i|\le 2$.
Hence, on $\mathcal G_n$,
\begin{align*}
\E[|R_{t,n}|\mid\mathcal F_n]&\le 2\max_{j\in\{0,1\}}\E_{X\sim P_j}
\left[|\hat\ell_n(X)-\ell(X)|\,\middle|\,\mathcal D_n'\right]\\
&\le 2\max_{j\in\{0,1\}}\left(\E_{X\sim P_j}
\left[|\hat\ell_n(X)-\ell(X)|^2\,\middle|\,\mathcal D_n'
\right]\right)^{1/2}\le
2\max_{j\in\{0,1\}}\left(\frac{\Gamma_j}{\eta}\right)^{1/2}.
\end{align*}
Thus, by Markov's inequality and
\eqref{eq:conditional_learning_error_small}, on $\mathcal G_n$,
\begin{equation}\label{eq:learned_error_bound_conditional}
\P\left(|R_{t,n}|>\frac{\tau}{8}I(P_0,P_1)\,\middle|\,
\mathcal F_n\right)\le\eta.
\end{equation}
Combining~\eqref{eq:oracle_deviation_learned} and
\eqref{eq:learned_error_bound_conditional}, we obtain
\[
\P\left(\Delta_{t,\hat\xi_n}>-\frac{\tau}{8}I(P_0,P_1),
\,\mathcal H_n\cap\mathcal G_n\right)\le2\eta.
\]
Together with~\eqref{eq:good_learning_event}, this gives
\begin{equation}\label{claim:neg_upper_bound_on_Delta}
\P\left(\Delta_{t,\hat\xi_n}>-\frac{\tau}{8}I(P_0,P_1),\,
\mathcal H_n\right)\le4\eta.
\end{equation}

\paragraph{Step~4: Completing the proof.}

By Lemma~\ref{lem:upper_bound_on_var} and
\eqref{eq:large_n_conditions_learned},
\begin{equation}\label{eq:vhat_bound_learned}
\P\left(
\hat v_n>
\frac{8\{\sigma_\star^2+\mathrm J^2(P_0,P_1)\}}{\eta}
\right)
\le4\eta.
\end{equation}
Define
\[
\mathcal A_{t,n}:=\left\{\hat v_n
\le\frac{8\{\sigma_\star^2+\mathrm J^2(P_0,P_1)\}}{\eta}
\right\}\cap\left\{\Delta_{t,\hat\xi_n}
\le-\frac{\tau}{8}I(P_0,P_1)\right\}.
\]
On $\mathcal H_n\cap\mathcal A_{t,n}$,
\eqref{eq:data_dependent_bound_learned} gives
\begin{align}
p_{t,n}
&\le\frac{1}{2n^{1/2}}\frac{
8\{\sigma_\star^2+\mathrm J^2(P_0,P_1)\}/\eta}{\tau^2I^2(P_0,P_1)/64
}=\frac{256\{\sigma_\star^2+\mathrm J^2(P_0,P_1)\}}{\tau^2\eta I^2(P_0,P_1)n^{1/2}}.
\label{eq:far_pvalue_bound_learned}
\end{align}
Since $p_{t,n}\le1$, it follows from
\eqref{claim:neg_upper_bound_on_Delta} and
\eqref{eq:vhat_bound_learned} that
\begin{align}
\E\left[p_{t,n}\One{\mathcal H_n}\right]
&\le\frac{256\{\sigma_\star^2+\mathrm J^2(P_0,P_1)\}}{\tau^2\eta I^2(P_0,P_1)n^{1/2}}+\P(\mathcal A_{t,n}^c\cap\mathcal H_n)\notag\\
&\le\frac{256\{\sigma_\star^2+\mathrm J^2(P_0,P_1)\}
}{\tau^2\eta I^2(P_0,P_1)n^{1/2}}+8\eta.
\label{eq:expected_far_pvalue_bound}
\end{align}
This bound holds for every $t\in\mathcal I$. Substituting it into
\eqref{eqn:prob_decomposition} yields
\begin{align*}
\E\left[\frac{1}{m-1}\sum_{i=1}^{m-1}p_{i,n}\right]
&\le\frac{256\{\sigma_\star^2+\mathrm J^2(P_0,P_1)\}
}{\tau^2\eta I^2(P_0,P_1)n^{1/2}}+9\eta+\frac{2\kappa n^{1/2}+1}{m-1}.
\end{align*}
Since $m=n/2$, both terms depending on $n$ converge to zero.
Consequently,
\[
\limsup_{n\to\infty}\E\left[\frac{1}{m-1}\sum_{i=1}^{m-1}p_{i,n}\right]\le 9\eta.
\]
Because $\eta>0$ is arbitrary, the result follows.
\end{proof}

\paragraph{Proof of Theorem~\ref{thm:conch-consistency}.}
Let us write 
\[
\mathcal C^{(2),\conch}_{m,1-\alpha}
:=\{t\in[m-1]:p_{t,n}>\alpha\}
\]
for the \conch{} confidence sets defined on the even subsequence $\mathcal{D}_{n,2}$. Recall from Section~\ref{app:oracle_sharpness}, 
\[
|\,\mathcal{C}^{\conch\text{-SPLIT}}_{n,1-\alpha}\,|\le 2|\,\mathcal C^{(2),\conch}_{m,1-\alpha}\,|+2.
\]
Fix $\eta>0$. By Markov's inequality, we have that 
\[
\P\left(\frac{|\,\mathcal C^{(2),\conch}_{m,1-\alpha}\,|}{n-1}\ge \eta\right)\le \frac{\E\left[|\,\mathcal C^{(2),\conch}_{m,1-\alpha}\,|\right]}{(n-1)\eta}.
\]
Next, note that 
\[
\frac{1}{m-1}\,\E\left[|\,\mathcal C^{(2),\conch}_{m,1-\alpha}\,|\right]\ = \ \frac{1}{m-1}\, \sum_{i=1}^{m-1} \P\left(p_{i,n}> \alpha\right).
\]
By applying Markov's inequality again, we have
\begin{align*}
    \frac{1}{m-1}\,\E\left[|\,\mathcal C^{(2),\conch}_{m,1-\alpha}\,|\right]\le \frac{1}{\alpha}\cdot \frac{1}{(m-1)}\sum_{i=1}^{m-1} \E[p_{i,n}].
\end{align*}
Hence, by Theorem~\ref{thm:conch_consistency_non_uniform}, it follows that
\[
\P\left(\frac{|\,\mathcal C^{(2),\conch}_{m,1-\alpha}\,|}{n-1}\ge \eta\right)\longrightarrow 0,\qquad \text{as}~n\to\infty.
\]
Since $\eta$ is arbitrary, the result follows.
\hfill$\square$

\subsubsection{Asymptotic sharpness of \conch{}-SPLIT under local alternatives}
\label{app:local_oracle_sharpness}

We now extend the asymptotic sharpness results to a local-alternative setting in which the pre- and post-change distributions vary with the sample size and become progressively less distinguishable. For each $n$, suppose that
\[
(X_{1,n},\ldots,X_{\xi_n,n})\sim P_{0,n}^{\xi_n},
\qquad
(X_{\xi_n+1,n},\ldots,X_{n,n})\sim P_{1,n}^{(n-\xi_n)}.
\]
Let $f_{0,n}$ and $f_{1,n}$ denote densities with respect to a common dominating measure, and define the oracle log-likelihood ratio
\[
\ell_n(x):=\log\frac{f_{0,n}(x)}{f_{1,n}(x)}.
\]
We further write the KL measures
\[
a_n:=\mathrm{KL}(P_{1,n}\|P_{0,n}),
\qquad
b_n:=\mathrm{KL}(P_{0,n}\|P_{1,n}),
\qquad
I_n:=a_n\wedge b_n.
\]
Thus, $I_n$ quantifies the separation between the pre- and post-change distributions.
As in the preceding sections, we assume for simplicity that $n$ is even and write $m=n/2$. Moreover, we use the same interlaced split
\[
\mathcal D_{n,1}= (X_{1,n},X_{3,n},\ldots,X_{n-1,n}),
\qquad
\mathcal D_{n,2}= (Y_{1,n},\ldots,Y_{m,n}),
\]
where $Y_{i,n}:=X_{2i,n}$. The changepoint on the calibration subsequence $\mathcal{D}_{2,n}$ is given by $\xi_n':=\left\lfloor\frac{\xi_n}{2}\right\rfloor$. Using the training split, as before, we compute the oracle reference MLE
\[
\hat\xi_n
\in
\argmax_{s\in[m-1]}
\sum_{i=1}^{s}\ell_n(X_{2i-1,n}).
\]
We first consider the sharpness result for the oracle LLR score, i.e., we run \conch{} on the calibration split using the score in~\eqref{score:CPP_split_learnt} with $\hat\ell_n\equiv\ell_n$. We write $(p_{1,n},\ldots,p_{m-1,n})$ for the resulting \conch{} $p$-values.

To prove the asymptotic sharpness result, we impose the following assumptions:
\begin{enumerate}
    \item The changepoint remains in the interior of the timeline, i.e., $\frac{\xi_n}{n}\to\tau $ for some $\tau\in(0,1)$.

    \item The signal strength satisfies $I_n\longrightarrow 0$ and $ nI_n\longrightarrow\infty.$

    \item The two KL divergences are of comparable order: there exists a constant $C_I<\infty$ such that, for all sufficiently large $n$, $a_n\vee b_n\leq C_I I_n$.

    \item There exist constants $c,B,\gamma>0$ such that, for every $n$ and every $r\in\mathbb N$,
    \begin{align}\label{eq:local_maximal_drift}
    &\P_{1,n}\left(\sup_{k\geq r}\frac1k\sum_{i=1}^{k}\ell_n(X_i)\geq
    -c a_n\right)\leq Be^{-\gamma ra_n},~~\text{and}
    \\
    &\P_{0,n}\left(\inf_{k\geq r}\frac1k\sum_{i=1}^{k}\ell_n(Y_i)\leq c b_n\right)\leq Be^{-\gamma  rb_n},
    \notag
    \end{align}
    where $X_1,X_2,\ldots\overset{iid}{\sim}P_{1,n}$ and
    $Y_1,Y_2,\ldots\overset{iid}{\sim}P_{0,n}$.

    \item There exists a constant $C_{\mathrm v}<\infty$ such that, for all sufficiently large $n$,
    \begin{equation}\label{eq:local_variance_assumption}
        \mathrm{Var}_{P_{0,n}}(\ell_n(X))\vee \mathrm{Var}_{P_{1,n}}(\ell_n(X))\leq
    C_{\mathrm v}I_n.
    \end{equation}
\end{enumerate}

The maximal inequalities in~\eqref{eq:local_maximal_drift} are local-alternative analogues of Lemma~\ref{lem:negative_drift_partial_sums}, used in the fixed-distribution setting. Rest of the assumptions are either carried forward from the fixed distribution setting, or mild regularity assumption on the underlying distributions.

\begin{theorem}[Oracle sharpness under local alternatives]
\label{thm:local_oracle_sharpness}
Under the preceding assumptions, for every $\alpha,\eta\in(0,1)$, there exists $M<\infty$ such that
\[
\limsup_{n\to\infty}\,\P\biggl(\max_{\substack{t\in[m-1]\\ |t-\xi_n'|\geq M/I_n}}p_{t,n}>\alpha\biggr)\leq 10\eta.
\]
Consequently, 
the \conch{}-SPLIT confidence set satisfies $\left|\mathcal C^{\conch\text{-SPLIT}}_{n,1-\alpha}\right|=O_P(I_n^{-1})$.
In particular, since $nI_n\to\infty$, we further have
\[
\frac{\left|\mathcal C^{\conch\text{-SPLIT}}_{n,1-\alpha}\right|}{n-1}\xrightarrow{P}0.
\]
\end{theorem}

\begin{proof}
The proof closely follows that of Theorem~\ref{thm:conch_consistency_oracle} for the fixed-distribution setting. We therefore only highlight the main changes required under local alternatives.

Fix $\alpha,\eta\in(0,1)$. By symmetry, it suffices to show that, for a sufficiently large constant $M$,
\[
\limsup_{n\to\infty}\P\left(\max_{\substack{t\in[m-1]\\ t\geq \xi_n'+M/I_n}} p_{t,n}>\alpha \right) \leq 5\eta.
\]

We first establish the consistency of the reference MLE, following the proof of Theorem~\ref{thm:mle_oracle_consistency}. Let $\xi_n''$ denote the changepoint location on the odd subsequence $(X_{1,n},X_{3,n},\ldots)$, and note that
\[
|\xi_n'-\xi_n''|\leq1.
\]
For any $u>0$, by~\eqref{eq:local_maximal_drift},
\begin{align*}
\P\left(\hat\xi_n\geq\xi_n''+\frac{u}{I_n}\right)
&\leq \P_{1,n}\left(\sup_{m-\xi_n''\geq r\geq u/I_n}
\frac1r\sum_{i=1}^{r}\ell_n(X_{2(\xi_n''+i)-1,n})
\geq0\right)\\
&\leq B\exp\left( -\gamma u\frac{a_n}{I_n} \right)
\leq Be^{-\gamma u}.
\end{align*}
An analogous argument controls the left tail. Consequently,
\[
|\hat\xi_n-\xi_n''|=O_P(I_n^{-1}),\qquad\text{and hence}\qquad |\hat\xi_n-\xi_n'|=O_P(I_n^{-1}).
\]
This explains the $I_n^{-1}$ localization rate appearing in the theorem.

Choose $C_0<\infty$ such that
\[
\limsup_{n\to\infty}\P\left(|\hat\xi_n-\xi_n'|>\frac{C_0}{I_n}\right)\leq\eta,
\]
and define
\[
\mathcal E_n:=\left\{|\hat\xi_n-\xi_n'|
\leq\frac{C_0}{I_n}\right\}.
\]
Conditional on the training split $\mathcal D_{n,1}$, the reference estimator $\hat\xi_n$ is fixed and independent of the calibration split $\mathcal D_{n,2}$. Fix $t\geq\xi_n'+M/I_n$, where $M>C_0+2$. On $\mathcal E_n$,
\[
\hat m_{t,n}:=t-\hat\xi_n \geq \frac{M-C_0}{I_n}.
\]
Moreover, since $\xi_n'/m\to\tau$, there exists a constant $C_\tau<\infty$ such that $m/t\leq C_\tau$ for all sufficiently large $n$.
For $t>\hat\xi_n$, define
\[
\Delta_{t,\hat\xi_n} := \frac{1}{t-\hat\xi_n}
\sum_{i=\hat\xi_n+1}^{t}\ell_n(Y_{i,n})- \frac1t\sum_{i=1}^{t}\ell_n(Y_{i,n}).
\]
By Lemma~\ref{lem:deterministic_bound_onn_tildep}, thus we have
\[
p_{t,n} \leq \frac{m}{t}
\frac{\hat v_n}{\hat m_{t,n}\Delta_{t,\hat\xi_n}^{\,2}}
\one{\Delta_{t,\hat\xi_n}<0}+ \one{\Delta_{t,\hat\xi_n}\geq0},
\]
where we define $\hat v_n:=\frac1m\sum_{i=1}^{m}\ell_n^2(Y_{i,n})$.

We next control $\hat v_n$. By~\eqref{eq:local_variance_assumption} and the comparability of $a_n$ and $b_n$ with $I_n$, for all sufficiently large $n$ and $j\in\{0,1\}$,
\[
\E_{P_{j,n}}[\ell_n^2(X)]=\mathrm{Var}_{P_{j,n}}(\ell_n(X))+\{\E_{P_{j,n}}[\ell_n(X)]\}^2 \leq C_2I_n
\]
for some constant $C_2<\infty$. Hence, by Markov's inequality,
\[
\P\left(\hat v_n>\frac{C_2I_n}{\eta}\right)\leq\eta.
\]

It remains to establish a uniform negative-drift bound for $\Delta_{t,\hat\xi_n}$. As in Step~3 of the proof of Theorem~\ref{thm:conch_consistency_oracle}, on $\mathcal E_n$ we decompose
\[
\Delta_{t,\hat\xi_n}=A_{1,t}+A_{2,t}+A_{3,t},
\]
where $A_{1,t}$ contains observations lying sufficiently far to the left of $\xi_n'$, $A_{2,t}$ contains the boundary block between $\xi_n'-C_0/I_n$ and $\xi_n'+C_0/I_n$, and $A_{3,t}$ contains observations sufficiently far to the right of $\xi_n'$.

The argument is identical to the fixed-distribution proof, except that Lemma~\ref{lem:negative_drift_partial_sums} is replaced by~\eqref{eq:local_maximal_drift}. In particular, there exists a constant $c_\tau>0$ such that, 
\[
\P\left(\sup_{t\geq\xi_n'+M/I_n}A_{1,t}
\leq -c_\tau I_n\right)\to 1.
\]
Similarly, the second inequality in~\eqref{eq:local_maximal_drift} yields
\[
\P\left(\sup_{t\geq\xi_n'+M/I_n}A_{3,t}>0\right)
\leq B\exp\{-\gamma(M-C_0)\}.
\]

Indeed, the boundary block $A_{2,t}$ contains $O(I_n^{-1})$ terms, while the corresponding coefficients are uniformly bounded by $C I_n/(M-C_0)$. Moreover, by the comparability and variance
assumptions,
\[
\max_{j\in\{0,1\}}\left\{|\E_{P_{j,n}}[\ell_n(X)]|,
\mathrm{Var}_{P_{j,n}}(\ell_n(X))\right\}=O(I_n).
\]
Therefore, we can prove
\[
\sup_{t\ge \xi_n'+M/I_n}|A_{2,t}|
=
O_P\left(\frac{I_n}{M-C_0}\right).
\]

Therefore, by choosing $M$ sufficiently large, we obtain
\[
\limsup_{n\to\infty}\P\left(\sup_{t\geq\xi_n'+M/I_n}
\Delta_{t,\hat\xi_n}>-\frac{c_\tau}{2}I_n,\ \mathcal E_n
\right)\leq 2\eta+ B\exp\{-\gamma(M-C_0)\}.
\]
Choose $M$ large enough that $B\exp\{-\gamma(M-C_0)\}\leq\eta$ and $\frac{4C_\tau C_2}{\eta c_\tau^2(M-C_0)}\leq\alpha$.
On the event 
\[
\left\{\hat v_n\leq\frac{C_2I_n}{\eta}\right\}\cap
\left\{\Delta_{t,\hat\xi_n}\leq-\frac{c_\tau}{2}I_n\right\}\cap \mathcal E_n
\]
it holds uniformly over $t\geq\xi_n'+M/I_n$ that 
\[
p_{t,n}\leq C_\tau \frac{C_2I_n/\eta} {\{(M-C_0)/I_n\}\{c_\tau^2I_n^2/4\}}=\frac{4C_\tau C_2}{\eta c_\tau^2(M-C_0)}\leq\alpha.
\]
Consequently,
\[
\limsup_{n\to\infty} \P\left( \max_{\substack{t\in[m-1]\\t\geq\xi_n'+M/I_n}} p_{t,n}>\alpha
\right) \leq 5\eta, 
\]
as required.
This completes the first part. The second part follows by an argument similar to that in the proof of fixed-distribution result.
\end{proof}

Similar to the oracle LLR score, one may establish the asymptotic sharpness of \conch{}-SPLIT with a learned LLR score by suitably adapting the proof in Appendix~\ref{app:proof_of_consistency}. In addition to the assumptions stated above, one might need to assume a local-alternative analogue of the learnability condition for $\ell_n$. For instance, in the setting of Section~\ref{sec:sharpness_learned_llr}, where $\hat{\ell}_n$ is learned from an independent dataset, we may assume that, for each $P\in\{P_{0,n},P_{1,n}\}$,
\begin{equation*}
\label{eq:local_llr_learnability}
\frac{\left(\E_{\substack{\mathcal D_n^\prime,\,X\sim P,\\X\independent\mathcal D_n^\prime}}
\left[\bigl|\hat{\ell}_n(X)-\ell_n(X)\bigr|^2
\right]\right)^{1/2}}{I_n}\longrightarrow 0
\qquad\text{as }n\to\infty.
\end{equation*}
Since the proof would be an adaptation of the argument in Appendix~\ref{app:proof_of_consistency}, we omit these formal details.
\subsubsection{Asymptotic sharpness with \conch-MC procedure}\label{app:mc_sharpness_theory}

The asymptotic sharpness results presented so far are stated for the full-permutation \conch{} procedure. In practice, however, one would almost always employ the Monte--Carlo approximation \conch-MC (cf. \Cref{alg:conch_MC}) introduced in Section~\ref{sec:conch_mc}. It is therefore natural to ask whether the sharpness guarantees established in Section~\ref{sec:sharpness_learned_llr} and Appendix~\ref{app:consistency} continue to hold for this variant.

Recall that given any CPP score function $S:\mathcal{X}^n\to \R^{n-1}$, for a candidate changepoint $t$, the \conch-MC $p$-value is given by
\[
\hat p_{t,n}=\frac{1+\sum_{m=1}^{M_n}\One{S_t(\pi^{(m)}(\bX)) \leq S_t(\bX)}}
{M_n+1},
\]
where $\pi^{(1)},\ldots,\pi^{(M_n)}$ are drawn i.i.d.\ from the corresponding conditional permutation distribution. As the sample size $n$ increases, the size of the permutation group also grows. Consequently, the number of Monte--Carlo samples $M_n$ must scale with $n$ to ensure that the Monte--Carlo approximation remains accurate.
\begin{lemma}
For fixed $n\in \N$, it holds that as $M_n\to\infty$,
\[
\sup_{t\in[n-1]}\left|\hat p_{t,n}-p_{t,n}\right|=O_P\!\left(\sqrt{\frac{\log n}{M_n}}\right).
\]
Further, suppose $n\to\infty$ and $M_n\to \infty$ such that $\log{n}/M_n\to 0$ as $n\to\infty$. Then,
\[
\sup_{t\in[n-1]}\left|\hat p_{t,n}-p_{t,n}\right|=o_P(1).
\]
\end{lemma}

\begin{proof}
We start by defining
\[
p^\dagger_{t,n}=\frac{1}{M_n}\sum_{m=1}^{M_n}\One{S_t(\pi^{(m)}(\bX)) \leq S_t(\bX)}.
\]
Note that $\hat{p}$ and $p^\dagger$ only differs in the `+1' correction, required for finite sample validity of $\hat{p}$. Further, one can easily check that
\begin{align}\label{eq:p_dagger_and_hat}
    \left|\hat p_{t,n}-p^\dagger_{t,n}\right|
&=\left|\frac{1+\sum_{m=1}^{M_n}\One{S_t(\pi^{(m)}(\bX)) \leq S_t(\bX)}}{M_n+1}-\frac{1}{M_n}\sum_{m=1}^{M_n}\One{S_t(\pi^{(m)}(\bX)) \leq S_t(\bX)}\right|\notag\\ &\leq \frac{1}{M_n+1}.
\end{align}
Therefore, it suffices to study the concentration of $p^\dagger_{t,n}$ around the exact \conch{} $p$-value $p_{t,n}$.

Now fix a $n$ and $t\in [n-1]$. We note that $p^\dagger_{t,n}$ is an empirical average, that is
\[
p^\dagger_{t,n}=\frac{1}{M_n}\sum_{m=1}^{M_n} Z_{m,t},
\qquad Z_{m,t}:=\One{S_t(\pi^{(m)}(\bX)) \leq S_t(\bX)},
\]
where conditional on the observed data sequence $\mathcal{D}_n$, $Z_{1,t},\ldots,Z_{M_n,t}$ are i.i.d.\ Bernoulli random variables with mean $p_{t,n}$. By Hoeffding's inequality, therefore
\[
\P\left(\left|p^\dagger_{t,n}-p_{t,n}\right|>\varepsilon\;\middle|\;\mathcal D_n\right)
\leq 2\,\exp(-2M_n\varepsilon^2).
\]
Moreover, applying a union bound over $t\in[n-1]$ yields
\[
\P\left(\sup_{t\in[n-1]}\left|p^{\dagger}_{t,n}-p_{t,n}\right|>\varepsilon\;\middle|\;
\mathcal D_n\right)\leq 2\,n\exp(-2M_n\varepsilon^2).
\]
Together with~\eqref{eq:p_dagger_and_hat}, this completes the proof.
\end{proof}

The above lemma shows that the Monte--Carlo $p$-values uniformly approximate the exact \conch{} $p$-values whenever $M_n$ grows faster than $\log n$. It further implies that, as long as $M_n \gg \log n$, the asymptotic sharpness results continue to hold for the \conch-MC procedure.
\subsubsection{Auxiliary lemmas}\label{app:lemma_consistency}

\begin{lemma}\label{lem:negative_drift_partial_sums}
Suppose that $X_1,\ldots,X_k\overset{iid}{\sim}P_1$ and
$Y_1,\ldots,Y_k\overset{iid}{\sim}P_0$, and that $0<\mathrm{KL}(P_0\|P_1),\ \mathrm{KL}(P_1\|P_0)<\infty$.
Then there exist constants $c,\gamma>0$, depending only on $P_0$ and $P_1$, such that, for every $k\in\mathbb N$,
\begin{align*}
    \P\left(\frac{1}{k}\sum_{i=1}^k\ell(X_i)\geq-c\,I (P_0,P_1)\right)\leq e^{-\gamma k\,I(P_0,P_1)},~~\text{and}\\
    \P\left(\frac{1}{k}\sum_{i=1}^k\ell(Y_i)\leq c\,I(P_0,P_1)\right)\leq e^{-\gamma k\,I(P_0,P_1)}.
\end{align*}
\end{lemma}

\begin{proof}
To prove the first part, we start with observing that for any $\theta\in (0,1)$, we have that by H\"{o}lder's inequality,
\[
M(\theta):=\E_{X\sim P_1}[e^{\theta\ell(X)}]=\int f_0^\theta(x) f_1^{1-\theta}(x)\;\mathsf{d}x\le  \left(\int f_0(x)\;\mathsf{d}x\right)^{\theta} \left(\int f_1(x)\;\mathsf{d}x\right)^{1-\theta}\le 1.
\]
Moreover, we have $M(0)=M(1)=1$ and
\[
\lim_{\theta\to 0^{+}}M'(\theta)=\E_{X\sim P_1}[\ell(X)]=-\mathrm{KL}(P_1\|P_0)<0.
\]
Now, by Chernoff's bound, we have that for any $\theta \in (0,1)$,
\begin{align*}
    \P\biggl(\,\sum_{i=1}^k \ell(X_i)\ge -(k/2)\mathrm{KL}(P_1\|P_0)\biggr)&\le \exp\left(\theta(k/2)\mathrm{KL}(P_1\|P_0)\right)\,\E_{X_{1},\cdots,X_k\sim P_1}\left[\exp\bigl(\sum_{i=1}^k \theta \ell(X_i)\bigr)\right]\\
    &\le \exp \left(k\log(M(\theta))+\theta(k/2)\mathrm{KL}(P_1\|P_0)\right).
\end{align*}
We call $f(\theta)=\log(M(\theta))+(\theta/2)\,\mathrm{KL}(P_1\|P_0)$ and note that $f(0)=0$, $f(1)=\mathrm{KL}(P_1\|P_0)/2$ and $\lim_{\theta\to 0^{+}}f'(\theta)<0$, 
implying that there exists $\theta_\star\in (0,1)$ such that $f(\theta_\star)= -\gamma_1\mathrm{KL}(P_1\|P_0)$ with $\gamma_1>0$.
Consequently, it follows that 
\[
\P\biggl(\,\frac{1}{k}\sum_{i=1}^k \ell(X_i)\ge -(1/2)\,\mathrm{KL}(P_1\|P_0)\biggr)\le e^{-\gamma_1\,k\, \mathrm{KL}(P_1\|P_0))}
\]

Therefore, the first part follows by recalling $I(P_0,P_1)=\mathrm{KL}(P_0\|P_1)\wedge \mathrm{KL}(P_1\|P_0)$. The second part follows by an analogous argument with a different $\gamma_0>0$. Finally, we choose $c=1/2$ and $\gamma=\min\{\gamma_0,\gamma_1\}$ to conclude the proof.
\end{proof}
\begin{lemma}\label{lem:L2_error_ell_n}
    Fix $\delta\in (0,1)$. In the setting of Appendix~\ref{app:consistency}, suppose $0<\Gamma_0,\Gamma_1<\infty$. Then, we have that
    \begin{multline*}
        \P\,\left(\left|\sum_{s=\xi'_n+1}^{\xi'_n+k} \hat{\ell}_n(Y_s)-\ell(Y_s)\right|\le \bigl(k+\sqrt{m}\bigr)\cdot \Gamma_1^{1/2}/\delta~~\text{for all}~~k\le m-\xi'_n\right)\ge 1-2\delta,~\text{and}\\
        \P\,\left(\left|\sum_{s=1}^{k} \hat{\ell}_n(Y_s)-\ell(Y_s)\right|\le \bigl(k+\sqrt{m}\bigr)\cdot \Gamma_0^{1/2}/\delta~~\text{for all}~~k\le \xi'_n\right)\ge 1-2\delta,
    \end{multline*}
    where $Y_1,\ldots,Y_{\xi'_n}\overset{iid}{\sim} P_0$ and $Y_{\xi'_n+1},\ldots,Y_{n}\overset{iid}{\sim} P_1$. 
\end{lemma}
\begin{proof}
    We define the partial sum
\[
S^{(1)}_k:=\sum_{s=\xi'_n+1}^{\xi'_n+k} \hat{\ell}_n(Y_s)-\ell(Y_s),
\]
for $k=1,\ldots, m-\xi_n$. Note that conditional on $\mathcal{D}^{\prime}_n$, each of the summands has variance bounded by $\E_{X\sim P_1}[|\hat{\ell}_n(X)-\ell(X)|^2\mid \mathcal{D}^{\prime}_n]$. Fix any $\delta\in (0,1)$, and note that by Kolmogorov's inequality, 
\[
\P\,\left(\max_{k\le m-\xi'_n} |S_k^{(1)}-E[S_k^{(1)}\mid \mathcal{D}^{\prime}_n]|\ge \frac{\left(m\,\E_{X\sim P_1}[|\hat{\ell}_n(X)-\ell(X)|^2\mid \mathcal{D}^{\prime}_n]\right)^{1/2}}{\sqrt{\delta}}\ \middle| \ \mathcal{D}^{\prime}_n\right)\le \delta.
\]
Moreover, by the triangle inequality and the Cauchy-Schwarz inequality.
\begin{align*}
    |S_k^{(1)}|&\le |S_k^{(1)}-E[S_k^{(1)}\mid \mathcal{D}^{\prime}_n]|+k\,|\E_{X\sim P_1}[\hat{\ell}_n(X)-\ell(X) \mid \mathcal{D}^{\prime}_n]|\\
    &\le |S_k^{(1)}-E[S_k^{(1)}\mid \mathcal{D}^{\prime}_n]|+k\,\bigl(\E_{X\sim P_1}[|\hat{\ell}_n(X)-\ell(X)|^2\mid \mathcal{D}^{\prime}_n]\bigr)^{1/2}.
\end{align*}
By Markov's inequality, further with probability at least $1-\delta$,
\[
\E_{X\sim P_1}[|\hat{\ell}_n(X)-\ell(X)|^2\mid \mathcal{D}^{\prime}_n]\le \frac{\E_{X\sim P_1}[|\hat{\ell}_n(X)-\ell(X)|^2]}{\delta}=\frac{\Gamma_1}{\delta}.
\]
Therefore, in aggregate,
\[
\P\left(
\left|S_k^{(1)}\right|
\le\frac{k+\sqrt{m}}{\delta}\,\Gamma_1^{1/2}
\quad\text{for all }k\le m-\xi_n'
\right)
\ge 1-2\delta.
\]
This concludes the proof of the first part. The second part follows likewise.
\end{proof}

\begin{lemma}\label{lem:deterministic_bound_onn_tildep}
In the notation of Appendix~\ref{app:oracle_sharpness},
for $k\in [m-1],k\neq \hat{\xi}_n$, we have deterministically
\begin{equation}\label{eq:deterministic_bound_onn_tildep}
p_{k,n}\le \left(\frac{m}{\max\{k,m-k\}}\cdot\frac{\hat{v}_n}{|k-\hat{\xi}_n|\, \Delta^2_{k,\hat{\xi}_n}}\cdot \one{\Delta_{k,\hat{\xi}_n}<0}\right) + \one{\Delta_{k,\hat{\xi}_n}\ge 0},
\end{equation}
where
\[
\Delta_{k,\hat{\xi}_n}:=\frac{S_{k,n}(\bX)}{|k-\hat{\xi}_n|}
-\hat{\mu}_{k,n,L}\one{k\ge \hat{\xi}_n}
-\hat{\mu}_{k,n,R}\one{k< \hat{\xi}_n}.
\]
\end{lemma}

\begin{proof}
Fix $k\in [m-1]$ such that $k\neq \hat\xi_n$ and set $\hat{m}_n=|k-\hat{\xi}_n|$. Without loss of generality, assume $k>\hat{\xi}_n$. Then
\[
S_{k,n}(\bx)=\sum_{i=\hat{\xi}_n(\bx)+1}^{k}\hat{\ell}_{n}(x_i),
\]
a sum of $\hat{\ell}_n$ over a block of $\hat{m}_n$ observations.
Let $\mathcal{X}_L:=\{X_1,\ldots,X_{k}\}$ be the multiset of the first $k$ observations. The permutation $p$-value $p_{k,n}$ is given by
\[
p_{k,n}:=\P_{\pi\sim \mathrm{Unif}(\Pi_{k})}\!\left(S_{k,n}(\pi(\bX))\le S_{k,n}(\bX)\ \middle|\ \bX\right).
\]
Since $k \ge \hat{\xi}_n$, sampling $\pi$ uniformly from $\Pi_{k}$ is equivalent to sampling without replacement (WOR) from $\mathcal{X}_L$. Thus,
\[
p_{k,n}
=\P\!\left(\frac{1}{\hat{m}_n}\sum_{j=1}^{\hat{m}_n}\hat{\ell}_n(\tilde{X}_j)\le \frac{S_{k,n}(\bX)}{\hat{m}_n}\ \middle|\ \bX\right),
\]
where $\tilde{X}_1,\ldots,\tilde{X}_{\hat{m}_n}$ are drawn WOR from $\mathcal{X}_L$. Hence,
\[
\E\left[\frac{1}{\hat{m}_n}\sum_{j=1}^{\hat{m}_n}\hat{\ell}_n(\tilde{X}_j)\ \middle| \ \bX\right]=\hat{\mu}_{k,n,L},\quad 
\mathrm{Var}\!\left(\frac{1}{\hat{m}_n}\sum_{j=1}^{\hat{m}_n}\hat{\ell}_n(\tilde{X}_j)\ \middle| \ \bX\right)=\frac{v_{k,n}}{\hat{m}_n}\cdot \frac{k-\hat{m}_n}{k-1},
\]
where $v_{k,n}=\frac{1}{k}\sum_{i=1}^{k} {\hat{\ell}_n}^{\,2}(X_i) - (\hat{\mu}_{k,n,L})^2$. Note that
\[
v_{k,n}\le \frac{1}{k}\sum_{i=1}^{k}{\hat{\ell}_n}^{\,2}(X_i)\le \frac{m}{k}\,\hat{v}_n.
\]
Now, we can write $\tilde{p}_{k,n}$ as
\begin{multline*}
    \P\!\left(\frac{1}{\hat{m}_n}\sum_{j=1}^{\hat{m}_n}\hat{\ell}_n(\tilde{X}_j)\le \frac{S^{(n)}_{k}(\bX)}{\hat{m}_n}\ \middle|\ \bX\right)
    = \P\!\left(\frac{1}{\hat{m}_n}\sum_{j=1}^{\hat{m}_n}\hat{\ell}_n(\tilde{X}_j)-\hat{\mu}_{k,n,L}\le \frac{S^{(n)}_{k}(\bX)}{\hat{m}_n}-\hat{\mu}_{k,n,L}\ \middle|\ \bX\right).
\end{multline*}
Observe that $\Delta_{k,\hat{\xi}_n}= \frac{S^{(n)}_{k}(\bX)}{\hat{m}_n}-\hat{\mu}_{k,n,L}$, and if $\Delta_{k,\hat{\xi}_n}<0$, then by Chebyshev's inequality,
\[
p_{k,n}\le \frac{v_{k,n}}{\hat{m}_n\, \Delta^2_{k,\hat{\xi}_n}}\cdot \frac{k-\hat{m}_n}{k-1}\ \le \ \frac{m}{k}\,\frac{\hat{v}_n}{\hat{m}_n\, \Delta^2_{k,\hat{\xi}_n}}.
\]
On the other hand, if $\Delta_{k,\hat{\xi}_n}\ge 0$, then $p_{k,n}\le 1$. This proves the lemma. 
\end{proof}

\begin{lemma}\label{lem:upper_bound_on_var}
    In the setting of Section~\ref{app:consistency}, suppose $0<\sigma_\star<\infty$. Then, for any $\eta\in (0,1)$,
    \[
    \P\left(\hat{v}_n\le \frac{2(\sigma^2_\star+\mathrm{J}^2(P_0,P_1))(1+\varepsilon_{2,n}/\eta+2\sqrt{\varepsilon_{2,n}/\eta}\,)}{\eta}\right)\ge 1-4\eta,
    \]
    where we let
    \[
    \varepsilon_{2,n}=\max\left\{\frac{\Gamma_0}{\E_{X\sim P_0}[\ell^{\,2}(X)]},\, \frac{\Gamma_1}{\E_{X\sim P_1}[\ell^{\,2}(X)]}\right\}.
    \]
\end{lemma}
\begin{proof}
    Fix $\eta\in(0,1)$. We can write
    \begin{align*}
        \hat{v}_n&=\frac{\xi'_n}{m}\cdot\frac{1}{\xi'_n}\sum_{i=1}^{\xi'_n} {\hat{\ell}_n}^{\,2}(Y_i)\ +\ \frac{m-\xi'_n}{m}\cdot \frac{1}{m-\xi'_n}\sum_{i=\xi'_n+1}^{m} {\hat{\ell}_n}^{\,2}(Y_i)\\
        &\le \frac{1}{\xi'_n}\sum_{i=1}^{\xi'_n} {\hat{\ell}_n}^{\,2}(Y_i)\ +\ \frac{1}{m-\xi'_n}\sum_{i=\xi'_n+1}^{m} {\hat{\ell}_n}^{\,2}(Y_i).
    \end{align*}
    We first derive an upper bound on the first term.
    Recall that conditional on $\mathcal{D}^{\prime}_n$, $\hat{\ell}_n^{\,2}(Y_1),\ldots,\hat{\ell}_n^{\,2}(Y_{\xi'_n})$ are i.i.d.\ with mean $\E_{X\sim P_0}[\hat{\ell}_n^{\,2}(X)\mid \mathcal{D}^{\prime}_n]$.
    Moreover, by the Cauchy-Schwarz inequality, 
    \begin{align*}
        \E_{X\sim P_0}[\hat{\ell}_n^{\,2}(X)\mid \mathcal{D}^{\prime}_n]&\le\E_{X\sim P_0}\bigl[(\ell(X)+|\hat{\ell}_n(X)-\ell(X)|)^2\mid \mathcal{D}^{\prime}_n\bigr]\\
        &\le  \E_{X\sim P_0}[\ell^{\,2}(X)]+\E_{X\sim P_0}[|\hat{\ell}_n(X)-\ell(X)|^2\mid \mathcal{D}^{\prime}_n] \\&\hspace{5em}+2\E_{X\sim P_0}[|\ell(X)|\,|\hat{\ell}_n(X)-\ell(X)|\mid \mathcal{D}^{\prime}_n]\\
        &\le  \E_{X\sim P_0}[\ell^{\,2}(X)]+\E_{X\sim P_0}[|\hat{\ell}_n(X)-\ell(X)|^2\mid \mathcal{D}^{\prime}_n] \\&\hspace{5em}+2(\E_{X\sim P_0}[|\ell(X)|^2])^{1/2}\,(\E_{X\sim P_0}[|\hat{\ell}_n(X)-\ell(X)|^2\mid \mathcal{D}^{\prime}_n])^{1/2}.
    \end{align*}
    By Markov's inequality, with probability at least $1-\eta$,
    \[
    \E_{X\sim P_0}[\hat{\ell}_n^{\,2}(X)\mid \mathcal{D}^{\prime}_n]\le \E_{X\sim P_0}[\ell^{\,2}(X)](1+\varepsilon_{2,n}/\eta+2\sqrt{\varepsilon_{2,n}/\eta}\,).
    \]
    Further, we observe that
    \[
    \E_{X\sim P_0}[\ell^{\,2}(X)]=\mathrm{Var}_{X\sim P_0}(\ell(X))+(\E_{X\sim P_0}[\ell(X)])^2=\sigma^2_0+\mathrm{KL}^2(P_0\|P_1).
    \]
    By Markov's inequality, therefore
    \[
    \P\!\left(\frac{1}{\xi'_n}\sum_{i=1}^{\xi'_n} \hat\ell_n^2(Y_i)\ge \frac{\bigl(\sigma^2_0+\mathrm{KL}^2(P_0\|P_1)\bigr)(1+\varepsilon_{2,n}/\eta+2\sqrt{\varepsilon_{2,n}/\eta}\,)}{\eta} \right)\le \eta.
    \]
    Similarly, we can show that
    \[
    \P\!\left(\frac{1}{m-\xi'_n}\sum_{i=\xi'_n+1}^m \hat\ell_n^2(Y_i)\ge \frac{\bigl(\sigma^2_1+\mathrm{KL}^2(P_1\|P_0)\bigr)(1+\varepsilon_{2,n}/\eta+2\sqrt{\varepsilon_{2,n}/\eta}\,)}{\eta} \right)\le \eta.
    \]
    Since $\sigma_0,\,\sigma_1\le \sigma_\star$ and $\mathrm{KL}(P_0\|P_1),\, \mathrm{KL}(P_1\|P_0)\le \mathrm{J}(P_0,P_1)$, a union bound therefore gives us
    \[
    \P\left(\hat{v}_n\le \frac{2\bigl(\sigma^2_\star+\mathrm{J}^2(P_0,P_1)\bigr) (1+\varepsilon_{2,n}/\eta+2\sqrt{\varepsilon_{2,n}/\eta}\,)}{\eta}\right)\ge 1-4\eta,
    \]
    as required.
\end{proof}

\subsection{Proof of Theorem~\ref{thm:universality_changepoint}: Universality Theorem}\label{app:proof of universality theorem}
The proof of this universality theorem is inspired by the classical universality result for full-conformal procedures in the predictive inference framework (see \citealp[Chapter~2.4]{vovk2005algorithmic}; \citealp[Theorem~9.6]{angelopoulos2024theoretical}). 

Fix $n\in \N$. First, given the confidence set $C$, we consider the CPP score
\[
S_t(\bx) = \One{t \in C(\bx)} \in \{0,1\},
\]
for any $t \in [n-1]$. 
We will show that the \conch{} confidence set constructed from this score, denoted $\Ccal^{\conch}_{1-\alpha}$, coincides exactly with the given confidence set~$C$.

We start with showing that $\mathcal{C}^{\conch}_{1-\alpha}(\bX) \supseteq C(\bX)$; that is, if $t \in C(\bX)$, then it holds that $p_t > \alpha$, where $p_t$ is as defined in \eqref{eq:pvalue_conch}.
This is immediate by observing that if $t \in C(\bX)$, then $S_t(\bX) = 1$, and consequently,
\[
p_t = \frac{1}{|\Pi_t|}\sum_{\pi \in \Pi_t} \One{S_t(\pi(\bX)) \leq S_t(\bX)}
     = \frac{1}{|\Pi_t|}\sum_{\pi \in \Pi_t} \One{S_t(\pi(\bX)) \leq 1} = 1.
\]

Next, we show that $\mathcal{C}^{\conch}_{1-\alpha}(\bX) \subseteq C(\bX)$, i.e., if $t \notin C(\bX)$, then $p_t \leq \alpha$.
To that end, we first claim that for any $t \in [n-1]$ and any vector $\bx \in \mathcal{X}^n$,
\begin{equation}\label{eq:claim}
    \frac{1}{|\Pi_t|}\sum_{\pi \in \Pi_t} \One{t \in C(\pi(\bx))} \geq 1 - \alpha.
\end{equation}

We now prove this claim. Fix $t \in [n-1]$, and sample $\pi$ uniformly from the set of permutations $\Pi_t$. Define $\tilde{\bX} := (\tilde{X}_1, \ldots, \tilde{X}_n) := \pi(\bx)$. Conditional on the multisets $\{x_1, \ldots, x_t\}$ and $\{x_{t+1}, \ldots, x_n\}$, we have
\[
\tilde{X}_1, \ldots, \tilde{X}_t \text{ are exchangeable, and } \tilde{X}_{t+1}, \ldots, \tilde{X}_n \text{ are exchangeable.}
\]
Moreover, conditional on the multisets, $(\tilde{X}_1, \ldots, \tilde{X}_t)$ and $(\tilde{X}_{t+1}, \ldots, \tilde{X}_n)$ are independent, implying that the sampling process of $\tilde{\bX}$ satisfies \Cref{assn:exchangeability}. Consequently,
\[
\P_{\pi \sim \text{Unif}(\Pi_t)}\big(t \in C(\tilde{\bX}) \,\big|\, \{x_1, \ldots, x_t\}, \{x_{t+1}, \ldots, x_n\}\big) \geq 1 - \alpha,
\]
or equivalently, \eqref{eq:claim} holds.

Returning to the main proof, observe that if $t \notin C(\bX)$, then $S_t(\bX) = 0$. Consequently, 
\begin{align*}
p_t &= \frac{1}{|\Pi_t|}\sum_{\pi \in \Pi_t} \One{S_t(\pi(\bX)) \leq S_t(\bX)} \\
    &= \frac{1}{|\Pi_t|}\sum_{\pi \in \Pi_t} \One{S_t(\pi(\bX)) \leq 0}
     = \frac{1}{|\Pi_t|}\sum_{\pi \in \Pi_t} \One{t \notin C(\pi(\bX))} \leq \alpha,
\end{align*}
where the last step follows from \eqref{eq:claim}. This completes the proof. \hfill$\square$

\subsection{Proving asymptotic validity of \conch-SEG}
\label{app:multiple_changepoint_theory}

In this section, we establish the asymptotic validity of the extension of \conch{} to multiple changepoint localization. A direct analysis of \conch-SEG from Section~\ref{sec:mult_chpt_extension} is difficult because
segmentation and \conch{} are applied to the same observations, potentially violating the conditions needed to apply the finite-sample validity result for \conch{}. To decouple these steps, similarly to \conch-SPLIT, we introduce a closely related variant, \conch-SEG-\texttt{crossfit}, formally given in
\Cref{alg:conch-seg-crossfit}.

\IncMargin{1.2em}
\begin{algorithm}[t]
    \caption{\conch-SEG-\texttt{crossfit}}
    \label{alg:conch-seg-crossfit}
    \KwIn{$(X_t)_{t=1}^n$ (data);
    $S:\mathcal{X}^n\to\R^{n-1}$ (CPP score function);
    segmentation algorithm $\texttt{SEG}$}
    \KwOut{$\mathcal{C}^{\conch\text{-SEG-}\texttt{crossfit}}_{1-\alpha}$}
    $\mathcal{I}_1 \gets \{t\le n:\ t \text{ odd}\}$,\,
    $\mathcal{I}_2 \gets \{t\le n:\ t \text{ even}\}$\;
    $\mathcal{C} \gets \varnothing$\;
    \For{$r\in\{1,2\}$}{
        $(\hat{K}^{(r)},\hat{\xi}^{(r)}_1,\ldots,
        \hat{\xi}^{(r)}_{\hat{K}^{(r)}})
        \gets
        \texttt{SEG}\bigl((X_t)_{t\in\mathcal{I}_r}\bigr)$\;
        Compute
        $(\tilde{X}^{(r)}_0,\ldots,
        \tilde{X}^{(r)}_{\hat{K}^{(r)}})$
        by~\eqref{eq:segmentation} based on
        $(\hat{\xi}^{(r)}_1,\ldots,
        \hat{\xi}^{(r)}_{\hat{K}^{(r)}})$\;
        $J^{(r)}_\ell
        \gets
        [\,\tilde{X}^{(r)}_{\ell-1}+1,\,
        \tilde{X}^{(r)}_\ell\,]\cap\mathcal{I}_{3-r}$
        for $\ell\in[\hat{K}^{(r)}]$\;
        \label{line:segment}
        \For{$\ell\in[\hat{K}^{(r)}]$}{
            Let $X^{(r,\ell)}$ be the subsequence
            $(X_t)_{t\in J^{(r)}_\ell}$ ordered by increasing index~$t$\;
            Define
            $S^{(r,\ell)}:
            \mathcal{X}^{|J^{(r)}_\ell|}
            \to\R^{|J^{(r)}_\ell|-1}$\;
            Compute \conch{} $p$-values
            $\{p_t:
            t\in J^{(r)}_\ell\setminus
            \{\max J^{(r)}_\ell\}\}$
            as in~\eqref{eq:pvalue_conch}, using score
            $S^{(r,\ell)}$ on $X^{(r,\ell)}$\;
            $\mathcal{C}^{(r)}_\ell
            \gets
            \{\,t\in J^{(r)}_\ell
            \setminus\{\max J^{(r)}_\ell\}:p_t>\alpha\,\}$\;
            \label{line:segment_confidence_set}
            $\mathcal{C}\gets\mathcal{C}\cup
            \mathcal{C}^{(r)}_\ell$\;
        }
    }
    \Return{
    $\mathcal{C}^{\conch\text{-SEG-}\texttt{crossfit}}_{1-\alpha}
    \gets\mathcal{C}$}
\end{algorithm}
\DecMargin{1.2em}

In particular, we partition the index set into two disjoint folds
\[
\mathcal{I}_1:=\{t\in[n]:t\text{ odd}\},
\qquad\mathcal{I}_2:=\{t\in[n]:t\text{ even}\}.
\]
For $r\in\{1,2\}$, we run the segmentation algorithm on
$\mathcal{I}_r$ to obtain $\hat K^{(r)}$ and
\[
0=\hat{\xi}^{(r)}_0<\hat{\xi}^{(r)}_1<\cdots<\hat{\xi}^{(r)}_{\hat K^{(r)}}<\hat{\xi}^{(r)}_{\hat K^{(r)}+1}=n.
\]
Next, we form the segment boundaries as in~\eqref{eq:segmentation},
based on the estimated changepoints, to obtain
\[
0=\tilde X^{(r)}_0<\tilde X^{(r)}_1<\cdots<\tilde X^{(r)}_{\hat K^{(r)}-1}<\tilde X^{(r)}_{\hat K^{(r)}}=n,
\]
and define the disjoint segments
\[
J^{(r)}_\ell:=[\,\tilde X^{(r)}_{\ell-1}+1,\,
\tilde X^{(r)}_\ell\,]\cap\mathcal I_{3-r},\qquad
\ell\in[\hat K^{(r)}].
\]
We then run \conch{} on the restricted segment $J^{(r)}_\ell$ to produce a segmentwise confidence set and aggregate the resulting sets across segments and folds. In particular, when segmentation is performed on $\mathcal I_1$ (respectively, $\mathcal I_2$), \conch{} is run on $\mathcal I_2$ (respectively, $\mathcal I_1$).

Throughout this section, subsequences indexed by $\mathcal I_1$ or $\mathcal I_2$ retain their original time indices. In particular, the outputs of both $\texttt{SEG}$ and \conch{} are expressed on the original timeline $[n]$.

For the asymptotic analysis, we consider a sequence of multiple-changepoint models indexed by $n$. Fix $K>1$, and for each
$n$, let
\[
0=\xi_{0,n}<\xi_{1,n}<\cdots<\xi_{K,n}<\xi_{K+1,n}=n
\]
denote the true changepoints and the two boundary points. We assume that the observations are mutually independent and, for each $j\in[K+1]$,
\[
X_{\xi_{j-1,n}+1},\ldots,X_{\xi_{j,n}} \overset{\mathrm{i.i.d.}}{\sim}\Pcal_{j,n},
\]
where $\Pcal_{j,n}\in\mathcal M(\mathcal X)$ and $\Pcal_{j,n}\neq\Pcal_{j+1,n}$ for every $j\in[K]$. We define the minimum segment length
\[
\Delta_n:=\min_{0\le j\le K} \bigl(\xi_{j+1,n}-\xi_{j,n}\bigr),
\]
and assume that $\Delta_n\longrightarrow\infty$ as $n\to\infty$. A standard special case is $\xi_{j,n}=\lfloor n\tau_j\rfloor$ for fixed $0<\tau_1<\cdots<\tau_K<1$. We retain the dependence on $n$ in the notation for the true and estimated changepoints throughout the following theorem.

\begin{theorem}[Asymptotic coverage of\conch-SEG-\texttt{crossfit}]\label{thm:mult_changepoint}
Fix $\alpha\in(0,1)$ and consider the asymptotic
multiple-changepoint model described above. Suppose that, for each $r\in\{1,2\}$, the segmentation algorithm satisfies the following conditions.

\begin{enumerate}[label=(\alph*),ref=(\alph*)]
    \item \textbf{Consistent changepoint count.}
    \[
    \P\bigl(\hat K^{(r)}=K\bigr)\to1\qquad\text{as }n\to\infty.
    \]

    \item \textbf{Vanishing normalized localization error.}
    \[
    \max_{k\in[\hat K^{(r)}]}\min_{j\in[K]}\frac{
    |\hat\xi^{(r)}_{k,n}-\xi_{j,n}|}{\Delta_n}\xrightarrow{P}0
    \qquad\text{as }n\to\infty.
    \]

    \item \textbf{No clustering of estimated changepoints.}
    There exists $\eta>0$ such that
    \[
    \P\left(\hat K^{(r)}=K,\ \frac{
    \min_{k\in[K-1]}\bigl(\hat\xi^{(r)}_{k+1,n}
    -\hat\xi^{(r)}_{k,n}\bigr)
    }{\Delta_n}>\eta\right)\to1
    \qquad\text{as }n\to\infty.
    \]
\end{enumerate}
Then, for every $k\in[K]$, we have
\[
\P\left(\xi_{k,n}
\in\mathcal C^{\conch\text{-SEG-}\texttt{crossfit}}_{1-\alpha}
\right)\ge 1-\alpha-\mathrm{o}(1) \qquad\text{as }n\to\infty.
\]
\end{theorem}

\begin{proof}
Fix $k\in[K]$. For each $n$, let
\[
r_n
:=
\begin{cases}
2, & \text{if $\xi_{k,n}$ is odd},\\
1, & \text{if $\xi_{k,n}$ is even}.
\end{cases}
\]
Thus, the segmentation algorithm is run on $\mathcal I_{r_n}$, while
\conch{} is run on $\mathcal I_{3-r_n}$, which contains the true
changepoint $\xi_{k,n}$.

We start by defining the event
\[
\mathcal G_{k,n}:=\Bigl\{\hat K^{(r_n)}=K,\quad
\xi_{k-1,n}\le\tilde X^{(r_n)}_{k-1}\le
\xi_{k,n}-2,\quad\xi_{k,n}+2
\le\tilde X^{(r_n)}_k\le\xi_{k+1,n}\Bigr\}.
\]
On $\mathcal G_{k,n}$, the segment $J^{(r_n)}_k$ contains the true changepoint $\xi_{k,n}$, contains no other true changepoint, and also contains the next index after $\xi_{k,n}$ belonging to the same parity fold. Consequently, $\xi_{k,n}\in J^{(r_n)}_k \setminus \{\max J^{(r_n)}_k\}$.
In particular,
\[
\left\{\xi_{k,n}\in\mathcal C^{(r_n)}_k\right\}\subseteq
\left\{ \xi_{k,n}\in\mathcal C^{\conch\text{-SEG-}\texttt{crossfit}}_{1-\alpha}\right\}.
\]
Since $\mathcal G_{k,n}$ and the segment boundaries are measurable with respect to $(X_t)_{t\in\mathcal I_{r_n}}$, we further obtain by the tower law,
\begin{align}
\P\left(\xi_{k,n}\in\mathcal C^{\conch\text{-SEG-}\texttt{crossfit}}_{1-\alpha}\right)
&\ge\P\left(\xi_{k,n}\in\mathcal C^{(r_n)}_k,\mathcal G_{k,n}\right)\notag\\
&=\E\left[\One{\mathcal G_{k,n}}\P\left(\xi_{k,n}\in\mathcal C^{(r_n)}_k\;\middle|\;(X_t)_{t\in\mathcal I_{r_n}}\right)\right].
\label{eq:lotp-G}
\end{align}

Conditional on $(X_t)_{t\in\mathcal I_{r_n}}$, the segment $J^{(r_n)}_k$ is fixed. Moreover, because the observations are mutually independent, the observations in the inference fold $\mathcal I_{3-r_n}$ are independent of the observations used for segmentation. On $\mathcal G_{k,n}$, the subsequence $(X_t)_{t\in J^{(r_n)}_k}$ therefore follows a single-changepoint model with changepoint at $\xi_{k,n}$, with i.i.d.\ observations on either side of the changepoint.

Consequently, by Theorem~\ref{thm:coverage-conch},
\[
\P\left( \xi_{k,n}\in\mathcal C^{(r_n)}_k\;\middle|\; (X_t)_{t\in\mathcal I_{r_n}}\right) \ge1-\alpha\qquad\text{almost surely on }\mathcal G_{k,n}.
\]
Substituting this inequality into~\eqref{eq:lotp-G} yields
\[
\P\left(\xi_{k,n}\in\mathcal C^{\conch\text{-SEG-}\texttt{crossfit}}_{1-\alpha}\right)
\ge(1-\alpha)\P(\mathcal G_{k,n})\ge1-\alpha-\P(\mathcal G_{k,n}^{c}).
\]
It remains to show that $\P(\mathcal G_{k,n})\to1$. Define $c_\eta:=\min\{1/8,\eta/2\}$ and let
\[
\mathcal A_n:=\Bigl\{\hat K^{(r_n)}=K,\quad
\min_{j\in[K-1]}\bigl(\hat\xi^{(r_n)}_{j+1,n}-\hat\xi^{(r_n)}_{j,n}\bigr)>\eta\Delta_n,\quad\max_{j\in[K]}\min_{i\in[K]}\left|\hat\xi^{(r_n)}_{j,n}-\xi_{i,n}\right|\le c_\eta\Delta_n\Bigr\}.
\]
By the theorem hypothesis, $\P(\mathcal A_n)\to1$.
On $\mathcal A_n$, for each $j\in[K]$, choose $i_j\in[K]$ attaining
the inner minimum:
\[
\min_{i\in[K]}\left|\hat\xi^{(r_n)}_{j,n}-\xi_{i,n}
\right|=\left|\hat\xi^{(r_n)}_{j,n}-\xi_{i_j,n}\right|.
\]
We first observe that $i_{j+1}\neq i_j$ for every $j\in[K-1]$. Indeed, otherwise,
\begin{align*}
\eta\Delta_n<\hat\xi^{(r_n)}_{j+1,n}-\hat\xi^{(r_n)}_{j,n}
\le\left|\hat\xi^{(r_n)}_{j+1,n}-\xi_{i_{j+1},n}\right|+
\left|\hat\xi^{(r_n)}_{j,n}-\xi_{i_j,n}\right|
\le 2c_\eta\Delta_n \le \eta\Delta_n,
\end{align*}
which is a contradiction.

We next claim that
\[
i_1<i_2<\cdots<i_K.
\]
Suppose instead that $i_{j+1}<i_j$ for some $j\in[K-1]$. Since
distinct true changepoints are separated by at least $\Delta_n$,
\begin{align*}
\hat\xi^{(r_n)}_{j+1,n}-\hat\xi^{(r_n)}_{j,n}
&=\bigl( \hat\xi^{(r_n)}_{j+1,n}-\xi_{i_{j+1},n}\bigr)
+\bigl(\xi_{i_{j+1},n}-\xi_{i_j,n}\bigr)+\bigl(\xi_{i_j,n}-\hat\xi^{(r_n)}_{j,n}\bigr)\\
&\le\frac{\Delta_n}{8}-\Delta_n+\frac{\Delta_n}{8}=-\frac{3\Delta_n}{4},
\end{align*}
contradicting
$\hat\xi^{(r_n)}_{j+1,n}>\hat\xi^{(r_n)}_{j,n}$.
Thus, $(i_1,\ldots,i_K)$ is a strictly increasing sequence of $K$
elements of $[K]$, and hence $i_j=j$ for every $j\in[K]$. Therefore, on $\mathcal A_n$,
\begin{equation}
\label{eq:close-hat-true}
\left|\hat\xi^{(r_n)}_{j,n}-\xi_{j,n}\right|\le\frac{\Delta_n}{8},\qquad j\in[K].
\end{equation}
For notational convenience, we set $\hat\xi^{(r_n)}_{0,n}=0$ and $\hat\xi^{(r_n)}_{K+1,n}=n$. Using~\eqref{eq:close-hat-true} and the definition of $\Delta_n$, we
obtain, for every $j\in\{0,\ldots,K\}$,
\begin{equation}
\label{eq:estimated-spacing}
\hat\xi^{(r_n)}_{j+1,n}-\hat\xi^{(r_n)}_{j,n}\ge\frac{3\Delta_n}{4}.
\end{equation}
For the two boundary spacings, the stronger lower bound
$7\Delta_n/8$ holds. Recall that $\tilde X^{(r_n)}_0=0$ and $\tilde X^{(r_n)}_K=n$, and for $j\in[K-1]$,
\[
\tilde X^{(r_n)}_j=\left\lfloor\frac{\hat\xi^{(r_n)}_{j,n}+\hat\xi^{(r_n)}_{j+1,n}}{2}\right\rfloor.
\]
Therefore, by~\eqref{eq:estimated-spacing}, we obtain
\begin{align}
\tilde X^{(r_n)}_j-\hat\xi^{(r_n)}_{j,n}
\ge\frac{1}{2}\left(\hat\xi^{(r_n)}_{j+1,n}-\hat\xi^{(r_n)}_{j,n}\right)-1\ge\frac{3\Delta_n}{8}-1,\quad
\hat\xi^{(r_n)}_{j+1,n}-\tilde X^{(r_n)}_j\ge\frac{1}{2}
\left(\hat\xi^{(r_n)}_{j+1,n}-\hat\xi^{(r_n)}_{j,n}\right)
\ge\frac{3\Delta_n}{8}.
\label{eq:floor-left}
\end{align}
Combining~\eqref{eq:close-hat-true}--\eqref{eq:floor-left}, we obtain, for every $j\in[K-1]$,
\begin{align}\label{eq:midpoint-lower-upper}
\tilde X^{(r_n)}_j-\xi_{j,n}\ge\frac{\Delta_n}{4}-1,\qquad
\xi_{j+1,n}-\tilde X^{(r_n)}_j\ge\frac{\Delta_n}{4}.
\end{align}
At the two endpoints,
\[
\xi_{1,n}-\tilde X^{(r_n)}_0=\xi_{1,n}\ge\Delta_n,
\qquad\tilde X^{(r_n)}_K-\xi_{K,n}=n-\xi_{K,n}\ge\Delta_n.
\]
Since $\Delta_n\to\infty$, by \eqref{eq:midpoint-lower-upper}, for all sufficiently large $n$, on $\mathcal A_n$,
\[\xi_{k-1,n}\le\tilde X^{(r_n)}_{k-1}\le\xi_{k,n}-2,\quad \xi_{k,n}+2\le\tilde X^{(r_n)}_k\le\xi_{k+1,n}.
\]
Thus, $\mathcal A_n\subseteq\mathcal G_{k,n}$ for all sufficiently large $n$. It follows that $\P(\mathcal G_{k,n})\ge\P(\mathcal A_n)\longrightarrow1$.
Consequently,
\[
\P\left(
\xi_{k,n}\in\mathcal C^{\conch\text{-SEG-}\texttt{crossfit}}_{1-\alpha}\right)\ge1-\alpha-\P(\mathcal G_{k,n}^{c})
=1-\alpha-\mathrm{o}(1),
\]
as claimed.
\end{proof}

In particular, kernel-based changepoint detection (KCP) yields
estimators satisfying conditions~(a)--(c) under standard regularity
conditions, including consistency of the estimated number of
changepoints and localization error that is negligible relative to
the minimum segment length
\citep[see, e.g.,][]{garreau2018consistent}.
The result therefore establishes asymptotically valid marginal coverage for each changepoint in the \conch-SEG-\texttt{crossfit} confidence set when \conch{} is wrapped around KCP.

To obtain simultaneous coverage of all changepoints, one may run \conch{} on each resulting segment at level $\alpha/\hat K$. Provided that $\hat K=K$ with probability tending to one, Bonferroni's inequality then yields asymptotic simultaneous coverage at level $1-\alpha$. In our experiments, whether or not we employ cross-fitting has little effect on performance; see Appendix~\ref{app:multiple_changepoint_expt} for empirical evidence.

\section{Additional Experiments}\label{app:additional_expt}
\subsection{Simulations}
\subsubsection{Gaussian mean-shift: comparison with \citet{dandapanthula2025conformal}}\label{app:gaussian_all_pvalue_comaprison}
In the Gaussian mean-shift setting described in Section~\ref{sec:gaussian_mean_shift}, we compare our framework against the changepoint localization method of \citet{dandapanthula2025conformal}, which also constructs distribution-free confidence sets for changepoints using a matrix of conformal $p$-values. 

To ensure a fair comparison, we compare the strongest versions of these two methods. Following the recommendation of \citet{dandapanthula2025conformal}, we equip MCP with the oracle likelihood-ratio between the true pre- and post-change distributions. For \conch{}, we consider the four CPP scores from Section~\ref{sec:practical_scores}, including the oracle and learned LLR scores. Thus, both methods are supplied with the same underlying information about the changepoint model, allowing us to compare the sharpness of the resulting confidence sets more closely.

Figure~\ref{fig:gaussian_mean_shift_all_pvalue} displays the $p$-value distributions from both methods. Their approach yields a confidence set over $[362,432]$, which is  broader than the widest interval obtained by \conch{} (using the weighted-mean score). 
\begin{figure}[!h]
    \centering
    \includegraphics[width=1\linewidth]{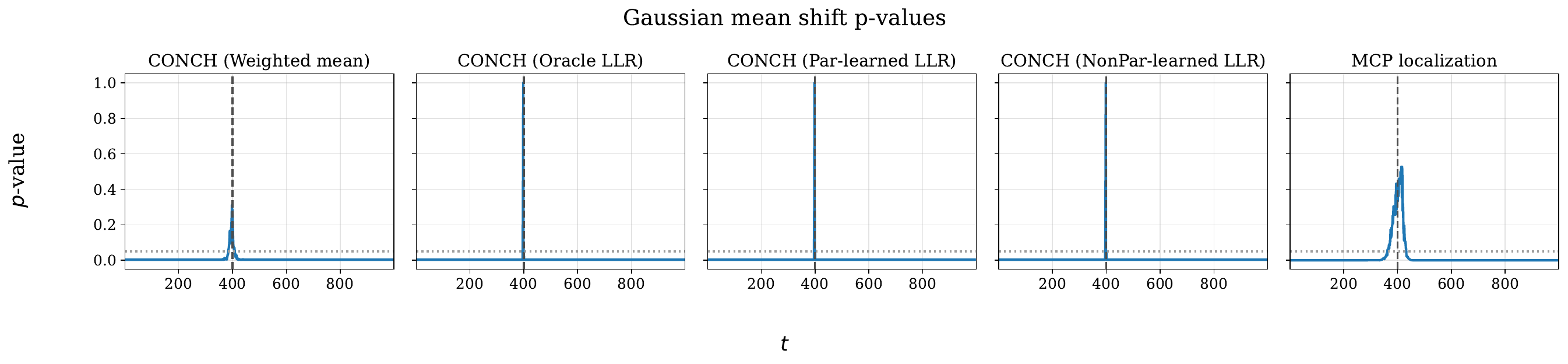}
    \caption{\conch{} and MCP $p$-value profiles under the Gaussian mean-shift model. Across all CPP scores, \conch{} produces a narrower confidence set than MCP.}
    \label{fig:gaussian_mean_shift_all_pvalue}
\end{figure}

\subsubsection{Comparison of \conch{} with existing changepoint estimation methods}
\label{app:conch_estimate_comaprison}

The primary goal of \conch{} is changepoint localization, namely constructing confidence sets that contain the true changepoint with target coverage level. While this differs from the well-studied changepoint detection problem, most existing changepoint detection procedures also produce a point estimate of the changepoint location. To facilitate comparison with such methods, we evaluate the point estimator induced by \conch{} as
\[
\hat{\xi}_{\conch}
=
\argmax_{t\in[n-1]} p_t,
\]
that is, the candidate changepoint attaining the largest \conch{} $p$-value.

We generate a sequence of $n=500$ i.i.d.\ observations with a changepoint at $\xi=200$. The pre-change distribution is $\Pcal_{0,\xi}
=\bigotimes_{t=1}^{\xi}\mathcal{N}(-\mu,1)$,
while the post-change distribution is $\Pcal_{1,\xi}=\bigotimes_{t=\xi+1}^{n}\mathcal{N}(\mu,1)$.
We consider three signal strengths,
\[
\mu=1,\quad\mu=0.5,\quad\text{and}\quad\mu=0.25,
\]
corresponding respectively to strong, moderate, and weak mean shifts. For each point estimation method, we report the localization accuracy
\[
\mathbb P\bigl(|\hat{\xi}-\xi|\le 25\bigr),
\]
estimated using $100$ independent repetitions. This metric measures the probability that the estimated changepoint lies within a neighborhood of radius $25$ around the true changepoint. Since the resulting window has length $50$, corresponding to only $10\%$ of the entire timeline, this provides a reasonably stringent measure of localization accuracy.

Under the aforementioned setting, we compare the following methods:
\begin{enumerate}
    \item \conch-MC (Algorithm~\ref{alg:conch_MC}) with the following CPP scores:
    (a) oracle log-likelihood ratio (LLR), where the pre- and post-change distributions are known exactly, and
    (b) parametrically learned LLR, where only the Gaussian family is assumed for the pre- and post-change distributions.

    \item The kernel changepoint procedure (KCP) of \citet{arlot2019kernel}.

    \item The graph-based changepoint detection procedure of \citet{chu2019asymptotic}, implemented with
    (a) the edge-count scan statistic,
    (b) the generalized edge-count scan statistic, and
    (c) the weighted edge-count scan statistic.

    \item The changeforest algorithm of \citet{londschien2023random}, which builds upon random-forest-based classification, similar to our classifier-based LLR score from Section~\ref{sec:practical_scores}.
\end{enumerate}

\begin{table}[t]
\centering
\begin{tabular}{lccc}
\toprule
Method & $\mu=1$ & $\mu=0.5$ & $\mu=0.25$ \\
\midrule
\conch{} (Oracle LLR) & \textbf{1.00} & \textbf{1.00} & \textbf{0.85} \\
\conch{} (Parametric LLR) & \textbf{1.00} & \textbf{1.00} & 0.83 \\
KCP & 0.97 & 0.47 & 0.02 \\
Graph-based (Edge Count) & \textbf{1.00} & 0.78 & 0.25 \\
Graph-based (Generalized Edge Count) & 0.99 & 0.75 & 0.18 \\
Graph-based (Weighted Edge Count) & 0.99 & 0.76 & 0.22 \\
Changeforest & \textbf{1.00} & 0.73 & 0.02 \\
\bottomrule
\end{tabular}
\caption{Localization accuracy $\mathbb{P}(|\hat{\xi}-\xi|\le 25)$ under the Gaussian mean-shift model. Larger values indicate more accurate localization, and the best result for each signal strength is shown in bold. The LLR-based \conch{} estimators remain highly accurate as the signal weakens and outperform the competing methods in the moderate- and weak-signal settings.}
\label{tab:point_estimate_comparison}
\end{table}

Table~\ref{tab:point_estimate_comparison} reports the localization accuracy of the competing methods. As expected, localization becomes more challenging as the signal strength decreases, leading to lower accuracy across all methods. Nevertheless, the point estimates induced by \conch{}, with both oracle and parametrically learned LLR scores, consistently achieve strong localization performance across all signal strengths.

Figure~\ref{fig:point_estimate_error} complements this comparison by plotting the absolute estimation error $|\hat{\xi}-\xi|$ as a function of the signal strength $\mu$ for all competing methods. As expected, the estimation error decreases as $\mu$ increases and the changepoint becomes easier to detect. Moreover, the \conch{} estimators based on oracle and parametrically learned LLR scores consistently attain the smallest errors across the range of signal strengths, demonstrating accurate and efficient localization of the changepoint.

\begin{figure}[!t]
    \centering
    \includegraphics[width=0.6\linewidth]{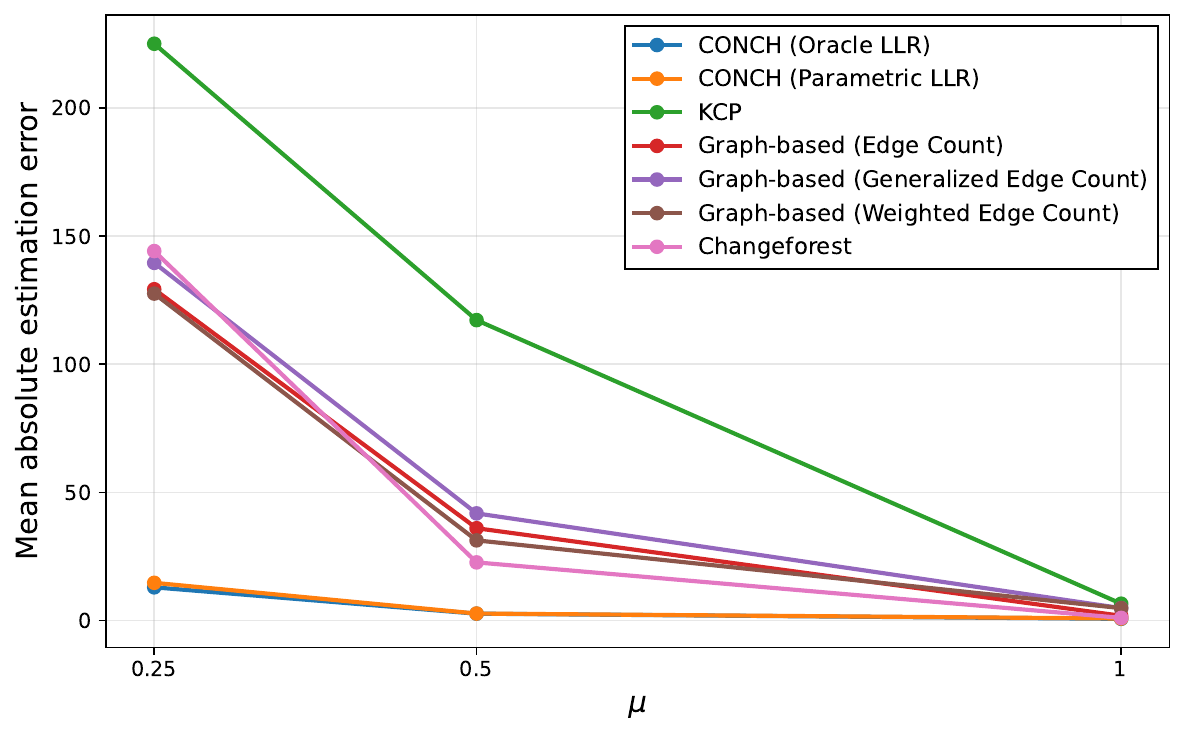}
    \caption{Mean absolute estimation error $|\hat{\xi}-\xi|$ across signal strengths for different changepoint estimators. Smaller values indicate more accurate localization. The oracle and parametrically learned LLR versions of \conch{} consistently achieve the smallest localization errors.}
    \label{fig:point_estimate_error}
\end{figure}

\subsection{Gaussian mean shift: evaluating coverage}\label{app:coverage_expt}

Now, we empirically evaluate the coverage of the \conch{}-MC algorithm (see~\Cref{alg:conch_MC}), and we empirically validate its coverage guarantee established in Theorem~\ref{thm:coverage-conch_MC}. Specifically, we generate a sequence of $n=1000$ i.i.d.\ observations with a changepoint at $\xi=400$: the pre-change distribution is $\Pcal_{0,\xi} = \bigotimes_{t=1}^{\xi}\mathcal{N}(-\mu,1)$, while the post-change distribution is $\Pcal_{1,\xi} = \bigotimes_{t=\xi+1}^{n}\mathcal{N}(\mu,1)$.

We evaluate \conch-MC using the same four CPP scores considered in Section~\ref{sec:gaussian_mean_shift}: 
(a) weighted mean difference, 
(b) oracle log-likelihood ratio (LLR), 
(c) parametrically learned LLR and 
(d) nonparametrically learned LLR.

Further, we consider three choices of $\mu$:
\[
\mu=1, 
\quad 
\mu=0.5, ~\text{and}
\quad 
\mu=0.25,
\]
respectively corresponding to a strong, moderate and weak shift in the mean.
Under each setting, we independently generate $500$ datasets and compute the empirical coverage of the resulting \conch-MC confidence sets (with $M=500$) for each choice of CPP score. 

\begin{figure}[!h]
\centering
\includegraphics[width=0.9\linewidth]{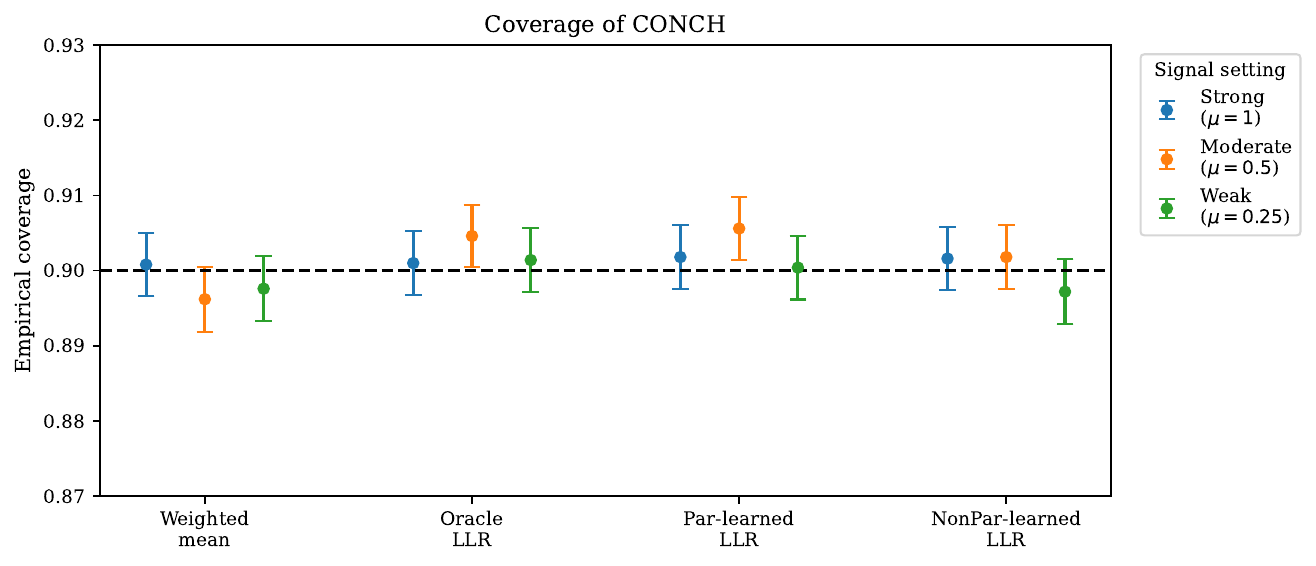}
\caption{Empirical coverage of \conch{}-MC confidence sets under the Gaussian mean-shift model. Coverage remains concentrated near the nominal level across all CPP scores and signal strengths.}

\label{fig:gaussian_coverage}
\end{figure}

Figure~\ref{fig:gaussian_coverage} reports the empirical coverage across the different score choices and signal strengths, together with standard error bars computed over the $500$ independent simulations. Across all settings, the observed coverage remains close to the nominal target level $0.9$ (shown by the horizontal dashed line), illustrating the finite-sample validity of the \conch-MC procedure. By Theorem~\ref{thm:coverage-conch_MC}, the coverage of \conch-MC is bounded below by $1-\alpha$. That said, the coverage can often be slightly conservative in finite samples, as is particularly visible for the oracle and parametrically learned LLR scores.

\subsection{\conch{}-SPLIT: effect of the split proportion on efficiency}
\label{app:conch_split_expt}

We evaluate \conch{}-SPLIT (Algorithm~\ref{alg:conch-split}) in the Gaussian mean-shift setting of Section~\ref{app:coverage_expt}. Specifically, we generate sequences of length $n=1000$ with a changepoint at $\xi=400$, where the pre-change distribution is $\Pcal_{0,\xi}=\bigotimes_{t=1}^{\xi}\mathcal{N}(-\mu,1)$ and the post-change distribution is $\Pcal_{1,\xi}=\bigotimes_{t=\xi+1}^{n}\mathcal{N}(\mu,1)$. We consider
\[
\mu=1, 
\quad 
\mu=0.5, ~\text{and}
\quad 
\mu=0.25,
\]
corresponding to strong, moderate, and weak mean shifts, respectively.

As stated in Section~\ref{sec:conch_split}, we use an interlaced full-timeline split, allocating a fraction $1-\delta$ of the observations to score learning and the remaining fraction $\delta$ to conformal calibration. We vary $\delta\in\{0.1,0.2,\ldots,0.9\}$ and examine how this split proportion affects the efficiency of \conch{}-SPLIT.

To learn the CPP score, we first obtain an initial changepoint estimate using the weighted mean-difference CUSUM statistic. This estimate induces preliminary pre- and post-change labels, which are used to train a logistic classifier and estimate the log-likelihood ratio, as described in Section~\ref{sec:practical_scores}. We then update the changepoint estimate by computing the MLE based on the corresponding estimated likelihood. The final classifier and changepoint estimate are held fixed while \conch{}-MC is applied to the calibration subsequence, after which the resulting confidence set is mapped back to the original timeline according to~\eqref{eqn:conch_split_CI}.

Figure~\ref{fig:conch_split_fraction} reports the coverage and mean relative confidence-set length $|\mathcal C^{\conch\text{-SPLIT}}_{1-\alpha}|/(n-1)$ across calibration fractions, with standard error bars. When $\delta$ is small, the calibration subsequence is sparse, and mapping the initial \conch{} confidence set back to the full timeline introduces an approximately $1/\delta$-fold expansion. Conversely, when $\delta$ is large, fewer observations remain for learning the classifier and changepoint estimate, resulting in a less informative CPP score and wider confidence sets. This deterioration for larger $\delta$ is generally less pronounced than the expansion observed for very small $\delta$. Overall, an intermediate calibration fraction provides the best tradeoff between calibration resolution and score-learning quality.

\begin{figure}[t]
\centering
\includegraphics[width=0.9\linewidth]{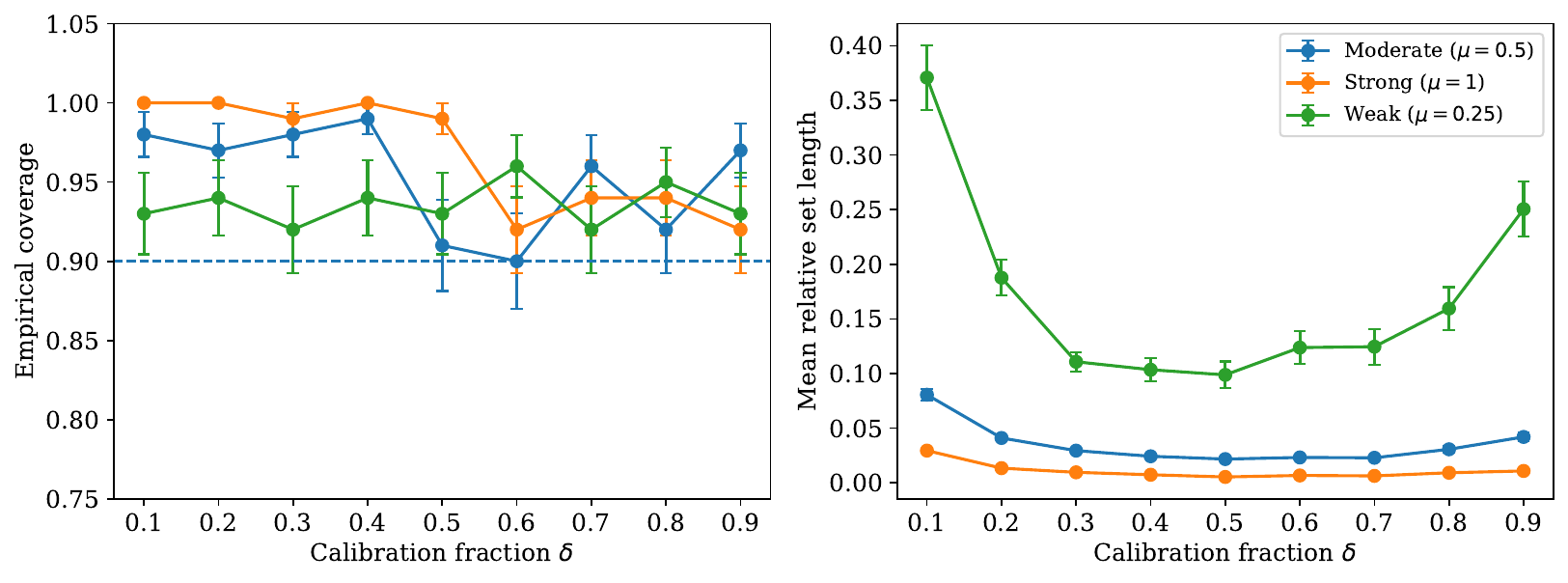}
\caption{Coverage and efficiency of \conch{}-SPLIT across calibration fractions $\delta$ under strong, moderate, and weak Gaussian mean shifts. Very small $\delta$ leads to coarse remapping onto the full timeline, whereas very large $\delta$ leaves insufficient data for score learning; intermediate fractions provide the best performance.}
\label{fig:conch_split_fraction}
\end{figure}

\subsection{Evaluating \conch{} as a calibration procedure}
\subsubsection{Calibrating resampling-based confidence sets using \conch-CAL}\label{sec:calibration_expt}
We recall that the \conch-CAL procedure (\Cref{alg:conch-calib}) can refine confidence sets that were not originally designed with distribution-free validity guarantees. In particular, the Gaussian mean-shift model has been extensively studied, and several bootstrap-based methods provide asymptotically valid intervals that perform well in practice. However, under even mild model misspecification, these intervals can become overly wide or may miss the true changepoint~$\xi$. In the following experiment, we consider two settings with $n=500$ observations, and a changepoint at $\xi=200$:
\begin{enumerate}[label=(\roman*), ref=(\roman*), noitemsep]
    \item Gaussian mean-shift model: $\Pcal_{0,\xi} = \bigotimes_{t=1}^{\xi}\mathcal{N}(-1,3)$ and $\Pcal_{1,\xi} = \bigotimes_{t=\xi+1}^{n}\mathcal{N}(1,3)$
    \item Laplace mean-shift model: $\Pcal_{0,\xi} = \bigotimes_{t=1}^{\xi}\mathrm{Laplace}(-1,3)$ and $\Pcal_{1,\xi} = \bigotimes_{t=\xi+1}^{n}\mathrm{Laplace}(1,3)$.
\end{enumerate}

In both settings, we compute bootstrap confidence intervals and evaluate how the \conch-CAL procedure refines them to produce distribution-free confidence sets. We employ a residual bootstrap scheme to construct the initial confidence sets: for each replicate, the changepoint is re-estimated on a resampled sequence formed from centered residuals, producing an empirical distribution of $\hat{\tau}$ from which the bootstrap $p$-values are obtained. As a baseline, we also apply \conch{} directly using the unnormalized bootstrap $p$-values as the CPP score. This provides a naive distribution-free calibration against which the normalization used by \conch-CAL can be compared.

\begin{figure}[!h]
    \centering
    \includegraphics[width=0.73\linewidth]{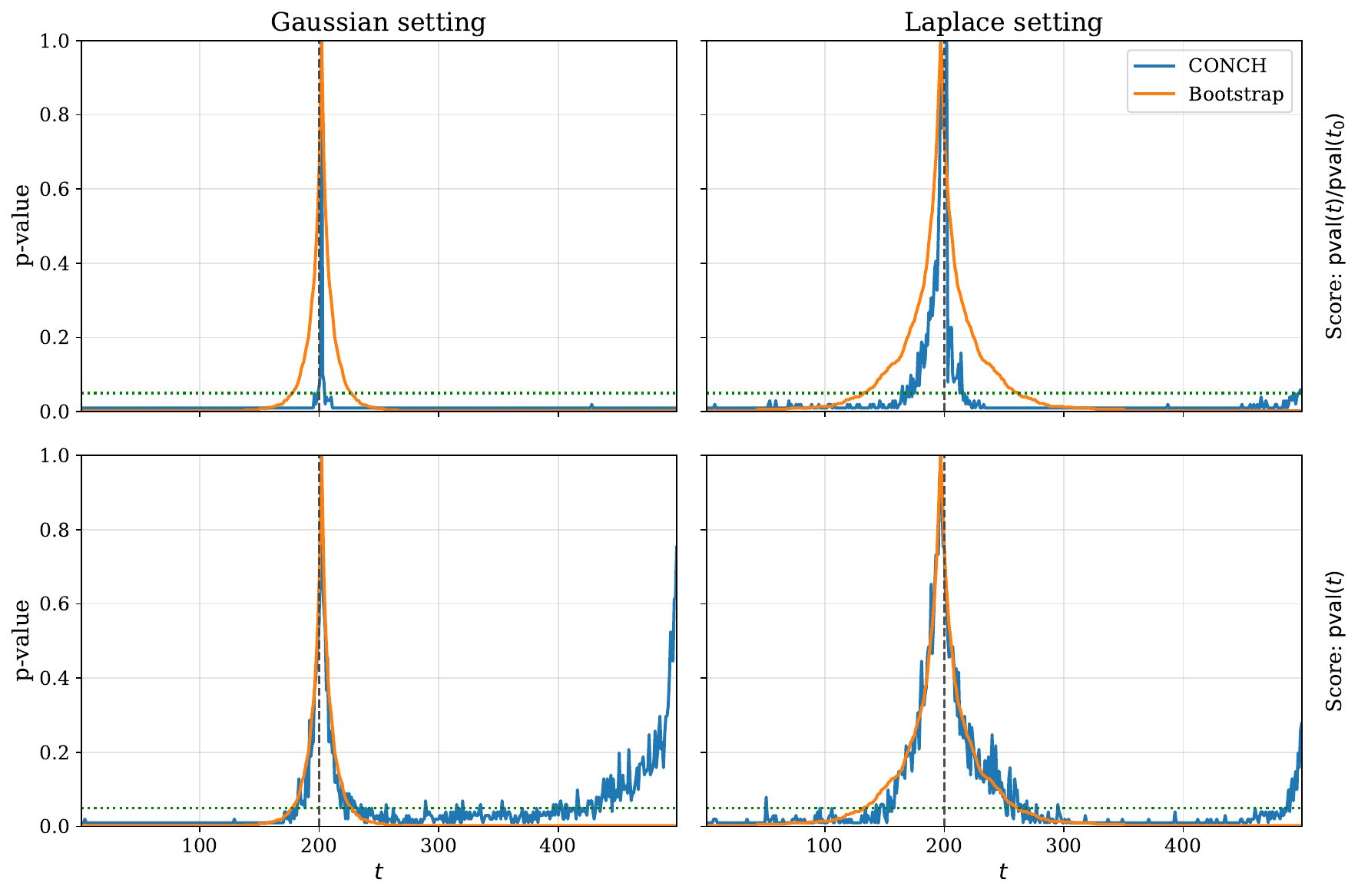}
    \caption{Calibration of bootstrap-based confidence sets under Gaussian and Laplace mean shifts. The top row uses the normalized \conch{}-CAL score $\mathrm{pval}(t)/\mathrm{pval}(t_0)$, while the bottom row uses the unnormalized score $\mathrm{pval}(t)$. Normalization yields substantially sharper confidence sets, even under model misspecification.}

    \label{fig:conch_calibration}
\end{figure}

Under the correctly specified Gaussian model, the bootstrap interval $[180,224]$ is refined by \conch-CAL to produce the narrower interval $[197,205]$. Under model misspecification, in the Laplace setting, the bootstrap interval $[140,258]$ is considerably inflated by the heavy-tailed noise. After calibration, it is reduced to $[196,202]$, enabling more precise localization. The bootstrap $p$-values are generally diffuse (see Figure~\ref{fig:conch_calibration} top row for a visualization), whereas the \conch-CAL $p$-values remain sharply concentrated near the true changepoint, highlighting the robustness of the procedure across distributional regimes.

Applying \conch{} directly with the bootstrap $p$-values as CPP scores also produces a valid distribution-free confidence set, but yields substantially wider sets in both settings, as shown in the bottom row of Figure~\ref{fig:conch_calibration}. This empirically supports the intuition in Section~\ref{sec:calibration}: the normalization in the CPP score in \eqref{score:p-value_score} is not essential for validity, but is important for retaining the relative localization information encoded by the original procedure and obtaining sharper confidence sets.

\subsubsection{Calibrating point estimates into confidence sets}\label{app:conch_around_changepoint}

Next, we empirically investigate the procedure for calibrating a point estimate from Section~\ref{sec:point_estimate_calibration}. In particular, we take the changepoint estimate produced by the kernel changepoint procedure (KCP) of \citet{arlot2019kernel}, and we consider the Gaussian mean-shift setting from Section~\ref{sec:gaussian_mean_shift} with $n=200$ observations and a changepoint at $\xi=80$, where the pre-change and post-change samples are generated i.i.d.\ from $\mathcal{N}(-1.5,1)$ and $\mathcal{N}(1.5,1)$, respectively. The KCP estimate for the observed data sequence is given by $\hat{\xi}_0=79$.

We apply \conch{} with the recommended scores from Section~\ref{sec:point_estimate_calibration}: the weighted mean-difference score and the learned LLR score. For the latter, the log-likelihood ratio is either estimated parametrically under a Gaussian model or approximated using classifier logits obtained from a random forest classifier with maximum depth $4$. Figure~\ref{fig:KCP_conch_single} displays the resulting \conch{} $p$-values. For all the scores, the largest $p$-values occur in a small neighborhood of the true changepoint and the KCP estimate, while the $p$-values away from this region are typically close to zero. Consequently, the resulting confidence sets provide a meaningful quantification of the uncertainty associated with the initial point estimate.

The parametrically learned LLR score produces a noticeably sharper $p$-value profile than the weighted mean-difference score, consistent with our other experimental findings in Section~\ref{sec:gaussian_mean_shift}. Overall, the experiment demonstrates that \conch{} can be used to wrap an existing changepoint estimator and transform an almost accurate point estimate into a finite-sample valid confidence set for the true changepoint.

\begin{figure}[!h]
\centering
\includegraphics[width=0.84\linewidth]{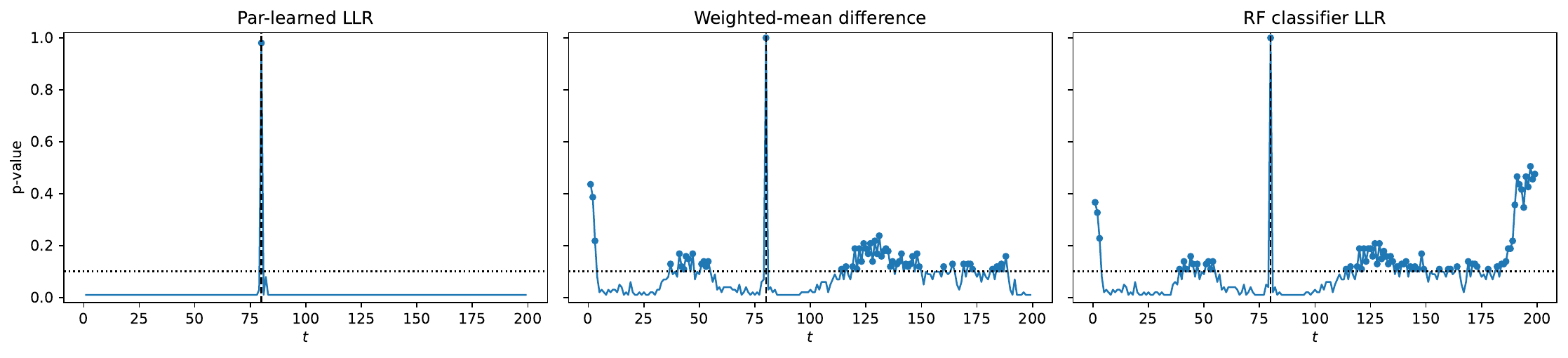}
\caption{\conch{} $p$-value profiles obtained by calibrating the KCP estimate $\hat{\xi}_0=79$ using three CPP scores. The $p$-values concentrate near the true changepoint $\xi=80$, with the parametrically learned LLR score producing the sharpest localization. }
\label{fig:KCP_conch_single}
\end{figure}

\subsection{Empirical evaluation of \conch{} extensions}
\subsubsection{\conch-SEG: localization of multiple changepoints}
\label{app:multiple_changepoint_expt}
We consider a multiple-changepoint Gaussian mean-shift model to illustrate the performance of \conch-SEG (\Cref{alg:conch-seg}). In particular, we generate $n=1500$ observations with true changepoints at $\xi_1=150$, $\xi_2=500$, $\xi_3=820$, and $\xi_4=1100$, segment means $\mu_1=-1$, $\mu_2=0.5$, $\mu_3=1.5$, $\mu_4=-2$, and $\mu_5=-1$, and common variance $\sigma^2=1$, i.e.,
\[
X_t \sim \mathcal{N}(\mu_j,\sigma^2)\quad \text{for } t\in(\xi_{j-1},\xi_j], 
\qquad \text{and}~~\xi_0=0,\ \xi_5=n.
\]

Initial changepoint estimates are obtained via KCP \citep{garreau2018consistent} with a Gaussian kernel, yielding the sequence $(150, 497, 820, 1091)$. As expected, when the pre- and post-change distributions around a changepoint are more distinct, KCP detects the changepoint more accurately, whereas the estimate is offset by a small margin when the adjacent distributions are more similar.

We then apply \conch-SEG, wrapping the \conch{} framework around KCP and using the parametric CPP score specialized to the Gaussian family with known variance (see~\eqref{score:learned_llr}). The resulting confidence set is
\[
[145,152]\ \cup\ [483,515]\ \cup\ [820,821]\ \cup\ [1084,1103].
\]
Figure~\ref{fig:KCP_conch} displays the observed sequence (left) and the corresponding \conch{} $p$-values (right). The procedure sharply localizes all the changepoints. The sets around $\xi_2$ and $\xi_4$ are wider than the sets around $\xi_1$ and $\xi_3$ because the distributions on either side are more similar.

\begin{figure}[!h]
    \centering
    \includegraphics[width=0.9\linewidth]{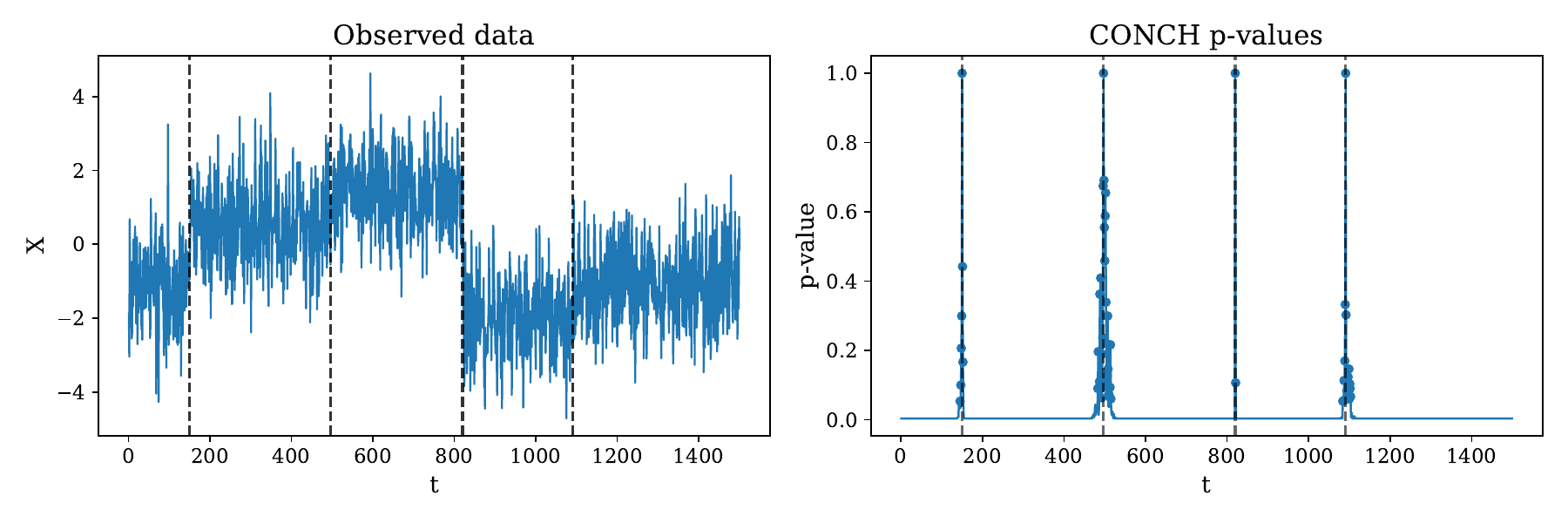}
    \caption{Multiple-changepoint localization using \conch{}-SEG. Left: observed Gaussian mean-shift sequence. Right: \conch{} $p$-values obtained by wrapping KCP with the parametric LLR score. All four changepoints are sharply localized, with wider confidence sets for the more difficult-to-detect changes.}
    \label{fig:KCP_conch}
\end{figure}

\subsubsection{\conch-DEP: changepoint localization under temporal dependence}
\label{app:dependent_expt}

We now investigate the empirical performance of \conch-DEP, introduced in Appendix~\ref{app:conch-dep}, on temporally dependent data. Throughout, we consider a sequence of length $n=500$ with a changepoint at $\xi=200$ and compute the candidate-wise $p$-values $p_t^{(m)}$ for all $t\in[n-1]$ using the \conch-DEP procedure.

We consider two forms of temporal dependence. The first corresponds to the idealized setting of $m$-dependence covered by Theorem~\ref{thm:conch_dep}. Specifically, on either side of the changepoint, observations are generated from independent stationary moving-average processes of order $m$. For $k\in\{0,1\}$, let
\[
Y_i^{(k)}=\mu_k+\varepsilon_i^{(k)}+
\sum_{j=1}^{m}
\theta_j \varepsilon_{i-j}^{(k)}, \qquad
\varepsilon_i^{(k)}
\stackrel{\mathrm{iid}}{\sim}
N(0,\sigma^2).
\]
The observed sequence is then given by
\[
X_i=\begin{cases}
Y_i^{(0)}, & i\le \xi,\\
Y_{i-\xi}^{(1)}, & i>\xi,
\end{cases}
\]
so that the process is stationary and $m$-dependent on either side of the changepoint. In this setting, we can apply \conch-DEP with the correct spacing parameter $m$.

Our second experimental setting considers a dependent model that falls outside the scope of the finite-sample validity result in Theorem~\ref{thm:conch_dep}. Specifically, we generate an autoregressive process of order one, by computing 
\[
X_i=\begin{cases}
\mu_0+\phi\cdot (X_{i-1}-\mu_0)+\varepsilon_i,
&i\le \xi,\\
\mu_1+\phi \cdot (X_{i-1}-\mu_1)+\varepsilon_i,
&i>\xi,
\end{cases}
\]
where $\varepsilon_i\stackrel{\mathrm{iid}}{\sim}N(0,\sigma^2)$ and $\phi=0.5$. Here, the dependence weakens geometrically with the lag, so observations that are sufficiently far apart are approximately independent. However, the process does not satisfy the stationarity or $m$-dependence assumptions exactly. Nevertheless, as discussed in Appendix~\ref{app:conch-dep}, we apply \conch-DEP with spacing parameter $m=3$.

For both dependence structures, we evaluate the performance of \conch-DEP with two CPP scores: the oracle likelihood-ratio score, which assumes complete knowledge of the pre- and post-change Gaussian models, and a parametric learned likelihood-ratio score, which estimates the Gaussian parameters from the observations on either side of the candidate changepoint.

Figure~\ref{fig:conch_dep_profiles} displays the resulting $p$-values for both dependence structures and both choices of CPP score. In the $m$-dependent setting, the oracle and parametrically learned LLR scores produce confidence sets that are constrained to the ranges $[169,227]$ and $[180,215]$, respectively. In the AR(1) setting, the corresponding ranges are $[160,220]$ and $[180,220]$. Despite the presence of temporal dependence, all four confidence sets remain tightly localized around the true changepoint.

\begin{figure}[t]
    \centering
    \includegraphics[width=0.49\linewidth]{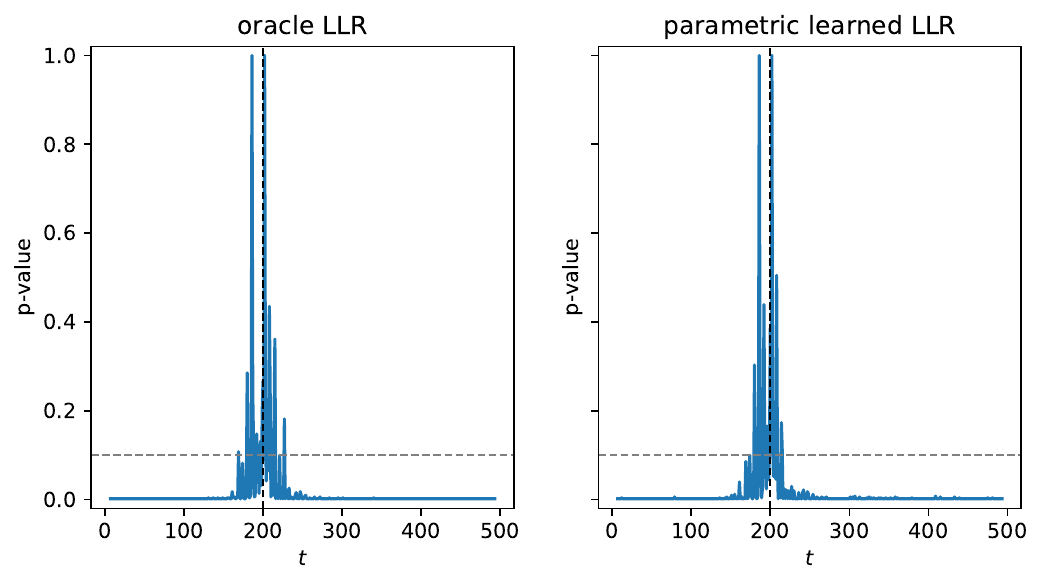}
    \hfill
    \includegraphics[width=0.49\linewidth]{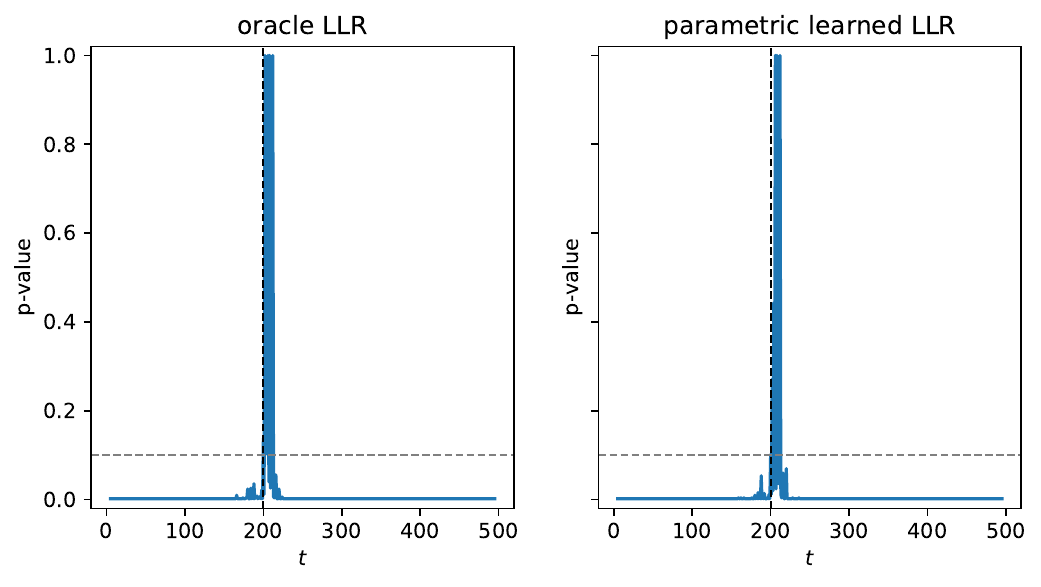}
    \caption{\conch{}-DEP $p$-value profiles under temporal dependence. The left two panels correspond to the stationary $m$-dependent setting, and the right two to the AR(1) setting. The confidence sets remain tightly localized despite temporal dependence.}
    \label{fig:conch_dep_profiles}
\end{figure}
\subsection{Two urns model: effect of dissimilarity between $\Pcal_{0,\xi}$ and $\Pcal_{1,\xi}$ on confidence set length}
While we have demonstrated the performance of \conch{} on a variety of changepoint detection tasks, our experiments so far have focused on i.i.d.\ settings, that is, changepoint models within $\Pcal_{\textnormal{IID}}$. 
In what follows, we go beyond the i.i.d.\ assumption and show that the \conch{} framework requires only exchangeability to produce valid confidence sets. 

To illustrate this, we evaluate the performance of the \conch{} confidence sets on a two-urn model with finite populations. 
Specifically, we consider two urns, each containing $2500$ balls colored either red or blue. 
The proportions of red balls in the first and second urns are $0.5 - \delta$ and $0.5 + \delta$, respectively, for some $\delta \in (0,0.5]$. 
We draw balls without replacement: the first $\xi = 350$ draws come from urn~1, and the remaining from urn~2, yielding a total of $n = 800$ observations. 
Our goal is to detect the changepoint $\xi$. 
We use the weighted mean difference as the CPP score, and for each $\delta \in \{0.05, 0.10, \ldots, 0.50\}$, we run the \conch-MC algorithm (\Cref{alg:conch_MC}) with $M = 300$ permutations to obtain confidence sets.

When $\delta$ is small, the pre-change and post-change distributions are nearly indistinguishable. Consequently, no method can sharply localize the changepoint, including \conch{} confidence sets. As $\delta$ increases, the two distributions become more distinct. In the extreme case $\delta=0.5$, the first urn contains only blue balls and the second only red balls, allowing perfect localization with absolute confidence. Accordingly, the length of the \conch{} confidence sets, averaged over $15$ independent runs, decreases with $\delta$, as shown in the right panel of Figure~\ref{fig:two_urn_expt}, where the shaded region denotes one standard error around the mean.
 Across the whole collection of $\delta$ values, the true change-point $\xi=350$ lies within the reported confidence set, demonstrating the validity of our procedure (left panel of Figure~\ref{fig:two_urn_expt}).
\begin{figure}[!h]
    \centering
    \includegraphics[width=0.85\linewidth]{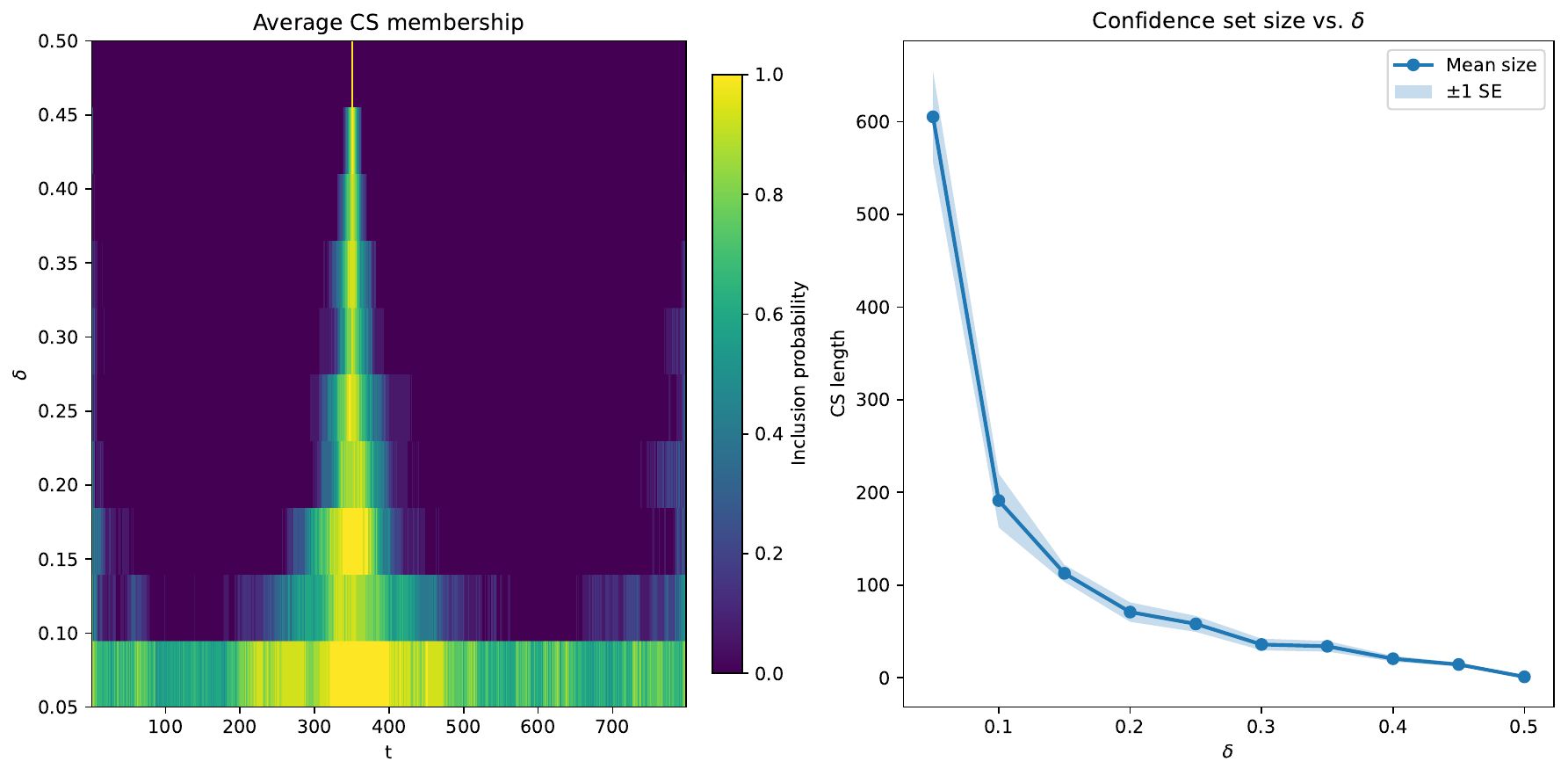}
    \caption{Two-urn changepoint experiment across dissimilarity levels $\delta$. Left: empirical confidence-set membership probabilities across candidate indices. Right: mean confidence-set length, with the shaded region denoting one standard error. The true changepoint remains covered, while the confidence set contracts as the two urn distributions become more distinguishable.}
    \label{fig:two_urn_expt}
\end{figure}

\subsection{Real data experiments}
\subsubsection{MNIST: detect change in digits}
We conduct a simulation based on the MNIST handwritten digits dataset \citep{deng2012mnist} to evaluate the performance of \conch{} for a digit shift localization. In particular, suppose we observe a sequence of $1,000$ images: the first $\xi=400$ observations consist of i.i.d. samples of the digit ``$1$'', and the latter observations are i.i.d. samples of the digit ``$7$'' (Figure \ref{fig:mnist-sample}).
\begin{figure}[!h]
    \centering
    \includegraphics[width=0.8\linewidth]{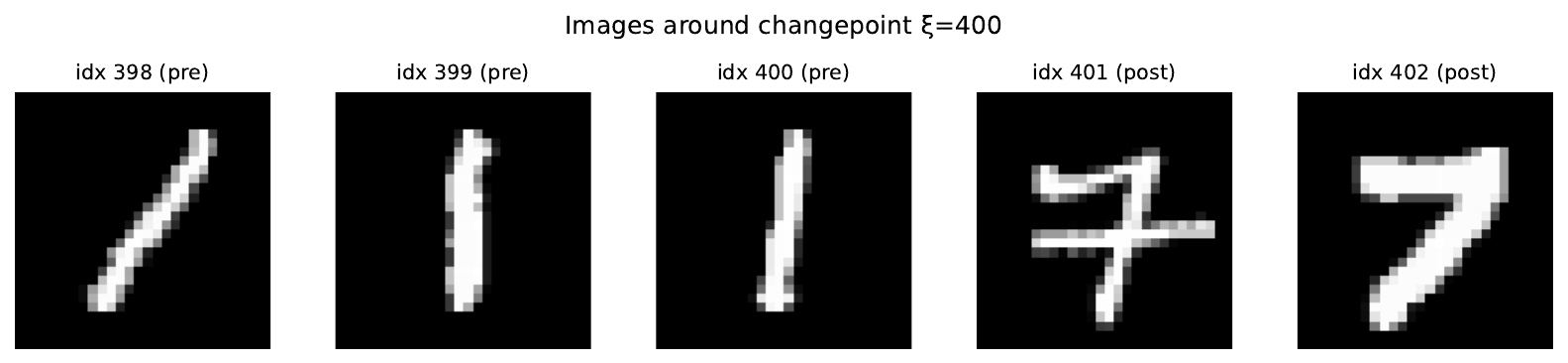}
    \caption{MNIST changepoint setup with a shift from digit ``1'' to digit ``7'' at $\xi=400$ $(n=1000)$. The visual similarity between the two digit classes makes localization nontrivial.}
    \label{fig:mnist-sample}
\end{figure}

As in our main experiments, we use a classifier based log-likelihood ratio as CPP score in our \conch{} algorithm. Specifically, we first train a convolutional neural network classifier on an independent labeled dataset to distinguish between the two digits; its logits then define the CPP score, which is then passed to \conch{} to produce a confidence interval for the changepoint. 

Although the handwritten digits ``$1$'' and ``$7$'' often exhibit substantial visual similarity, our approach accurately detects the changepoint, yielding an exceptionally narrow, in fact singleton confidence set $\{400\}$ (Figure~\ref{fig:mnist_pvalue}). We remark that the sharp localization here is partially a consequence of the strong classifier, which can confidently distinguish between the two digits. In the next section, we investigate how classifier strength influences the width of \conch{} confidence sets on the CIFAR-100 dataset.

\begin{figure}[!h]
    \centering
    \includegraphics[width=0.6\linewidth]{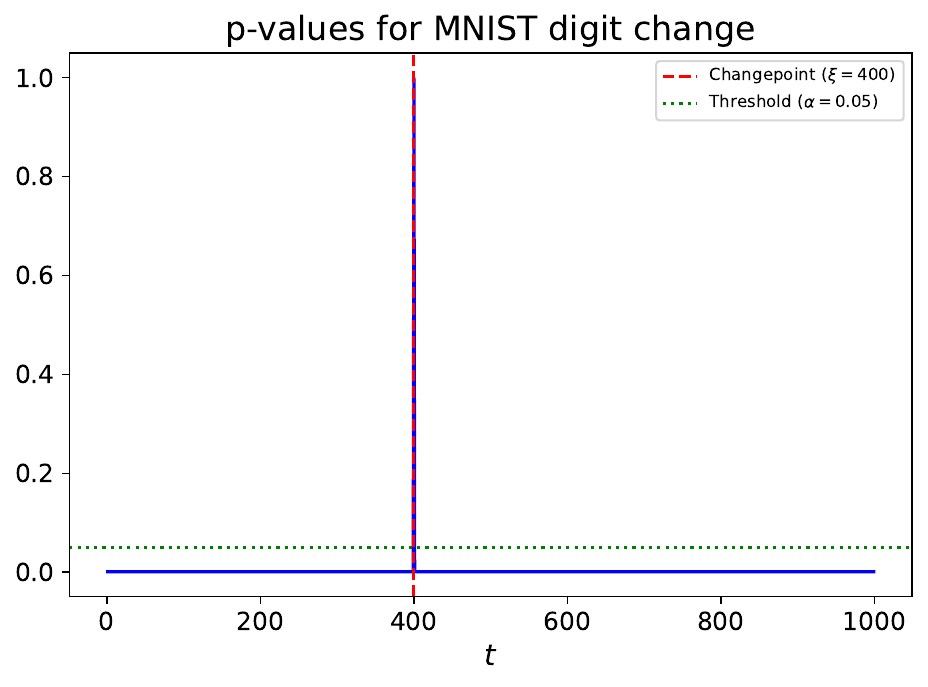}
    \caption{\conch{} $p$-values for the MNIST shift from digit ``1'' to digit ``7'' at $\xi=400$. The dotted line marks $\alpha=0.05$. The classifier-based CPP score exactly localizes the changepoint, yielding the singleton confidence set $\{400\}$.}

    \label{fig:mnist_pvalue}
\end{figure}

\subsubsection{CIFAR100: classifier strength affects power of \conch{}}\label{app:cifar100_expt}
We simulate a class-shift scenario using the CIFAR-100 image dataset \citep{krizhevsky2009learning} to evaluate \conch{} under a challenging setting. Specifically, we construct a sequence of $n=1,000$ observations with a changepoint at $\xi=400$: the pre-change distribution $\P_{0,\xi}$ consists of i.i.d. images of bears, while the post-change distribution $\P_{1,\xi}$ consists of i.i.d. images of beavers (Figure~\ref{fig:cifar100-sample}). Because bears and beavers share many visual attributes, accurately localizing the changepoint is a non-trivial task.

\begin{figure}[!h]
    \centering
    \includegraphics[width=0.85\linewidth]{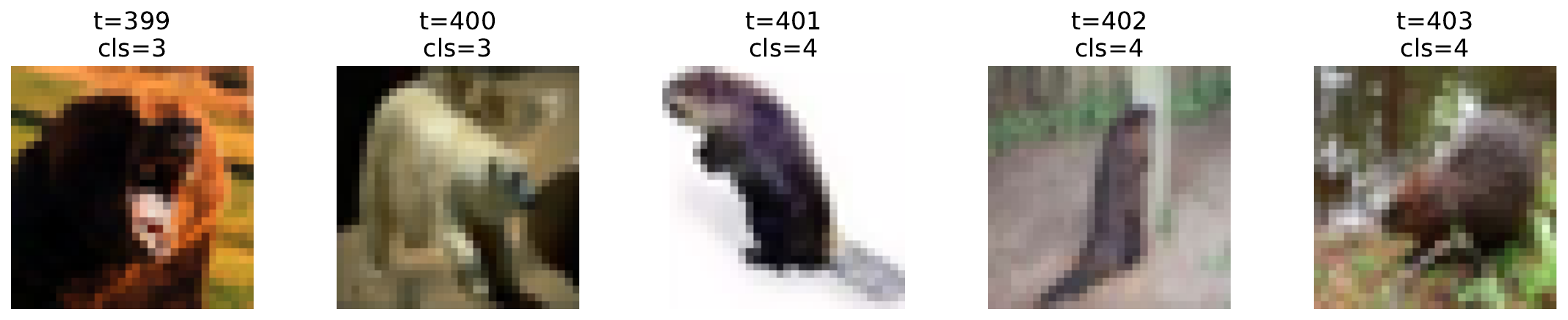}
    \caption{CIFAR-100 changepoint setup with a shift from bear images to beaver images at $\xi=400$ $(n=1000)$. The substantial visual similarity between bears and beavers creates a challenging localization problem.}
    \label{fig:cifar100-sample}
\end{figure}

We pre-train a small three-block convolutional network with a lightweight classification head on an independent labeled dataset. We first train this network for $5$ epochs to obtain a weak classifier and then train it further for an additional $20$ epochs to obtain a stronger classifier. The resulting logits from each model define a CPP score, which we pass to \conch{} to produce a changepoint confidence interval.

Figure~\ref{fig:cifar_pvalue} reports the $p$-value distributions and confidence sets produced by \conch{}. As anticipated, the stronger classifier yields sharper separation between the two classes, leading to a much narrower confidence set 
$[398,405]\cup\{408,415,419\}$ compared to the weaker model’s wider interval 
$[387,434]$. This experiment highlights both the sensitivity of \conch{} to classifier quality and its ability to localize changepoints even under subtle visual differences between classes.

\begin{figure}[!h]
    \centering
    \includegraphics[width=0.85\linewidth]{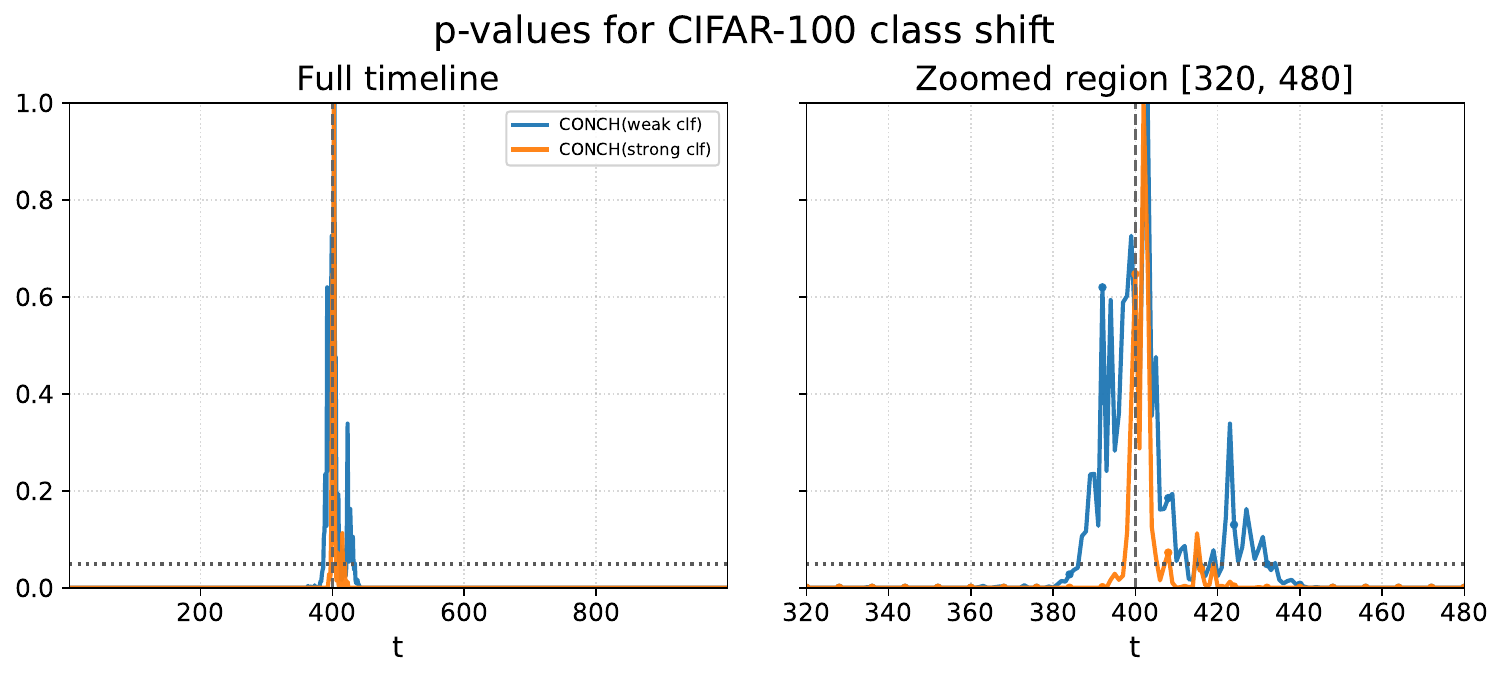}
    \caption{\conch{} $p$-value profiles for the CIFAR-100 bear-to-beaver shift using weak and strong classifiers. Left: full timeline. Right: enlarged region around the true changepoint $\xi=400$. A stronger classifier produces a substantially sharper confidence set, illustrating the effect of classifier quality on localization power.}
    \label{fig:cifar_pvalue}
\end{figure}

\clearpage

\end{document}